\documentclass[10pt]{amsart}

\usepackage{amsmath}
\usepackage{amssymb}
\usepackage{epic}
\usepackage[final]{graphicx}

\newtheorem{thm}{Theorem}
\newtheorem{lem}[thm]{Lemma}
\newtheorem{prop}[thm]{Proposition}
\newtheorem{cor}[thm]{Corollary}

\theoremstyle{definition}

\theoremstyle{remark}
\newtheorem{remark}{Remark}[thm] 

\theoremstyle{plain}

\numberwithin{equation}{section}

\def\CC{{\mathbb C}}

\def\HH{{\mathbb H}}
\def\NN{{\mathbb N}}

\def\RR{{\mathbb R}}
\def\SS{{\mathbb S}}

\def\ZZ{{\mathbb Z}}

\def\e{\mathrm{e}}
\def\i{\mathrm{i}}

\def\id{\operatorname{id}}

\def\D{\operatorname{D{}}}
\def\E{\operatorname{E{}}}

\def\G{\operatorname{G{}}}

\def\PSL{\operatorname{PSL}}

\def\Area{\operatorname{Area}}
\def\ord{\operatorname{ord}}

\def\scrA{{\mathcal A}}

\def\scrF{{\mathcal F}}

\def\scrI{{\mathcal I}}

\def\scrL{{\mathcal L}}

\def\Re{\operatorname{Re}}
\def\Im{\operatorname{Im}}

\newcommand{\be}{\begin{equation}}
\newcommand{\ee}{\end{equation}}
\newcommand{\ba}{\begin{align}}
\newcommand{\ea}{\end{align}}

\begin{document}

\begin{center}
\title[The trace formula for singular perturbations]{The trace formula for singular perturbations \\ of the Laplacian on hyperbolic surfaces}
\author{Henrik Uebersch\"ar}
\address{School of Mathematics, University of Bristol, Bristol BS8 1TW, U.K.}
\email{h.ueberschaer@bris.ac.uk}
\date{5 May 2009. The author is supported by an EPSRC doctoral training grant.}
\maketitle
\end{center}

\begin{abstract}
We prove an analogue of Selberg's trace formula for a delta potential on a hyperbolic surface of finite volume. For simplicity we restrict ourselves to surfaces with at most one cusp, but our methods can easily be extended to any number of cusps. In the case of a noncompact surface we derive perturbative analogues of Maass cusp forms, residual Maass forms and nonholomorphic Eisenstein series. The latter satisfy a functional equation as in the classical case. We also introduce a perturbative analogue of Selberg's zeta function and apply the trace formula to prove its meromorphic continuation to the complex plane as well as a functional equation. 
\end{abstract}

\tableofcontents

\section{Introduction}

Selberg's trace formula \cite{S} is a central tool in the spectral theory of automorphic forms \cite{Hj2}, \cite{Hj3}, \cite{Iw} and has important applications in number theory \cite{Hj1}. It also plays a central role in the theory of quantum chaos as an exact analogue of Gutzwiller's celebrated trace formula  \cite{Bz}, \cite{CRR}, \cite{G}, \cite{PU} which links the distribution of the energy levels of a classically chaotic quantum system to the actions of periodic orbits. Selberg proved his trace formula in 1956 for the Laplacian on weakly symmetric Riemannian spaces. The originial goal was to prove the existence of cusp forms for the modular surface. These form the discrete spectrum of the Laplacian and play a crucial role in Selberg's proof of the trace formula. Selberg also introduced a zeta function associated with the eigenvalue problem of the Laplacian and showed that it had remarkably similar properties to Dirichlet series studied in number theory. He used the trace formula to show that his zeta function satisfied an analogue of the Riemann hypothesis.

Among number theorists singular perturbations of the Laplacian attracted attention in the early 80s when D. Hejhal \cite{Hj4} gave the mathematical answer to a problem of Haas \cite{Ha}. In 1977 Haas had discovered the nontrivial zeros of the Riemann zeta function in numerical investigations of the spectrum of the Laplacian on the modular surface. In 1979 Hejhal demystified this phenomenon when he showed \cite{Hj4} that the apparent eigenfunctions on the modular surface that Haas had discovered featured a logarithmic singularity at the conic singularity of order 3 which numerically was hard to detect, and thus failed to be eigenfunctions of the Laplacian. The full mathematical explanation of these pseudo-eigenfunctions was given by Colin de Verdiere \cite{CdV} in 1983 by revealing them to be eigenfunctions of rank one perturbations of the Laplacian by a delta potential at the respective conic singularity.

A semi-classical trace formula for a delta potential was proven by Hillairet \cite{Hi} in the case of a 3-dimensional Riemannian surface. In this paper we prove an exact trace formula for rank one perturbations by a delta potential on a hyperbolic Riemann surface with at most one cusp. In particular our trace formula expresses the difference between perturbed and unperturbed trace in terms of diffractive orbit terms and, in the noncompact case, an additional scattering term. There is a general trace formula for rank one perturbations due to Krein \cite{Kr} which, however, does not give any information about diffractive orbits or the scattering term. Our work is also closely related to Venkov's proof \cite{V} of an analogue of Selberg's trace formula for an automorphic Schr\"odinger operator with a continuous nonnegative potential. Venkov extended the notion of Eisenstein series to such potentials and proved that they satisfy an analogous functional equation. He also introduced an analogue of Selberg's zeta function in this case and verified that it possesses properties similar to those of the original zeta function.

\subsection{Notation}
Before we state our main results we fix some notation. Let $\HH=\left\{x+\i y\in\CC \mid y>0\right\}$ be the upper half-plane with Riemannian metric $ds^{2}=y^{-2}(dx^{2}+dy^{2})$ and volume element $d\mu(z)=y^{-2}dx\,dy$, where $z=x+\i y$. The geodesic distance on $\HH$ is given by $d(\cdot,\cdot):\HH\times\HH\to\RR_{+}$, where
\begin{equation}
\cosh d(z,w)=1+\frac{|z-w|^{2}}{2\Im z \Im w}.
\end{equation}
The Laplacian on $\HH$ is of the form
\begin{equation}
\Delta=y^{2}\left(\frac{\partial^{2}}{\partial x^{2}}+\frac{\partial^{2}}{\partial y^{2}}\right).
\end{equation}
The orientation-preserving isometries of $\HH$ are the fractional linear transformations
\begin{equation}
	 z\to\frac{az+b}{cz+d},\qquad
	 \begin{pmatrix}
   a & b\\
   c & d
   \end{pmatrix}
   \in\PSL(2,\RR).
\end{equation}

A hyperbolic Riemann surface of finite volume can be represented as a quotient $\Gamma\backslash\HH$, where $\Gamma\subset\PSL(2,\RR)$ is a Fuchsian group of the first kind. Let $z_{0}\in\Gamma\backslash\HH$. Consider the operator $\Delta_{\alpha,z_{0}}=\Delta+\alpha(\delta_{z_{0}},\cdot)\delta_{z_{0}}$, $\alpha\in\RR$, which we discuss in detail in section 2.2. 

\subsection{The compact case}

If $\Gamma\backslash\HH$ is compact the spectrum of the Laplacian on the surface is discrete. The spectrum of $\Delta_{\alpha,z_{0}}$ interlaces with the spectrum of $\Delta$,
\begin{equation}
\lambda_{-M}^{\alpha}<0=\lambda_{-M}<\lambda_{-M+1}^{\alpha}<\lambda_{-M+1}\leq\cdots\leq\lambda_{n-1}\leq\lambda_{n}^{\alpha}
\leq\lambda_{n}\leq\cdots\to\infty,
\end{equation}
where $\lambda_{j}<\tfrac{1}{4}$ for $j=-M,\cdots,-1$ and $\lambda_{j}\geq\tfrac{1}{4}$ for $j\geq0$.

Let us write an eigenvalue as $\lambda=\tfrac{1}{4}+\rho^{2}$, $\rho\in\RR_{+}\cup\i\RR_{+}$. Let $\sigma\geq\tfrac{1}{2}$ and $\delta>0$. We define $H_{\sigma,\,\delta}$ to be the space of functions $h:\CC\to\CC$, s. t.\\
\begin{enumerate}
\item[(i)] $h$ is even,
\item[(ii)] $h$ is analytic in the strip $\left|\Im{\rho}\right|\leq\sigma$,
\item[(iii)] $h(\rho)<\!\!<(1+|\Re\rho|)^{-2-\delta}$ uniformly in the strip $\left|\Im\rho\right|\leq\sigma$.\\
\end{enumerate}

Let $m_{\Gamma}=|\scrI|<+\infty$, where
\begin{equation}
\scrI=\lbrace\gamma\in\Gamma\mid\gamma z_{0}=z_{0}\rbrace.
\end{equation}
Let $\psi(s)=\tfrac{1}{2\pi}\Gamma'(s)/\Gamma(s)$, where $\Gamma(s)$ denotes the Gamma function. For $h\in H_{\sigma,\delta}$ and $\beta\in\RR$ we define the transform of $h$
\begin{equation}\label{trans}
g_{\beta,k}(t)=\frac{(-1)^{k}}{2\pi\i k}\int_{-\i\nu-\infty}^{-\i\nu+\infty}\frac{h'(\rho)\e^{-\i\rho t}}{(1+m_{\Gamma}\beta\psi(\tfrac{1}{2}+\i\rho))^{k}}d\rho,
\end{equation}
where $\nu>v_{\beta}$ if $1+m_{\Gamma}\beta\psi(\tfrac{1}{2}+\i\rho)$ has a zero $-\i v_{\beta}$ in the interval $(0,-\i\sigma)$, and $\nu=0$ otherwise. In fact there is at most one zero of $1+m_{\Gamma}\beta\psi(\tfrac{1}{2}+\i\rho)$ in the halfplane $\Im\rho<0$ and it lies on the imaginary axis. To see this consider the representation (B9) in \cite{Iw} (beware of our additional factor $\tfrac{1}{2\pi}$). It follows that
\begin{equation}
\Im(1+m_{\Gamma}\beta\psi(\tfrac{1}{2}+\i\rho))=m_{\Gamma}\beta\frac{\Re\rho}{2\pi}\sum_{n=0}^{\infty}|n+\tfrac{1}{2}+\i\rho|^{-2},
\end{equation}
so $1+m_{\Gamma}\beta\psi(\tfrac{1}{2}+\i\rho)$ has no zeros off the imaginary line. For $\rho=-\i v$ and $v>-\tfrac{1}{2}$ we have $1+m_{\Gamma}\beta\psi(\tfrac{1}{2}+v)\in\RR$ and
\begin{equation}
\frac{d}{dv}(1+m_{\Gamma}\beta\psi(\tfrac{1}{2}+v))=\frac{m_{\Gamma}\beta}{2\pi}\sum_{n=0}^{\infty}(n+\tfrac{1}{2}+v)^{-2}>0.
\end{equation}
This together with the asymptotics $\Gamma'(s)/\Gamma(s)=\log s+O(|s|^{-1})$ for $\Re s>\tfrac{1}{2}$ (cf. (B11), pp. 198-9 in \cite{Iw}) and the observation that $\lim_{v\to-1/2^{+}}\psi(\tfrac{1}{2}+v)=-\infty$ implies that $1+m_{\Gamma}\beta\psi(\tfrac{1}{2}+\i\rho)$ has exactly one zero in the halfplane $\Im\rho<\tfrac{1}{2}$ at $\rho=-\i v_{\beta}$, $v_{\beta}>-\tfrac{1}{2}$.

Selberg's trace formula relates the trace of the Laplacian to a sum over periodic orbits on $\Gamma\backslash\HH$. Our trace formula relates the difference between the traces of $\Delta_{\alpha,z_{0}}$ and $\Delta$ to so-called \textit{diffractive orbits}. By a diffractive orbit associated with a group element $\gamma\in\Gamma$, which does not fix $z_{0}$, we mean the following. We consider the unique geodesic which connects $z_{0}$ to $\gamma z_{0}$ in $\HH$ and project it onto the surface $\Gamma\backslash\HH$. We hence obtain an orbit which starts and returns to $z_{0}$, which is, however, not necessarily a periodic orbit. We denote the length $d(\gamma z_{0},z_{0})$ of such an orbit by $l_{\gamma,z_{0}}$.

The following trace formula for compact surfaces is proven in section 4.
\begin{thm}\label{thm1}
Suppose $\Gamma\backslash\HH$ is compact. Let $\sigma>C(\Gamma,\alpha,z_{0})$, where $C(\Gamma,\alpha,z_{0})$ is defined by \eqref{large1} and let $\beta(\alpha)=\alpha/(1+c\alpha)$ for some real constant $c$. Let $\nu$ be defined as above. For any $\delta>0$ and $h\in H_{\sigma,\delta}$ we have the identity
\begin{equation}
\begin{split}
\sum_{j=-M}^{\infty}\lbrace h(\rho^{\alpha}_{j})-h(\rho_{j})\rbrace
=\;&\frac{1}{2\pi}\int_{-\i\nu-\infty}^{-\i\nu+\infty}h(\rho)\frac{m_{\Gamma}\beta(\alpha)\psi'(\tfrac{1}{2}+\i\rho)}{1+m_{\Gamma}\beta(\alpha)\psi(\tfrac{1}{2}+\i\rho)}d\rho\\
&+\sum_{k=1}^{\infty}\beta(\alpha)^{k}\sum_{\gamma_{1},\cdots,\gamma_{k}\in\Gamma\backslash\scrI}\,\int_{l_{\gamma_{1},z_{0}}}^{\infty}\cdots\int_{l_{\gamma_{k},z_{0}}}^{\infty}\frac{g_{\beta,k}(t_{1}+\cdots+t_{n})\prod_{n=1}^{k}dt_{n}}{\prod_{n=1}^{k}\sqrt{\cosh t_{n}-\cosh l_{\gamma_{n},z_{0}}}}.
\end{split}
\end{equation}
\end{thm}

\subsection{The noncompact case}
If $\Gamma\backslash\HH$ has one cusp, one can employ a simple coordinate transformation to move the cusp of $\Gamma\backslash\HH$ to the standard cusp of width one at $\infty$. Its stability group is then given by
\begin{equation}
\Gamma_{\infty}=\left\{
\begin{pmatrix}
   1 & n\\
   0 & 1
   \end{pmatrix}\mid n\in\ZZ\right\},
\end{equation}
which acts by integer translations in the $x$-variable.

The generalised eigenfunctions of the Laplacian on $\Gamma\backslash\HH$ are given by non-holomorphic Eisenstein series
\begin{equation}
E(z,s)=\sum_{\gamma\in\Gamma_{\infty}\backslash\Gamma}(\Im\gamma z)^{s}, \qquad \Re s>1,
\end{equation}
which satisfy
\begin{equation}
(\Delta+s(1-s))E(\cdot,s)=0,
\end{equation}
and have the asymptotics
\begin{equation}
E(x+\i y,s)=y^{s}+\varphi(s)y^{1-s}+O(\e^{-2\pi y}), \qquad y\to\infty,
\end{equation}
in the cusp, where $\varphi(s)$ is the scattering coefficient. $E(z,s)$ can be meromorphically continued to the whole of $\CC$ and satisfies a functional equation
\begin{equation}
E(z,s)=\varphi(s)E(z,1-s).
\end{equation}
The spectrum of the Laplacian consists of a continuous part $[\tfrac{1}{4},\infty)$ which corresponds to Eisenstein series $\lbrace E(z,\tfrac{1}{2}+\i t)\rbrace_{t\in\RR}$ and a discrete part
\begin{equation}
0=\lambda_{-M}<\lambda_{-M+1}\leq\cdots\leq\lambda_{-1}<\tfrac{1}{4}\leq\lambda_{0}\leq\lambda_{1}\leq\cdots\leq\lambda_{j}\leq\cdots\to\infty,
\end{equation}
which consists of the residual part $\lbrace\lambda_{j}\rbrace_{j=-M}^{-1}$ with eigenfunctions which are either residues of $E(\cdot,s)$ in the interval $[0,1]$ or cusp forms, and the cuspidal part $\lbrace\lambda_{j}\rbrace_{j=0}^{\infty}$ all of whose eigenfunctions are cusp forms.

The spectrum of $\Delta_{\alpha,z_{0}}$ consists of a continuous part $[\tfrac{1}{4},\infty)$ which corresponds to perturbed Eisenstein series and a discrete part (cf. \cite{CdV}, Thm. 3) consisting of perturbed cusp forms and a finite number of perturbed Maass forms. The perturbed Eisenstein series play the analogous role of the scattering solutions in the unperturbed case. They satisfy a functional equation similar to that satisfied by classical Eisenstein series. We give the proof of the following Theorem in section 5.
\begin{thm}\label{thm2}
The generalised eigenfunctions of $\Delta_{\alpha,z_{0}}$ are given by perturbed Eisenstein series $E^{\alpha,z_{0}}(z,s)$ which have the asymptotics
\begin{equation}
E^{\alpha,z_{0}}(x+\i y,s)=y^{s}+\varphi_{\alpha,z_{0}}(s)y^{1-s}+O(\e^{-2\pi y}), \qquad y\to\infty,
\end{equation}
in the cusp, where $\varphi_{\alpha,z_{0}}(s)$ is the scattering coefficient for the perturbed Laplacian $\Delta_{\alpha,z_{0}}$. They satisfy an analogous funtional equation
\begin{equation}
E^{\alpha,z_{0}}(z,s)=\varphi_{\alpha,z_{0}}(s)E^{\alpha,z_{0}}(z,1-s).
\end{equation}
\end{thm}
The cuspidal part of the discrete spectrum includes possible degenerate eigenvalues of the Laplacian, but may also include new eigenvalues. To simplify notation we will only list the new eigenvalues which are nondegenerate (see section 3, Proposition \ref{eigencond})
\begin{equation}
\lambda^{\alpha}_{-M}<0<\lambda^{\alpha}_{-M+1}<\cdots<\lambda^{\alpha}_{0}<\lambda^{\alpha}_{1}<\cdots<\lambda^{\alpha}_{j}<\cdots.
\end{equation}
The small eigenvalues $\lbrace\lambda^{\alpha}_{j}\rbrace_{j=-M}^{-1}$ correspond to eigenfunctions which are either residues of $E^{\alpha,z_{0}}(z,s)$ or perturbed cusp forms. We will refer to the residues as residual Maass forms. We will sometimes refer to the small eigenvalues as the residual spectrum, although some of them may be associated with cusp forms instead of residues of the Eisenstein series. The remaining eigenvalues $\lbrace\lambda^{\alpha}_{j}\rbrace_{j\geq0}$ correspond to eigenfunctions which decay exponentially in the cusp and have a logarithmic singularity at $z_{0}$. We will refer to such eigenfunctions as pseudo cusp forms and the associated eigenvalues as the pseudo cuspidal spectrum. For a detailed discussion of the perturbed discrete spectrum see section 2. In section 3 we will prove that for almost all pairs $(\alpha,z_{0})$ the perturbed cuspidal spectrum is empty.

We now turn to our second main result, the trace formula for surfaces with one cusp, which we will prove in section 7. We interpret the sum over closed geodesics in the trace formula as a sum over diffractive orbits, where we sum over the number of visits to $z_{0}$ and all combinations of orbits connecting $z_{0}$ to itself on $\Gamma\backslash\HH$.

\begin{thm}\label{thm3}
Suppose $\Gamma\backslash\HH$ has one cusp. Let $(\alpha,z_{0})\in(\RR\backslash\lbrace0\rbrace)\times\Gamma\backslash\HH$. Let $\sigma>C(\Gamma,\alpha,z_{0})$, where $C(\Gamma,\alpha,z_{0})$ is defined as before. For any $\delta>0$ and $h\in H_{\sigma,\,\delta}$ we have the identity
\begin{equation}\label{trace}
\begin{split}
\sum_{j\geq-M}h(\rho^{\alpha}_{j})-\sum_{j\geq-M}h(\rho_{j})
=\;&\frac{1}{2\pi}\int_{-\i\tilde{\sigma}-\infty}^{-\i\tilde{\sigma}+\infty}h(\rho)\frac{m_{\Gamma}\beta\psi'(\tfrac{1}{2}+\i\rho)}{1+m_{\Gamma}\beta\psi(\tfrac{1}{2}+\i\rho)}d\rho\\
&+\tfrac{1}{2}\delta_{\Gamma}h(0)+\frac{1}{4\pi}\int_{-\infty}^{\infty}h(\rho)\left\{\frac{\varphi_{\alpha,\,z_{0}}'}{\varphi_{\alpha,\,z_{0}}}(\tfrac{1}{2}+\i\rho)-\frac{\varphi'}{\varphi}(\tfrac{1}{2}+\i\rho)\right\}d\rho\\
&+\sum_{k=1}^{\infty}\,\beta(\alpha)^{k}\sum_{\gamma_{1},\cdots,\gamma_{k}\in\Gamma\backslash\left\{\id\right\}}\,\int_{l_{\gamma_{1},z_{0}}}^{\infty}\cdots\int_{l_{\gamma_{k},z_{0}}}^{\infty}\frac{g_{\beta,k}(t_{1}+...+t_{k})\prod_{n=1}^{k}dt_{n}}{\prod_{n=1}^{k}\sqrt{\cosh t_{n}-\cosh l_{\gamma_{n},z_{0}}}}
\end{split}
\end{equation}
where $\delta_{\Gamma}=1$ if $\lambda=\tfrac{1}{4}$ is not an eigenvalue of the Laplacian, and $\delta_{\Gamma}=0$ otherwise.
\end{thm}
\begin{remark}
On the spectral side it does not matter whether we count multiplicities when summing over the eigenvalues. If we count multiplicities, we get cancellation with the degenerate eigenvalues of the Laplacian appearing in the spectrum of $\Delta_{\alpha,z_{0}}$. Please see section 2.3. for details.
\end{remark}

We conclude with a brief overview of the structure of this paper. Section 2 introduces the machinery which we will use throughout the paper such as automorphic functions and self-adjoint extensions of the Laplacian. Section 3 proceeds to introduce the key tool in the proof of the trace formula, the relative zeta function, and states its main properties which we will use later on. Of particular importance here is Theorem 11 which locates the zeros and poles of the relative function which correspond to perturbed and unperturbed eigenvalues and, in the noncompact case, resonances. Section 4 gives the proof of Theorem 1 and illustrates the general strategy of the proof of the trace formula in the case of compact surfaces. Section 5 concentrates on the derivation of the perturbative analogue of the classical non-holomorphic Eisenstein series (Lemma 20) which plays a key role in the spectral theory of singular perturbations of the Laplacian on noncompact hyperbolic surfaces. In particular we prove a functional equation (Corollary 21) for perturbed Eisenstein series. To prove the final trace formula for surfaces with a cusp we require a bound on the continous part of the relative zeta function which we prove in Section 6 (Proposition 24). In Section 7 we use the results of the previous section to give a full proof of Theorem 3. In analogy with the case of the unperturbed Laplacian it is possible to relate the leading term on the RHS of \eqref{trace} to the trace formula for singular perturbations on the sphere. We outline this analogy in section 8. Section 9 uses our previous work on the trace formula to derive an analogue of Selberg's zeta function for singular perturbations.

\section{Self-adjoint extensions on automorphic functions}

\subsection{Automorphic functions}

We define the space of automorphic functions to be
\begin{equation}
\scrA(\Gamma\backslash\HH)=\left\{f:\HH\to\CC\mid\forall\gamma\in\Gamma:f(\gamma z)=f(z)\right\}.
\end{equation}
We denote the space of twice continuously differentiable automorphic functions by 
\begin{equation}
C^{2}(\Gamma\backslash\HH)=\scrA(\Gamma\backslash\HH)\cap C^{2}(\HH).
\end{equation}
Furthermore we define the space of automorphic forms of eigenvalue $\lambda=s(1-s)$ by
\begin{equation}
\scrA_{s}(\Gamma\backslash\HH)=\left\{f\in\scrA(\Gamma\backslash\HH)\mid f\in C^{2}(\Gamma\backslash\HH),\;(\Delta+s(1-s))f=0\right\}.
\end{equation}
and the space of infinitely differentiable automorphic functions as 
\begin{equation}
C^{\infty}(\Gamma\backslash\HH)=\scrA(\Gamma\backslash\HH)\cap C^{\infty}(\HH).
\end{equation}
Let $\scrF$ be a connected fundamental domain for $\Gamma$. Let $f\in\scrA(\Gamma\backslash\HH)$ be a measurable function. We introduce the norm
\begin{equation}
\left\|f\right\|=\left(\int_{\scrF}|f(z)|^{2}d\mu(z)\right)^{1/2}
\end{equation}
and the space of square-integrable functions
\begin{equation}
L^{2}(\Gamma\backslash\HH)=\left\{f\in\scrA(\Gamma\backslash\HH)\mid\left\|f\right\|<+\infty\right\}.
\end{equation}
Let $f,g\in\scrA(\Gamma\backslash\HH)$ be measurable functions. We define an inner product on $L^{2}(\Gamma\backslash\HH)$ by
\begin{equation}
(f,g)=\int_{\scrF}f(z)\overline{g(z)}d\mu(z).
\end{equation}
We define the second Sobolev space on $\Gamma\backslash\HH$ by
\begin{equation}
H^{2}(\Gamma\backslash\HH)=\left\{f\in L^{2}(\Gamma\backslash\HH)\mid\Delta f\in L^{2}(\Gamma\backslash\HH)\right\}.
\end{equation}
If $\Gamma\backslash\HH$ has a cusp, every $f\in C^{\infty}(\Gamma\backslash\HH)$ is periodic in the $x$-variable and can therefore be expanded in a Fourier series
\begin{equation}
f(x+\i y)=\sum_{n\in\ZZ}f_{n}(y)e(nx),\qquad e(x)=\exp(2\pi\i x),
\end{equation}
where
\begin{equation}
f_{n}(y)=\int_{0}^{1}f(x+\i y)e(nx)dx
\end{equation}
and the series converges absolutely. Any $f\in\scrA_{s}(\Gamma\backslash\HH)$ with $f(x+\i y)=o(e^{2\pi y})$ has the asymptotics
\begin{equation}
f(x+\i y)=A_{f}(s)y^{s}+B_{f}(s)y^{1-s}+O(\e^{-2\pi y}), \qquad y\to\infty,
\end{equation}
in the cusp, where $A_{f}(s),B_{f}(s)\in\CC$. Let $f_{0}(y)=A_{f}y^{s}+B_{f}y^{1-s}$. We define the space of cuspidal functions to be
\begin{equation}
C^{\infty}_{cusp}(\Gamma\backslash\HH)=\left\{f\in C^{\infty}\left(\Gamma\backslash\HH\right)|f_{0}=0\right\}.
\end{equation}

An important example of an automorphic function is the automorphic Green's function $G^{\Gamma}_{s}(\cdot,z_{0})$ which is defined by the method of images for $\Re s>1$
\begin{equation}
G^{\Gamma}_{s}(z,z_{0})=\sum_{\gamma\in\Gamma}G_{s}(z,\gamma z_{0})
\end{equation}
where $G_{s}(\cdot,w)$ denotes the Green's function on $\HH$. It solves the equation 
\begin{equation}
(\Delta+s(1-s))G^{\Gamma}_{s}(\cdot,z_{0})=\delta_{z_{0}},
\end{equation}
and has a meromorphic continuation to all of $\CC$ (cf. \cite{Hj3}). 

If $\Gamma\backslash\HH$ is compact, we have
\begin{equation}
G^{\Gamma}_{1-s}(z,z_{0})=G^{\Gamma}_{s}(z,z_{0})
\end{equation}
which follows from the spectral expansion
\begin{equation}
G^{\Gamma}_{s}(z,z_{0})=\sum_{j=-M}^{\infty}\frac{\varphi_{j}(z)\overline{\varphi_{j}(z_{0})}}{\lambda_{j}-s(1-s)}.
\end{equation}
If $\Gamma\backslash\HH$ has a cusp, it satisfies the functional equation (cf. (7.19), p. 105, \cite{Iw})
\begin{equation}
G^{\Gamma}_{1-s}(z,z_{0})=G^{\Gamma}_{s}(z,z_{0})-\frac{E(z_{0},s)E(z,1-s)}{1-2s}.
\end{equation}
In the cusp $G^{\Gamma}_{s}(z,z_{0})$ has the asymptotics
\begin{equation}
G^{\Gamma}_{s}(z,z_{0})=\frac{E(z_{0},s)}{1-2s}y^{1-s}+O(\e^{-2\pi y}),\qquad y\to\infty,
\end{equation}
which shows straightaway that $G^{\Gamma}_{s}(\cdot,z_{0})\in L^{2}(\Gamma\backslash\HH)$ if $\Re s>\tfrac{1}{2}$ or $E(z_{0},s)=0$.

\subsection{Self-adjoint extensions}

Let $z_{0}\in\Gamma\backslash\HH$. In the following it will be clear from the context when we mean by $z_{0}$ a $\Gamma$-orbit or a representative in $\HH$. We consider the one-parameter family of rank-one perturbations $\Delta_{\alpha,z_{0}}=\Delta+\alpha\delta_{z_{0}}(\delta_{z_{0}},\cdot)$, $\alpha\in\RR$, where $\delta_{z_{0}}$ denotes the Dirac functional at $z_{0}$. It is known (cf. \cite{CdV}, section 1 or \cite{Al}, section 1.3.2) that $\lbrace\Delta_{\alpha,z_{0}}\rbrace_{\alpha\in\RR}$ can be realised as the family of self-adjoint extensions of $\Delta$ from the space 
\begin{equation}
D_{0}=\lbrace f\in H^{2}(\Gamma\backslash\HH)\mid f(z_{0})=0\rbrace. 
\end{equation}
It can easily be checked that $\Delta:D_{0}\to L^{2}(\Gamma\backslash\HH)$ is symmetric with respect to $(\cdot,\cdot)$. Also $\Delta_{\alpha,z_{0}}|_{D_{0}}=\Delta|_{D_{0}}$ for any $\alpha\in\RR$. It is well-known that $\Delta|_{\D_{0}}$ has deficiency indices $(1,1)$. The deficiency elements are the automorphic Green's function $\G^{\Gamma}_{\,t}(z,z_{0})$ and its conjugate $\G^{\Gamma}_{\,\bar{t}}(z,z_{0})$, where $t(1-t)=\i$ and $\Re t>\tfrac{1}{2}$. $\G^{\Gamma}_{\,t}(\cdot,z_{0})$ uniquely satisfies
\begin{equation}
(\Delta+t(1-t))G^{\Gamma}_{\,t}(\cdot,z_{0})=\delta_{z_{0}},
\end{equation}
and has the asymptotics
\begin{equation}
G^{\Gamma}_{\,t}(z,z_{0})=\frac{m}{2\pi}\log d(z,\gamma z_{0})+O(1),\qquad z\to z_{0},
\end{equation}
where $m=\ord(z_{0})$ if $z_{0}$ corresponds to an elliptic fixed point in $\HH$, and otherwise $m=1$. $G^{\Gamma}_{t}(\cdot,z_{0})$ is automorphic in the sense that for any $z\notin\Gamma z_{0}$
\begin{equation}
G^{\Gamma}_{t}(\gamma z,z_{0})=G^{\Gamma}_{t}(z,z_{0}).
\end{equation}

With respect to the graph inner product $\left\langle\cdot,\cdot\right\rangle=(\cdot,(\Delta\Delta^{*}+1)\cdot)$ we have the direct sum
\begin{equation}
\overline{\D_{0}}\oplus\scrL\lbrace G^{\Gamma}_{\,t}(\cdot,z_{0})\rbrace\oplus\scrL\lbrace G^{\Gamma}_{\,\bar{t}}(\cdot,z_{0})\rbrace
\end{equation}
which forms a subspace of $H^{2}((\Gamma\backslash\HH)\backslash\lbrace z_{0}\rbrace)$. Define the subspace $\D_{\varphi}\subset\overline{D_{0}}\oplus\scrL\lbrace G^{\Gamma}_{\,t}\left(\cdot,z_{0}\right)\rbrace\oplus\scrL\lbrace G^{\Gamma}_{\,\bar{t}}\left(\cdot,z_{0}\right)\rbrace$ by
\begin{equation}
\D_{\varphi}=\left\{g+cG^{\Gamma}_{\,t}(\cdot,z_{0})+c\,\e^{\i\varphi}G^{\Gamma}_{\,\bar{t}}(\cdot,z_{0})|g\in\D_{0}, c\in\CC\right\}
\end{equation}
with $\varphi\in(-\pi,\pi)$. This means any $f\in\D_{\varphi}$ satisfies the boundary condition
\begin{equation}
f(z)=c\left\{\frac{m}{2\pi}\log d(z,z_{0})+\Re k+\i\Im k\tan\frac{\varphi}{2}\right\}+o(1),\qquad z\to z_{0},
\end{equation}
where $c,k\in\CC$ are constants and $k$ depends on $t$. Let us introduce the function
\begin{equation}
A\left(s,t\right)=\frac{1}{2}\lim_{z\to z_{0}}(G^{\Gamma}_{s}(z,z_{0})-G^{\Gamma}_{t}(z,z_{0})).
\end{equation}
which is meromorphic in both variables. The limit clearly exists, since we have removed the logarithmic singularity at $z_{0}$. $\Delta|_{D_{0}}$ admits a one parameter family of self-adjoint extensions $\left\{\Delta_{\varphi}\right\}$, $\varphi\in\left(0,2\pi\right)$. The operator $\Delta_{\varphi}:\D_{\varphi}\to L^{2}(\Gamma\backslash\HH)$ is defined (cf. \cite{Al}, section 1.2.5) by
\begin{equation}
\Delta_{\varphi}f=\Delta g - c\,t(1-t)G^{\Gamma}_{\,t}(z,z_{0})
-c\,\bar{t}(1-\bar{t})\e^{\i\varphi}G^{\Gamma}_{\,\bar{t}}(z,z_{0}).
\end{equation}
We have the following theorem.
\begin{thm}
The operator $\Delta_{\varphi}$ is self-adjoint on the domain $D_{\varphi}$.
\end{thm}
\begin{proof}
Let $f_{1},f_{2}\in D_{\varphi}$. Then we have the unique representations
\begin{equation}
\begin{split}
&f_{1}=g_{1}+a_{1}G^{\Gamma}_{t}(\cdot,z_{0})+b_{1}G^{\Gamma}_{\bar{t}}(\cdot,z_{0}),
\qquad \text{where}\;b_{1}=a_{1}\e^{\i\varphi}\;\text{and}\;g_{1}(z_{0})=0,\\
&f_{2}=g_{2}+a_{2}G^{\Gamma}_{t}(\cdot,z_{0})+b_{2}G^{\Gamma}_{\bar{t}}(\cdot,z_{0}),
\qquad \text{where}\;b_{2}=a_{2}\e^{\i\varphi}\text{and}\;g_{2}(z_{0})=0.
\end{split}
\end{equation}
We make the calculation
\begin{equation}
\begin{split}
&(\Delta_{\varphi}f_{1},f_{2})-(f_{1},\Delta_{\varphi}f_{2})\\
=\;&(\Delta_{\varphi}(g_{1}+a_{1}G^{\Gamma}_{t}(\cdot,z_{0})+b_{1}G^{\Gamma}_{\bar{t}}(\cdot,z_{0})),g_{2}+a_{2}G^{\Gamma}_{t}(\cdot,z_{0})+b_{2}G^{\Gamma}_{\bar{t}}(\cdot,z_{0}))\\
&-(g_{1}+a_{1}G^{\Gamma}_{t}(\cdot,z_{0})+b_{1}G^{\Gamma}_{\bar{t}}(\cdot,z_{0}),
\Delta_{\varphi}(g_{2}+a_{2}G^{\Gamma}_{t}(\cdot,z_{0})+b_{2}G^{\Gamma}_{\bar{t}}(\cdot,z_{0})))\\
=\;&(\Delta g_{1}-\i a_{1}G^{\Gamma}_{t}(\cdot,z_{0})+\i b_{1}G^{\Gamma}_{\bar{t}}(\cdot,z_{0}) ,g_{2}+a_{2}G^{\Gamma}_{t}(\cdot,z_{0})+b_{2}G^{\Gamma}_{\bar{t}}(\cdot,z_{0}))\\
&-(g_{1}+a_{1}G^{\Gamma}_{t}(\cdot,z_{0})+b_{1}G^{\Gamma}_{\bar{t}}(\cdot,z_{0}),
\Delta g_{2}-\i a_{2}G^{\Gamma}_{t}(\cdot,z_{0})+\i b_{2}G^{\Gamma}_{\bar{t}}(\cdot,z_{0})).
\end{split}
\end{equation}
We expand the last line to obtain
\begin{equation}
\begin{split}
&(\Delta_{\varphi}f_{1},f_{2})-(f_{1},\Delta_{\varphi}f_{2})\\
=\;&(\Delta g_{1},g_{2})+a_{1}(-\i G^{\Gamma}_{t}(\cdot,z_{0}),g_{2})+ b_{1}(\i G^{\Gamma}_{\bar{t}}(\cdot,z_{0}),g_{2})\\
&+\overline{a_{2}}(\Delta g_{1},G^{\Gamma}_{t}(\cdot,z_{0}))
-\i a_{1}\overline{a_{2}}\left\|G^{\Gamma}_{t}(\cdot,z_{0})\right\|^{2}
+\i b_{1}\overline{a_{2}}(G^{\Gamma}_{\bar{t}}(\cdot,z_{0}),G^{\Gamma}_{t}(\cdot,z_{0}))\\
&+\overline{b_{2}}(\Delta g_{1},G^{\Gamma}_{\bar{t}}(\cdot,z_{0}))
-\i a_{1}\overline{b_{2}}(G^{\Gamma}_{t}(\cdot,z_{0}),G^{\Gamma}_{\bar{t}}(\cdot,z_{0}))
+\i b_{1}\overline{b_{2}}\left\|G_{\bar{t}}(\cdot,z_{0})\right\|^{2}\\
&-(g_{1},\Delta g_{2})-a_{1}(G^{\Gamma}_{t}(\cdot,z_{0}),\Delta g_{2})-b_{1}(G^{\Gamma}_{\bar{t}}(\cdot,z_{0}),\Delta g_{2})\\
&-\overline{a_{2}}(g_{1},-\i G^{\Gamma}_{t}(\cdot,z_{0}))
-\i a_{1}\overline{a_{2}}\left\|G^{\Gamma}_{t}(\cdot,z_{0})\right\|^{2}
-\i b_{1}\overline{a_{2}}(G^{\Gamma}_{\bar{t}}(\cdot,z_{0}),G^{\Gamma}_{t}(\cdot,z_{0}))\\
&-\overline{b_{2}}(g_{1},\i G^{\Gamma}_{\bar{t}}(\cdot,z_{0}))
+\i a_{1}\overline{b_{2}}(G^{\Gamma}_{t}(\cdot,z_{0}),G^{\Gamma}_{\bar{t}}(\cdot,z_{0}))
+\i b_{1}\overline{b_{2}}\left\|G^{\Gamma}_{\bar{t}}(\cdot,z_{0})\right\|^{2}
\end{split}
\end{equation}
Now observe that the property $\overline{G^{\Gamma}_{t}(z,z_{0})}=G^{\Gamma}_{\bar{t}}(z,z_{0})$ implies
\begin{equation}
\left\|G^{\Gamma}_{t}(\cdot,z_{0})\right\|^{2}
=\int_{\Gamma\backslash\HH}|G^{\Gamma}_{t}(z,z_{0})|^{2}d\mu(z)
=\int_{\Gamma\backslash\HH}|G^{\Gamma}_{\bar{t}}(z,z_{0})|^{2}d\mu(z)
=\left\|G^{\Gamma}_{\bar{t}}(\cdot,z_{0})\right\|^{2}.
\end{equation}
Consequently, we obtain
\begin{equation}
\begin{split}
(\Delta_{\varphi}f_{1},f_{2})-(f_{1},\Delta_{\varphi}f_{2})
=\;&(\Delta g_{1},g_{2})-(g_{1},\Delta g_{2})\\
&+a_{1}[(-\i G^{\Gamma}_{t}(\cdot,z_{0}),g_{2})-(G^{\Gamma}_{t}(\cdot,z_{0}),\Delta g_{2})]\\
&+b_{1}[(\i G^{\Gamma}_{\bar{t}}(\cdot,z_{0}),g_{2})-(G^{\Gamma}_{\bar{t}}(\cdot,z_{0}),\Delta g_{2})]\\
&+\overline{a_{2}}[(\Delta g_{1},G^{\Gamma}_{t}(\cdot,z_{0}))-(g_{1},-\i G^{\Gamma}_{t}(\cdot,z_{0}))]\\
&+\overline{b_{2}}[(\Delta g_{1},G^{\Gamma}_{\bar{t}}(\cdot,z_{0}))-(g_{1},\i G^{\Gamma}_{\bar{t}}(\cdot,z_{0}))]\\
&+2\i(b_{1}\overline{b_{2}}-a_{1}\overline{a_{2}})\left\|G^{\Gamma}_{t}(\cdot,z_{0})\right\|^{2}
\end{split}
\end{equation}
To see that the RHS vanishes let us consider the six terms separately. The first term vanishes by symmetricity of $\Delta$ on $D_{0}$. For the second term we use the distributional identity $(\Delta+\i)G^{\Gamma}_{t}(\cdot,z_{0})=\delta_{z_{0}}$ and $g_{2}(z_{0})=0$ to see
\begin{equation}
(-\i G^{\Gamma}_{t}(\cdot,z_{0}),g_{2})
=(\Delta G^{\Gamma}_{t}(\cdot,z_{0}),g_{2})+(\delta_{z_{0}},g_{2})
=(\Delta G^{\Gamma}_{t}(\cdot,z_{0}),g_{2})
=(G^{\Gamma}_{t}(\cdot,z_{0}),\Delta g_{2})
\end{equation}
where we have used integration by parts twice in the last step
\begin{equation}
\begin{split}
(\Delta G^{\Gamma}_{t}(\cdot,z_{0}),g_{2})
=&\int_{\Gamma\backslash\HH}\Delta G^{\Gamma}_{t}(z,z_{0})g_{2}(z)d\mu(z)\\
=&-\int_{\Gamma\backslash\HH}\nabla G^{\Gamma}_{t}(z,z_{0})\cdot\nabla g_{2}(z)d\mu(z)\\
=&(-1)^{2}\int_{\Gamma\backslash\HH}G^{\Gamma}_{t}(z,z_{0})\Delta g_{2}(z)d\mu(z)\\
=&\;(G^{\Gamma}_{t}(\cdot,z_{0}),\Delta g_{2})
\end{split}
\end{equation}
and the boundary terms vanish because of automorphicity. Terms 3-5 vanish analogously. The last term vanishes because of the constraints $b_{1}=a_{1}\e^{\i\varphi}$ and $b_{2}=a_{2}\e^{\i\varphi}$. Thus $\Delta_{\varphi}$ is symmetric on $D_{\varphi}$. It follows from Theorem X.2, p. 140, in \cite{Si} that $\Delta_{\varphi}$ is self-adjoint.
\end{proof}

The operator $\Delta_{\alpha,z_{0}}=\Delta+\alpha\left\langle\delta_{z_{0}},\cdot\right\rangle\delta_{z_{0}}$ is formally realised (cf. \cite{Al}, Thm. 1.3.2, p. 34) as a self-adjoint extension $\Delta_{\varphi(\alpha)}$ of $\Delta|_{D_{0}}$ where $\varphi(\alpha)\in(-\pi,\pi)$ is uniquely determined from the relation
\begin{equation}\label{ext}
-2\i\alpha \,A(t,\bar{t})=\cot\frac{\varphi(\alpha)}{2}
\end{equation}
which is obtained from the relation
\begin{equation}\label{rel}
(\phi,\tilde{\psi})=(-\tfrac{1}{\alpha}+c)b(\psi)
\end{equation}
in Theorem 1.3.2. In fact we have chosen the real constant $c$ to be zero. This choice simply corresponds to fixing a parameterisation of our one-parameter family of self-adjoint extensions. Following the conventions of section 1.2.5 in \cite{Al} we have
\begin{equation}
\tilde{\psi}(z)=\hat{\psi}(z)+\frac{1-\e^{\i\varphi}}{2}(G^{\Gamma}_{t}-G^{\Gamma}_{\bar{t}})(z,z_{0}),
\qquad\hat{\psi}(z_{0})=0,
\end{equation}
and
\begin{equation}
b(\psi)=\frac{1+\e^{\i\varphi}}{2},\qquad \phi=\delta_{z_{0}}.
\end{equation}
Substituting these expressions in \eqref{rel} gives \eqref{ext}.

\subsection{The spectral decomposition in the noncompact case}

If $\Gamma\backslash\HH$ has one cusp, the space $L^{2}(\Gamma\backslash\HH)$ has the orthogonal decomposition (cf. \cite{Iw}, Thm. 7.3., p. 103)
\begin{equation}
L^{2}(\Gamma\backslash\HH)=\overline{E}(\Gamma\backslash\HH)\oplus R(\Gamma\backslash\HH)\oplus \overline{C^{\infty}_{cusp}}(\Gamma\backslash\HH).
\end{equation}
Here $R(\Gamma\backslash\HH)$ is the space spanned by a finite number of residues of the Eisenstein series $\lbrace u_{-j}\rbrace_{j=1}^{M}$ in the interval $[0,1]$. $C^{\infty}_{cusp}(\Gamma\backslash\HH)$ is spanned by an orthonormal basis of cusp forms $\lbrace u_{j}\rbrace_{j\geq0}$ and $E(\Gamma\backslash\HH)$ is the space of incomplete Eisenstein series.

Any $\psi\in L^{2}(\Gamma\backslash\HH)$ has the expansion
\begin{equation}
\psi(z)=\frac{1}{2\pi}\int_{0}^{\infty}a(\rho)E(z,\tfrac{1}{2}+\i\rho)d\rho+\sum_{j\geq-M}a_{j}u_{j}(z),
\end{equation}
where
\begin{equation}
a(\rho)=(\psi,E(\cdot,\tfrac{1}{2}+\i\rho))
\end{equation}
and
\begin{equation}
a_{j}=(\psi,u_{j})
\end{equation}
which converges in the $L^{2}$-sense
\begin{equation}
\left\|\psi\right\|^{2}=\frac{1}{2\pi}\int_{0}^{\infty}|a(\rho)|^{2}d\rho+\sum_{j\geq-M}|a_{j}|^{2}<+\infty
\end{equation}

The spectrum of the Laplacian on a non-compact quotient $\Gamma\backslash\HH$ with a single cusp is therefore given by the union $[\tfrac{1}{4},\infty)\cup\lbrace\lambda_{j}\rbrace_{j=-M}^{-1}\cup\lbrace\lambda_{j}\rbrace_{j=0}^{\infty}$ of the continuous part corresponding to Eisenstein series along the critical line, the residual part and the cuspidal part of the discrete spectrum.

Formally the space $D_{\varphi}$ has the decomposition
\begin{equation}\label{decomp}
D_{\varphi}=\overline{E_{\varphi,z_{0}}}(\Gamma\backslash\HH)\oplus R_{\varphi,z_{0}}(\Gamma\backslash\HH)\oplus \overline{C_{\varphi,z_{0}}}(\Gamma\backslash\HH)
\end{equation}
where $E_{\varphi,z_{0}}(\Gamma\backslash\HH)$ is the perturbative analogue of the space of incomplete Eisenstein series (which we will not discuss in this paper) and $R_{\varphi,z_{0}}(\Gamma\backslash\HH)$ is spanned by a finite number of perturbed residual Maass forms. $C_{\varphi,z_{0}}(\Gamma\backslash\HH)$ decomposes into eigenspaces of cusp forms and 'pseudo' cusp forms which have a logarithmic singularity at $z_{0}$
\begin{equation}
C_{\varphi,z_{0}}(\Gamma\backslash\HH)=\bigoplus_{i}E_{\lambda_{i}}(\Gamma\backslash\HH)\oplus\bigoplus_{j}E_{\lambda^{\alpha}_{j}}(\Gamma\backslash\HH).
\end{equation}
Here $E_{\lambda_{i}}(\Gamma\backslash\HH)$ denotes an eigenspace corresponding to a degenerate eigenvalue of the Laplacian, whereas $E_{\lambda^{\alpha}_{j}}(\Gamma\backslash\HH)$ denotes an eigenspace corresponding to a new eigenvalue. The eigenspace $E_{\lambda_{j}}(\Gamma\backslash\HH)$ has dimension $m-1$ if $m$ is the multiplicity of the eigenvalue $\lambda_{j}$ and there is at least one eigenfunction in $E_{\lambda_{j}}(\Gamma\backslash\HH)$ that doesn't vanish at $z_{0}$. Otherwise it has dimension $m$. The eigenspaces $E_{\lambda^{\alpha}_{j}}(\Gamma\backslash\HH)$ are all of dimension one. Consult sections 2 and 3 in \cite{CdV} for details. In the proof of Theorem \ref{prop} we derive the crucial link between the relative zeta function and the old and new eigenvalues of $\Delta$ and $\Delta_{\alpha,z_{0}}$. Whereas the poles of the relative zeta function on the real and critical lines correspond to the old eigenvalues its zeros, and in particular those characterised by \eqref{speccond}, correspond to new eigenvalues.

The residual spectrum of $\Delta_{\alpha,z_{0}}$ in fact interlaces with the residual spectrum of $\Delta$
\begin{equation}
\lambda^{\alpha}_{-M}<0=\lambda_{-M}<\lambda_{-M+1}^{\alpha}\leq\cdots\leq\lambda^{\alpha}_{-1}\leq\lambda_{-1}<\tfrac{1}{4}
\end{equation}
and the associated eigenfunctions correspond to residues of the perturbed Eisenstein series if $E(z_{0},s^{\alpha}_{j})\neq0$, which is a consequence of the form of the perturbed Eisenstein series which we will determine in section 7 as scattering solutions of $\Delta_{\alpha,z_{0}}$ on $\Gamma\backslash\HH$. The perturbed Eisenstein series turn out to be singular perturbations of the original Eisenstein series and feature a logarithmic singularity at $z_{0}$.

\section{The relative zeta function}

In this section we determine a meromorphic function $S_{\alpha,z_{0}}\left(s\right)$ which contains in form of its zeros and poles all the information about the perturbed and unperturbed discrete spectrum. In section 9 we will see that $S_{\alpha,z_{0}}\left(s\right)$ has an interpretation as the quotient of a perturbed zeta function and Selberg's zeta function. The following is an analogue of Hilbert's formula for the iterated resolvent in terms of Green's functions.
\begin{lem}\label{Hilbert}
Let $\eta=t\left(1-t\right)$, $\lambda=s\left(1-s\right)$. The following identity holds
\begin{equation}
\left(\Delta+\lambda\right)\left(G^{\Gamma}_{s}\left(\cdot,z_{0}\right)-G^{\Gamma}_{t}\left(\cdot,z_{0}\right)\right)
=\left(\eta-\lambda\right)G^{\Gamma}_{t}\left(\cdot,z_{0}\right).
\end{equation}
\end{lem}
\begin{proof}
We have the following chain of identities straight from the definition of the Green's function,
\begin{equation}
\left(\Delta+\lambda\right)G^{\Gamma}_{s}(\cdot,z_{0})=\delta_{z_{0}}=(\Delta+\eta)G^{\Gamma}_{t}(\cdot,z_{0})=(\Delta+\lambda)G^{\Gamma}_{t}(\cdot,z_{0})+(\eta-\lambda)G^{\Gamma}_{t}(\cdot,z_{0}).
\end{equation}
\end{proof}

For $\Re s>1$ we introduce the function
\begin{equation}
S_{\alpha,z_{0}}(s)=\alpha^{-1}+\lim_{z\to z_{0}}(G^{\Gamma}_{s}(z,z_{0})-\tfrac{1}{2}G^{\Gamma}_{t}(z,z_{0})-\tfrac{1}{2}G^{\Gamma}_{\bar{t}}(z,z_{0})),
\end{equation}
which is well defined since the sum over the group elements is well known to converge absolutely for $\Re s>1$. Next we state some important properties of $S_{\alpha,z_{0}}(s)$.

\begin{prop}
We have the following functional equations.
\begin{itemize}
\item[(a)] If $\Gamma\backslash\HH$ is compact, then 
\begin{equation}\label{symid}
S_{\alpha,z_{0}}(s)=S_{\alpha,z_{0}}(1-s).\\
\end{equation}
\item[(b)] If $\Gamma\backslash\HH$ has one cusp, then 
\begin{equation}\label{fe}
S_{\alpha,\,z_{0}}(s)=S_{\alpha,\,z_{0}}(1-s)-\frac{E(z_{0},s)E(z_{0},1-s)}{1-2s}.
\end{equation}
\end{itemize}
\end{prop}
\begin{proof}
(a) The identity \eqref{symid} follows straightaway from the spectral expansion of the automorphic Green function which gives 
\begin{equation}\label{specexp}
S_{\alpha,\,z_{0}}(s)=\alpha^{-1}+\sum_{j=-M}^{\infty}|\varphi_{j}(z_{0})|^{2}\left\{\frac{1}{\lambda_{j}-s(1-s)}-\Re\left\{\frac{1}{\lambda_{j}-t(1-t)}\right\}\right\}.\\
\end{equation}
(b) The automorphic Green function satisfies the functional equation
\begin{equation}\label{Greenfe}
G^{\Gamma}_{1-s}(z,z_{0})=G^{\Gamma}_{s}(z,z_{0})-\frac{E(z_{0},s)E(z,1-s)}{1-2s}.
\end{equation}
The result is obtained by subtracting $\tfrac{1}{2}(G^{\Gamma}_{t}(z,z_{0})+G^{\Gamma}_{\bar{t}}(z,z_{0}))$ on both sides and taking the limit as $z\to z_{0}$.
\end{proof}

\begin{prop}
$S_{\alpha,z_{0}}(s)$ has a meromorphic continuation to the whole of $\CC$. 
\end{prop}
\begin{proof}
If $\Gamma\backslash\HH$ is compact the meromorphic continuation is immediately obtained from \eqref{specexp}. So suppose for the remainder of the proof that $\Gamma\backslash\HH$ has a cusp. The meromorphic continuation can be constructed from the spectral expansion of the automorphic Green function (cf. Thm. 3.5, p. 250 \cite{Hj3}), which is valid for $\Re s>\tfrac{1}{2}$. We obtain the meromorphic continuation to the critical line via contour integration. Let $\Re s\geq\tfrac{1}{2},\;\Im s>0$.
\begin{equation}
\begin{split}
S_{\alpha,z_{0}}(s)=&\;\alpha^{-1}+\sum_{j\geq-M}|\varphi_{j}(z_{0})|^{2}\left(\frac{1}{\lambda_{j}-s(1-s)}-\Re\left[\frac{1}{\lambda_{j}-t(1-t)}\right]\right)\\
&+\delta(s)\frac{E(z_{0},s)E(z_{0},1-s)}{1-2s}\\
&+\frac{1}{2\pi}\int_{\Gamma(s)}E(z_{0},\tfrac{1}{2}+\i w)E(z_{0},\tfrac{1}{2}-\i w)\left(\frac{1}{\tfrac{1}{4}+w^{2}-s(1-s)}-\Re\left[\frac{1}{\tfrac{1}{4}+w^{2}-t(1-t)}\right]\right)dw
\end{split}
\end{equation}
where $\Gamma(s):\RR_{+}\to\CC_{-}$ is defined by
\begin{equation}
\Gamma(s)(u)=
\begin{cases}
u-\i\epsilon(u),\qquad\text{if}\;\Re s-\tfrac{1}{2}<\epsilon(\Im s)\\
\\
u,\qquad\text{otherwise}
\end{cases}
\end{equation}
and 
\begin{equation}
\delta(s)=
\begin{cases}
1,\qquad\text{if}\;\Re s-\tfrac{1}{2}<\epsilon(\Im s)\\
\\
0,\qquad\text{otherwise}
\end{cases}
\end{equation}
where $\epsilon:\RR_{+}\to\RR_{+}$, $\lim_{u\to\infty}\epsilon(u)=0$ is a smooth function such that the region $\tfrac{1}{2}\leq\Re s\leq\epsilon(\Im s)$ contains no poles of the Eisenstein series. We obtain the meromorphic continuation to the full complex plane from the functional equation \eqref{fe}.
\end{proof}

In the previous section we saw that the discrete spectrum of $\Delta_{\alpha,z_{0}}$ consists of two parts: degenerate eigenvalues of the Laplacian (if any) that are inherited by $\Delta_{\alpha,z_{0}}$, and new eigenvalues. We refer to the latter part of the discrete spectrum of $\Delta_{\alpha,z_{0}}$ as the new part. $S_{\alpha,z_{0}}(s)$ contains information about these eigenvalues in form of its zeros.
\begin{prop}\label{eigencond}
We have the following spectral interpretation of zeros of $S_{\alpha,z_{0}}(s)$.
\begin{itemize}
\item[(a)] If $\Gamma\backslash\HH$ is compact, then
$\lambda=s(1-s)$ is in the new part of the discrete spectrum of $\Delta_{\alpha,z_{0}}$ if, and only if, $S_{\alpha,z_{0}}(s)=0$.\\
\item[(b)] If $\Gamma\backslash\HH$ has one cusp, then
$\lambda=s(1-s)$ is in the new part of the discrete spectrum of $\Delta_{\alpha,z_{0}}$ if, and only if, $S_{\alpha,z_{0}}(s)=0$ and $\Re s\geq\tfrac{1}{2}$. Zeros of $S_{\alpha,z_{0}}(s)$ in $\Re s<\tfrac{1}{2}$ correspond to resonances. 
\end{itemize}
The corresponding eigenfunctions are given by automorphic Green functions $G^{\Gamma}_{s}(\cdot,z_{0})$.
\end{prop}
\begin{proof}
(a) and (b). Let $\lambda=s(1-s)$ and $\eta=t(1-t)$. Assume that $f_{s}\in\D_{\varphi(\alpha)}\subset L^{2}(\Gamma\backslash\HH)$ is an eigenfunction of $\Delta_{\alpha,z_{0}}$ with eigenvalue $\lambda$ and that $\lambda$ does not lie in the discrete spectrum of $\Delta$. By definition
\begin{equation}
(\Delta_{\alpha,z_{0}}+\lambda)f_{s}=0.
\end{equation}
We may write equivalently, using the decomposition of $D_{\varphi\left(\alpha\right)}$,
\begin{equation}
(\Delta+\lambda)g+c(\lambda-\eta)G^{\Gamma}_{t}(\cdot,z_{0})+c\e^{\i\varphi}(\lambda-\bar{\eta})G^{\Gamma}_{\bar{t}}(\cdot,z_{0})=0.
\end{equation}
Applying the resolvent on both sides we get
\begin{equation}
g_{s}+c\frac{\lambda-\eta}{\Delta+\lambda}G_{t}^{\Gamma}(\cdot,z_{0})+c\e^{\i\varphi}\frac{\lambda-\bar{\eta}}{\Delta+\lambda}G_{t}^{\Gamma}(\cdot,z_{0})=0,
\end{equation}
and using Lemma \ref{Hilbert} we rewrite this as
\begin{equation}
g_{s}+c(G^{\Gamma}_{t}(\cdot,z_{0})-G^{\Gamma}_{s}(\cdot,z_{0}))
+c\e^{\i\varphi}(G^{\Gamma}_{\bar{t}}(\cdot,z_{0})-G^{\Gamma}_{s}(\cdot,z_{0}))=0.
\end{equation}
We take the limit as $z\to z_{0}$ and obtain
\begin{equation}
cA(s,t)+c\e^{\i\varphi}A(s,\bar{t})=0.
\end{equation}
At this point we can divide by $c$ since $c\neq0$. To see this suppose the contrary. It follows that $f=g\in D_{0}$. Therefore $0=(\Delta_{\alpha,z_{0}}+\lambda)f_{s}=(\Delta+\lambda)f_{s}$ which contradicts the assumption that $\lambda$ does not lie in the discrete spectrum of $\Delta$. After dividing we have
\begin{equation}\label{altform}
A(s,t)+\e^{\i\varphi}A(s,\bar{t})=0,
\end{equation}
which we rewrite as
\begin{equation}\label{speceqn}
\lim_{z\to z_{0}}(G^{\Gamma}_{s}-\tfrac{1}{2}\lbrace G^{\Gamma}_{t}+G^{\Gamma}_{\bar{t}}\rbrace)(z,z_{0}) -\i\tan\frac{\varphi}{2}A(t,\bar{t})=0.
\end{equation}
Finally, using \eqref{ext},
\begin{equation}
\lim_{z\to z_{0}}(G^{\Gamma}_{s}-\tfrac{1}{2}\lbrace G^{\Gamma}_{t}+G^{\Gamma}_{\bar{t}}\rbrace)(z,z_{0})+\alpha^{-1}=0.
\end{equation}

Let us now assume that $S_{\alpha,z_{0}}(s)=0$ and, if $\Gamma\backslash\HH$ has a cusp, $\Re s\geq\tfrac{1}{2}$. We claim that $G^{\Gamma}_{s}(\cdot,z_{0})$ is an eigenfunction of $\Delta_{\alpha,z_{0}}$. We have the decomposition
\begin{equation}
G^{\Gamma}_{s}(z,z_{0})
=\frac{1}{1+\e^{\,\i\varphi(\alpha)}}\left\{S_{\,\alpha,\,z_{0}}(z,s)+G^{\Gamma}_{t}(z,z_{0})+\e^{\,\i\varphi(\alpha)}G^{\Gamma}_{\bar{t}}(z,z_{0})\right\},
\end{equation}
where we have introduced
\begin{equation}
S_{\alpha,z_{0}}(z,s)=(G^{\Gamma}_{s}-G^{\Gamma}_{t})(z,z_{0})+\e^{\,\i\varphi(\alpha)}(G^{\Gamma}_{s}-G^{\Gamma}_{\bar{t}})(z,z_{0}).
\end{equation}
We see from \eqref{altform} and \eqref{speceqn} that $\lim_{z\to z_{0}}S_{\alpha,z_{0}}(z,s)=S_{\alpha,z_{0}}(s)=0$. So $G^{\Gamma}_{s}(\cdot,z_{0})\in D_{\varphi(\alpha)}$. In the above decomposition we multiply through by $1+\e^{\,\i\varphi(\alpha)}$ and by definition of $\Delta_{\alpha,z_{0}}$ obtain
\begin{equation}
(1+\e^{\,\i\varphi(\alpha)})(\Delta_{\alpha,z_{0}}+\lambda)G^{\Gamma}_{s}(\cdot,z_{0})
=(\Delta+\lambda)S_{\alpha,z_{0}}(\cdot,s)+(\lambda-\eta)G^{\Gamma}_{t}(\cdot,z_{0})+\e^{\i\varphi(\alpha)}(\lambda-\bar{\eta})G^{\Gamma}_{t}(\cdot,z_{0}).
\end{equation}
We apply Lemma \ref{Hilbert} to see
\begin{equation}
(\Delta+\lambda)S_{\alpha,z_{0}}(\cdot,s)=(\eta-\lambda)G^{\Gamma}_{t}(\cdot,z_{0})+\e^{\i\varphi(\alpha)}(\bar{\eta}-\lambda)G^{\Gamma}_{t}(\cdot,z_{0})
\end{equation}
which implies
\begin{equation}
(1+\e^{\,\i\varphi(\alpha)})(\Delta_{\alpha,z_{0}}+\lambda)G^{\Gamma}_{s}(\cdot,z_{0})=0.
\end{equation}
It follows $(\Delta_{\alpha,z_{0}}+\lambda)G^{\Gamma}_{s}(\cdot,z_{0})=0$ since $1+\e^{\,\i\varphi(\alpha)}\neq0$.

(b) If $\Gamma\backslash\HH$ has one cusp, then the zeroth Fourier coefficient of $G^{\Gamma}_{s}(\cdot,z_{0})$ is given by $(1-2s)^{-1}E(z_{0},s)y^{1-s}$ and we see straightaway that $s>\tfrac{1}{2}$ implies $G^{\Gamma}_{s}(\cdot,z_{0})\in L^{2}(\Gamma\backslash\HH)$. We shall see in the proof of Lemma \ref{prop} that for $\Re s=\tfrac{1}{2}$, $S_{\alpha,z_{0}}(s)=0$ implies $E(z_{0},s)=0$, so $G^{\Gamma}_{s}(\cdot,z_{0})\in L^{2}(\Gamma\backslash\HH)$ since it is a cusp form of $\Delta_{\alpha,z_{0}}$. It also follows for $\Re s=\tfrac{1}{2}$ that $G^{\Gamma}_{s}(\cdot,z_{0})=G^{\Gamma}_{1-s}(\cdot,z_{0})$ by the functional equation $G^{\Gamma}_{s}(\cdot,z_{0})-G^{\Gamma}_{1-s}(\cdot,z_{0})=(1-2s)^{-1}(E(z_{0},s)E(z,1-s))$. If $\Re s<\tfrac{1}{2}$ it follows from the form of the zeroth Fourier coefficient of $G^{\Gamma}_{s}(\cdot,z_{0})$ that $G^{\Gamma}_{s}(\cdot,z_{0})\notin L^{2}(\Gamma\backslash\HH)$. Thus zeros of $S_{\alpha,z_{0}}(s)$ in $\Re s<\tfrac{1}{2}$ correspond to resonances.
\end{proof}

Define
\begin{equation}
\scrI=\lbrace\gamma\in\Gamma\mid\gamma z_{0}=z_{0}\rbrace.
\end{equation}
$\scrI=\lbrace\id\rbrace$ unless $z_{0}$ is an elliptic fixed point in which case $\scrI$ equals the finite cyclic group generated by the corresponding elliptic group element. We can write $S_{\alpha,z_{0}}(s)$ in a more convenient form if $\Re s>1$.
\begin{prop}
Let $\Re s>1$. $S_{\alpha,z_{0}}(s)$ can be written in the form
\begin{equation}
S_{\alpha,z_{0}}(s)=\beta^{-1}+m\psi(s)+\sum_{\gamma\in\Gamma\backslash\scrI}G_{s}(z_{0},\gamma z_{0}),
\end{equation}
where $\beta=\beta(\alpha)$ and $m=|\scrI|$.
\end{prop}
\begin{proof}
It can be seen (cf. \cite{Mf}, p. 17, (116-18)) that
\begin{equation}
\psi(s)-\psi(t)=\lim_{z\to z_{0}}(G_{s}(z,z_{0})-G_{t}(z,z_{0})).
\end{equation}
where $\psi=(2\pi)^{-1}\Gamma/\Gamma'$. From the definition of $S_{\alpha,z_{0}}(s)$ we have for $\Re s>1$
\begin{equation}
\begin{split}
S_{\alpha,z_{0}}(s)&\;=
\alpha^{-1}+\lim_{z\to z_{0}}(G^{\Gamma}_{s}(z,z_{0})-\tfrac{1}{2}G^{\Gamma}_{t}(z,z_{0})-\tfrac{1}{2}G^{\Gamma}_{\bar{t}}(z,z_{0}))\\
&\;=\alpha^{-1}+|\scrI|\psi(s)-|\scrI|\Re\psi(t)+\sum_{\gamma\in\Gamma\backslash\scrI}
\lbrace G_{s}(z_{0},\gamma z_{0})-\Re G_{t}(z_{0},\gamma z_{0})\rbrace
\end{split}
\end{equation}
and we let
\begin{equation}
c(t)=|\scrI|\Re\psi(t)+\Re\sum_{\gamma\in\Gamma\backslash\scrI}G_{t}(z_{0},\gamma z_{0}).
\end{equation}
At this point we choose to reparameterise the coupling constant $\alpha$ according to
\begin{equation}
\alpha^{-1}-c(t)=\beta^{-1}
\end{equation}
or
\begin{equation}
\beta=\frac{\alpha}{1-\alpha c(t)}.
\end{equation}
We obtain the expression
\begin{equation}
S_{\alpha,z_{0}}(s)=\beta^{-1}+|\scrI|\psi(s)+\sum_{\gamma\in\Gamma\backslash\scrI}G_{s}(z_{0},\gamma z_{0})
\end{equation}
for $\Re s>1$.
\end{proof}

We have a uniform bound on the geometrical terms in the function $S^{\alpha,z_{0}}(s)$ for $\Re s>1$.
\begin{lem}\label{Greenbound}
For all $\rho\in\CC$ with $\Im\rho=-\sigma<-\tfrac{1}{2}$ we have the uniform bound
\begin{equation}
\sum_{\gamma\in\Gamma\backslash\scrI}\left|G_{\tfrac{1}{2}+\i\rho}\left(z_{0},\gamma z_{0}\right)\right|<\!\!<_{\Gamma,z_{0}}\sigma^{-1/2}
\end{equation}
\end{lem}
\begin{proof}
Let $\tau_{0}=\inf\lbrace d(\gamma z_{0},z_{0})\mid\gamma\in\Gamma\backslash\scrI\rbrace$. Discreteness of $\Gamma$ implies $\tau_{0}>0$. We will make use of the integral representation of the free Green's function
\begin{equation}\label{intrep}
G_{\tfrac{1}{2}+\i\rho}(z,w)=-\frac{1}{2\pi\sqrt{2}}\int_{d(z,w)}^{\infty}\frac{\e^{-\i\rho t}dt}{\sqrt{\cosh t-\cosh d(z,w)}}
\end{equation}
which is valid for $\Im\rho<-\tfrac{1}{2}$. We have
\begin{equation}\label{exactbound}
\begin{split}
\sum_{\gamma\in\Gamma\backslash\scrI}\left|G_{\tfrac{1}{2}+\i\rho}(z_{0},\gamma z_{0})\right|
&\leq\frac{1}{2\pi\sqrt{2}}\sum_{\gamma\in\Gamma\backslash\scrI}\int_{\tau_{\gamma}}^{\infty}\frac{\e^{-\sigma t}dt}{\sqrt{\cosh t-\cosh\tau_{\gamma}}}\\
&=\frac{1}{2\pi\sqrt{2}}\sum_{\gamma\in\Gamma\backslash\scrI}\e^{-\sigma\tau_{\gamma}}\int_{0}^{\infty}\frac{\e^{-\sigma t}dt}{\sqrt{\cosh (t+\tau_{\gamma})-\cosh\tau_{\gamma}}}\\
&\leq\,C_{\epsilon}\frac{1}{2\pi\sqrt{2\sinh\tau_{0}}}\int_{0}^{\infty}\frac{\e^{-\sigma t}dt}{\sqrt{t}}\\
&=\frac{C_{\epsilon}\,\sigma^{-1/2}}{2\pi\sqrt{4\pi\sinh\tau_{0}}}=C(\Gamma,z_{0})\,\sigma^{-1/2}
\end{split}
\end{equation}
since for $t>0$ and any $\gamma\in\Gamma\backslash\scrI$
\begin{equation}
\sinh\tau_{0}\leq\sinh\tau_{\gamma}\leq\frac{\cosh(\tau_{\gamma}+t)-\cosh\tau_{\gamma}}{t}
\end{equation}
and where $-\sigma<\epsilon<-1$ such that
\begin{equation}
\sum_{\gamma\in\Gamma\backslash\scrI}\e^{-\sigma\tau_{\gamma}}\leq\sum_{\gamma\in\Gamma\backslash\scrI}\e^{\epsilon\tau_{\gamma}}
\leq C_{\epsilon}
\end{equation}
where for $r<\frac{1}{2}\tau_{0}$, cf. Lemma 5 in \cite{Mf}, p. 19,
\begin{equation}
C_{\epsilon}=\frac{2\pi\e^{-2\pi\epsilon r}}{\Area(r)}\int_{0}^{\infty}\e^{\epsilon\tau}\sinh\tau d\tau.
\end{equation}
\end{proof}

$S_{\alpha,z_{0}}(s)$ contains all information about the unperturbed and perturbed discrete spectrum, as well as unperturbed and perturbed resonances in form of its poles and zeros. The following Theorem locates those and gives their spectral interpretation.

\begin{thm}\label{prop}
$S_{\alpha,z_{0}}(s)$ has the following zeros and poles.\\
\begin{itemize}
\item[(a)] Suppose $\Gamma\backslash\HH$ is compact.
\begin{itemize}
\item[(i)] There are simple poles at $\tfrac{1}{2}+\i\rho_{j}$ and $\tfrac{1}{2}-\i\rho_{j}$ corresponding to eigenvalues $\lambda_{j}=\tfrac{1}{4}+\rho_{j}^{2}$, $\rho_{j}\in\RR\cup\i\RR$.
\item[(ii)] There are simple zeros at $\tfrac{1}{2}+\i\rho^{\alpha}_{j}$ and $\tfrac{1}{2}-\i\rho^{\alpha}_{j}$, located in between the poles above, corresponding to perturbed eigenvalues $\lambda^{\alpha}_{j}=\tfrac{1}{4}+{\rho^{\alpha}_{j}}^{2}$, $\rho^{\alpha}_{j}\in\RR\cup\i\RR$. In particular one has for the lowest perturbed eigenvalue $|\Im\rho^{\alpha}_{0}|>\tfrac{1}{2}$.\\
\end{itemize}
\item[(b)] Suppose $\Gamma\backslash\HH$ has one cusp.
\begin{itemize}
\item[(i)] $\Re s>1$: There is a zero of order one at $\tfrac{1}{2}+(\tfrac{1}{4}-\lambda_{-M}^{\alpha})^{1/2}$ corresponding to the lowest perturbed eigenvalue $\lambda_{-M}^{\alpha}<0$.
\item[(ii)] $\tfrac{1}{2}<\Re s\leq1$: We have zeros of order one and simple poles in $\left(\tfrac{1}{2},1\right]$ corresponding to perturbed and unperturbed eigenvalues $<\tfrac{1}{4}$.
\item[(iii)] $\Re s=\tfrac{1}{2}$: There are simple poles corresponding to unperturbed eigenvalues $\geq\tfrac{1}{4}$ and a pole at $s=\tfrac{1}{2}$ which is of order 2 if $\lambda=\tfrac{1}{4}$ is in the discrete spectrum, otherwise it is simple. We also have zeros of order one corresponding to solutions of $E(z_{0},\tfrac{1}{2}+\i t)=0$ $\wedge$ $\Re\left\{S^{\alpha,z_{0}}(\tfrac{1}{2}+\i t)\right\}=0$, i. e. new eigenvalues.
\item[(iv)] $\Re s<\tfrac{1}{2}$: We have zeros corresponding to perturbed resonances and poles corresponding to unperturbed resonances. There are also simple poles in $[0,\tfrac{1}{2})$ which correspond to the small eigenvalues as in $(ii)$.
\end{itemize}
\end{itemize}
\end{thm}
\begin{proof}
(a) We have from \eqref{specexp}, for $t(1-t)=\tfrac{1}{4}+\rho(t)^{2}$,
\begin{equation}
S_{\alpha,z_{0}}(\tfrac{1}{2}+\i\rho)=\alpha^{-1}+\sum_{j=0}^{\infty}|\varphi(z_{0})|^{2}\left\{\frac{1}{\rho_{j}^{2}-\rho^{2}}-\Re\left\{\frac{1}{\rho_{j}^{2}-\rho(t)^{2}}\right\}\right\}
\end{equation}
and, for $\rho\in\RR\cup\i\RR$, and depending on whether $\rho=v$ or $\rho=\i v$, for $v\in\RR$,
\begin{equation}
\frac{d}{dv}S_{\alpha,z_{0}}(\tfrac{1}{2}+\i\rho(v))=\pm v\sum_{j=-M}^{\infty}\frac{|\varphi_{j}(z_{0})|^{2}}{(\rho_{j}^{2}\pm v^{2})^{2}}
\end{equation}
which shows that the zeros of $S_{\alpha,z_{0}}(s)$ lie in between the poles on the critical line and the real line. In particular one has $|\Im\rho^{\alpha}_{-M}|>|\Im\rho_{-M}|=\tfrac{1}{2}$.

(b) It can be seen from its meromorphic continuation that the regularised automorphic Green's function $(G^{\Gamma}_{s}-G^{\Gamma}_{t})(z,w)$ has simple poles corresponding to the eigenvalues of the Laplacian on the critical line and the interval $\left[0,1\right]$, as well as simple poles corresponding to the resonances of the Laplacian in $\Re s<\tfrac{1}{2}$. If $\lambda=\tfrac{1}{4}$ is an eigenvalue then it follows that $\G^{\Gamma}_{s}(z,w)$ has a pole of order 2 at $s=\tfrac{1}{2}$ and otherwise a simple pole (cf. \cite{Iw}, p.124-125, Thm. 9.2.). The function $S_{\alpha,\,z_{0}}(s)$ inherits these poles. See for instance Thm. 3.5, p. 250, in \cite{Hj3}.

We proceed with locating the zeros of $S_{\alpha,z_{0}}(s)$. The proof of Proposition \ref{eigencond} shows that the zeros of $S_{\alpha,z_{0}}(s)$ correspond to eigenfunctions of $\Delta_{\alpha,z_{0}}$. If $\Re s\geq\tfrac{1}{2}$ the zeros correspond to new eigenvalues and the corresponding eigenfunctions are in $L^{2}(\Gamma\backslash\HH)$ and so self-adjointness of $\Delta_{\alpha,z_{0}}$ rules out any zeros $s$ in $\Re s\geq\tfrac{1}{2}$, $\Im s\neq0$. However, if $\Re s<\tfrac{1}{2}$ the zeros correspond to resonances and generalised non-$L^{2}$ eigenfunctions which may have $\Im s\neq0$.

Let us consider the critical line. We recall that $S_{\alpha,z_{0}}(s)$ satisfies the functional equation
\begin{equation}\label{FE}
S_{\alpha,z_{0}}(s)=S_{\alpha,z_{0}}(1-s)-\frac{E(z_{0},s)E(z_{0},1-s)}{1-2s}.
\end{equation}
Now letting $s=\tfrac{1}{2}+\i t$, $t\in\RR$, we compute the imaginary part of $S_{\alpha,\,z_{0}}\left(\tfrac{1}{2}+\i t\right)$
\begin{equation}
\Im\left\{S_{\alpha,\,z_{0}}(\tfrac{1}{2}+\i t)\right\}
=\frac{1}{2\i}\left\{S^{\alpha,\,z_{0}}(\tfrac{1}{2}+\i t)-S_{\alpha,\,z_{0}}(\tfrac{1}{2}-\i t)\right\}
=\frac{\left|E(z_{0},\tfrac{1}{2}+\i t)\right|^{2}}{4t}
\end{equation}
and write
\begin{equation}
S_{\alpha,\,z_{0}}(\tfrac{1}{2}+\i t)=\Re\left\{S_{\alpha,\,z_{0}}(\tfrac{1}{2}+\i t)\right\}+\i\frac{\left|E(z_{0},\tfrac{1}{2}+\i t)\right|^{2}}{4t}
\end{equation}
so, since $\alpha$ is real, $S_{\alpha,\,z_{0}}(s)$ has a zero $\tfrac{1}{2}+\i t$ on the critical line if and only if
\begin{equation}\label{speccond}
E(z_{0},\tfrac{1}{2}+\i t)=0\,\wedge\,\Re\left\{S_{\alpha,z_{0}}(\tfrac{1}{2}+\i t)\right\}=0.
\end{equation}
Let $\rho_{\alpha,z_{0}}$ be a solution to the above equation. We recall the spectral expansion of the automorphic Green's function for $\Re s>\tfrac{1}{2}$
\begin{equation}
G^{\Gamma}_{\frac{1}{2}+\i\rho}(z,w)=\frac{1}{2\pi}\int_{0}^{\infty}\frac{E(z,\tfrac{1}{2}+\i\rho')E(w,\tfrac{1}{2}-\i\rho')}{\rho'^{2}-\rho^{2}}d\rho'+\sum_{j=-M}^{\infty}\frac{\varphi_{j}(z)\overline{\varphi_{j}(w)}}{\rho_{j}^{2}-\rho^{2}},
\end{equation}
where $\lbrace\varphi_{j}\rbrace_{j=-M}^{\infty}$ is an orthonormal basis of Maass forms. Since $E(z_{0},\tfrac{1}{2}+\i\rho^{\alpha,z_{0}})=0$, we have
\begin{equation}
\frac{d}{d\rho}|_{\rho=\rho_{\alpha,z_{0}}}\Re\left\{S_{\alpha,\,z_{0}}(\tfrac{1}{2}+\i\rho)\right\}=\frac{\rho_{\alpha,z_{0}}}{2\pi}\int_{-\infty}^{\infty}\frac{|E(z_{0},\tfrac{1}{2}+\i\rho')|^{2}}{(\rho'^{2}-\rho_{\alpha,z_{0}}^{2})^{2}}d\rho'
+2\rho_{\alpha,z_{0}}\sum_{j=-M}^{\infty}\frac{|\varphi_{j}(z_{0})|^{2}}{(\rho_{j}^{2}-\rho_{\alpha,z_{0}}^{2})^{2}},
\end{equation}
which shows that $\rho_{\alpha,z_{0}}$ must be a zero of order one.

Now let us look at zeros on the real line. Let $\rho=-\i v$, $v>0$ to ensure $s=\tfrac{1}{2}+v>\tfrac{1}{2}$. We obtain for the derivative of $S_{\alpha,\,z_{0}}(\frac{1}{2}+v)$ with respect to $v$
\begin{equation}
\frac{d}{dv}S_{\alpha,\,z_{0}}(\frac{1}{2}+v)=-\frac{v}{\pi}\int_{0}^{\infty}\frac{\left|E(z_{0},\tfrac{1}{2}+\i\rho')\right|^{2}}{(\rho'^{2}+v^{2})^{2}}d\rho'-2v\sum_{j=-M}^{\infty}\frac{\left|\varphi_{j}(z_{0})\right|^{2}}{(\rho_{j}^{2}+v^{2})^{2}}.
\end{equation}
Since $S_{\alpha,\,z_{0}}(\frac{1}{2}+v)$ is a real function on the real line, has poles in the interval $[0,\tfrac{1}{2}]$ and is monotonic in between these poles, we conclude that it must take its zeros in between these poles. Since $\lim_{v\to+\infty}$ $|S(\tfrac{1}{2}+v)|=\infty$ there is a zero of order one at $\rho^{\alpha}_{-M}>\rho_{-M}$. Because of monotonicity of $S_{\alpha,z_{0}}(\tfrac{1}{2}+v)$ on the half-line $(\tfrac{1}{2},\infty)$ we conclude that there are only finitely many zeros in $\Re s>\tfrac{1}{2}$. They correspond to perturbed eigenvalues less than $\tfrac{1}{4}$. As we will see later the corresponding eigenfunctions, which are automorphic Green's functions, are residues of the perturbed Eisenstein series unless $E(z_{0},\tfrac{1}{2}+v)=0$.

With regard to zeros in $(-\infty,\tfrac{1}{2})$ no such argument works, since the spectral expansion is not valid. Such zeros correspond to perturbed resonances if $E(z_{0},s)\neq0$ since the associated Green's functions fail to be in $L^{2}(\Gamma\backslash\HH)$. If $E(z_{0},s)=0$ the zero corresponds to a zero $1-s\in(\tfrac{1}{2},\infty)$ since the functional equation for $S_{\alpha,z_{0}}(s)$ implies $S_{\alpha,z_{0}}(s)=S_{\alpha,z_{0}}(1-s)$.
\end{proof}

For the remainder of this section we will be concerned with noncompact $\Gamma\backslash\HH$.
\begin{remark}
It is an interesting question to ask how the zeros of $S_{\alpha,z_{0}}(s)$, which correspond to eigenvalues or resonances, behave under continuous perturbation of the pair $(\alpha,z_{0})\in(\RR\backslash\lbrace0\rbrace)\times\Gamma\backslash\HH$. In view of the Phillips-Sarnak conjecture for smooth perturbations of the metric (cf. \cite{PhSa}) the case of eigenvalues of pseudo cusp forms which correspond to zeros on the critical line is particularly interesting. One can see (but we will not give a proof) from the dependence of $S_{\alpha,z_{0}}(s)$ on $(\alpha,z_{0})$ that under continous perturbations of $(\alpha,z_{0})$ in $(\RR\backslash\lbrace0\rbrace)\times\Gamma\backslash\HH$ the zeros of $S_{\alpha,z_{0}}(s)$ depend continously on $(\alpha,z_{0})$. Because of self-adjointness of $\Delta_{\alpha,z_{0}}$ and the condition $\Re S_{\alpha,z_{0}}(\tfrac{1}{2}+\i t)=0\;\wedge\;E(z_{0},\tfrac{1}{2}+\i t)=0$ for the existence of pseudo cusp forms the corresponding zeros will move continuously into the halfplane $\Re s<\tfrac{1}{2}$ under continuous perturbations about non-generic points in $(\RR\backslash\lbrace0\rbrace)\times\Gamma\backslash\HH$ for which pseudo cusp forms exist. So the associated eigenvalues turn into resonances.
\end{remark}

Colin de Verdiere observes in \cite{CdV} that for almost all pairs $(\alpha,z_{0})\in(\RR\backslash\lbrace0\rbrace)\times\Gamma\backslash\HH$ the new cuspidal part of the perturbed discrete spectrum is empty. We easily confirm this from the particular form of the relative zeta function $S_{\alpha,z_{0}}(s)$ on the critical line.
\begin{prop}\label{finite}
The equation
\begin{equation}\label{cond}
E(z_{0},\tfrac{1}{2}+\i r)=0\,\wedge\, \Re\left\{S_{\alpha,z_{0}}(\tfrac{1}{2}+\i r)\right\}=0,\qquad r\in\RR,
\end{equation}
has no solutions for almost all $(\alpha,z_{0})\in(\RR\backslash\lbrace0\rbrace)\times\Gamma\backslash\HH$.
\end{prop}
\begin{proof}
Fix $z_{0}\in\Gamma\backslash\HH$. Let 
\begin{equation}
S_{z_{0}}=\lbrace r\in\RR\mid E(z_{0},\tfrac{1}{2}+\i r)=0\rbrace
\end{equation}
which is countable since $E(z_{0},s)$ is meromorphic in $s$. Also let
\begin{equation}
g(r)=\Re\lim_{z\to z_{0}}\lbrace G^{\Gamma}_{1/2+\i r}-\tfrac{1}{2}G^{\Gamma}_{t}-\tfrac{1}{2}G^{\Gamma}_{\bar{t}}\rbrace(z,z_{0}).
\end{equation}
We define
\begin{equation}
A_{z_{0}}=\lbrace-g(r)^{-1}\mid r\in S_{z_{0}}\wedge g(r)\neq0\rbrace,
\end{equation}
and claim that for any $\alpha\notin A_{z_{0}}\cup\lbrace0\rbrace$ equation \eqref{cond} has no solution. To see this suppose the contrary. So there exists $r_{0}\in\RR$ such that $E(z_{0},\tfrac{1}{2}+\i r_{0})=0$ and $g(r_{0})=-\alpha^{-1}\neq0$ which leads to a contradiction since $\alpha\notin A_{z_{0}}$. The result follows since $A_{z_{0}}$ is countable.
\end{proof}

\begin{remark}
It seems very likely that even for $\alpha\in A_{z_{0}}$ the new part of the perturbed cuspidal spectrum is always finite. This is motivated from the intuition that $g$ cannot map an infinite subset of $A_{z_{0}}$ to $\lbrace-\alpha^{-1}\rbrace\subset\RR\backslash\lbrace0\rbrace$. We point out the analogy with the Phillips-Sarnak-Conjecture \cite{PhSa} for smooth perturbations of the metric of a hyperbolic surface.
\end{remark}

\begin{remark}
In what follows we will denote the small unperturbed eigenvalues by $\lbrace\rho_{j}\rbrace_{-M}^{j=-1}$, the large unperturbed eigenvalues by $\lbrace\rho_{j}\rbrace_{j=0}^{\infty}$. Similarly we write for the small and large perturbed eigenvalues $\lbrace\rho^{\alpha}_{j}\rbrace_{-M}^{j=-1}$ and $\lbrace\rho^{\alpha}_{j}\rbrace_{j}$. The unperturbed resonances we denote by $\lbrace r_{j}\rbrace_{j=0}^{\infty}$ and the perturbed resonances by $\lbrace r^{\alpha}_{j}\rbrace_{j=0}^{\infty}$.
\end{remark}

\section{The trace formula for compact surfaces}

In this section we will give the proof of Theorem \ref{thm1}. Throughout the proof we will count the perturbed and unperturbed eigenvalues ignoring degeneracies to keep notation simple. Since in Theorem \ref{thm1} only the difference of perturbed and unperturbed trace appears it does not matter whether one counts degeneracies (cf. section 2.3. which explains this in the noncompact case and holds analogously in the compact case).

Let $\Gamma\backslash\HH$ be compact. We first prove a truncated trace formula. Recall
\begin{equation}
S_{\alpha,z_{0}}(s)=\alpha^{-1}+\lim_{z\to z_{0}}\lbrace G_{s}^{\Gamma}(z,z_{0})-\Re G_{t}^{\Gamma}(z,z_{0})\rbrace.
\end{equation}

\begin{prop}
Let $h\in H_{\sigma,\delta}$ and $T>0$. Define 
\begin{equation}
B(T)=\lbrace\rho\in\CC\mid|\Im\rho|<\sigma,\;|\Re\rho|<T\rbrace.
\end{equation} 
Then 
\begin{equation}\label{trunc}
\begin{split}
\sum_{\rho_{j}^{\alpha}\in B(T)}h(\rho^{\alpha}_{j})-\sum_{\rho_{j}\in B(T)}h(\rho_{j})=&\frac{1}{2\pi\i}\int_{-\i\sigma-T}^{-\i\sigma+T}h(\rho)\frac{S'_{\alpha,z_{0}}}{S_{\alpha,z_{0}}}(\tfrac{1}{2}+\i\rho)d\rho\\
&+\frac{1}{2\pi\i}\int_{-\i\sigma+T}^{\i\sigma+T}h(\rho)\frac{S'_{\alpha,z_{0}}}{S_{\alpha,z_{0}}}(\tfrac{1}{2}+\i\rho)d\rho.
\end{split}
\end{equation}
\end{prop}
\begin{proof}
Let $t=\tfrac{1}{2}+\i\xi$. Since $\Gamma\backslash\HH$ is compact, we have for an orthonormal basis of eigenfunctions of the Laplacian $\lbrace\varphi_{j}\rbrace_{j=0}^{\infty}$ the spectral expansion
\begin{equation}
S_{\alpha,z_{0}}(\tfrac{1}{2}+\i\rho)=\alpha^{-1}+\sum_{j=-M}^{\infty}|\varphi_{j}(z_{0})|^{2}\left\{\frac{1}{\rho_{j}^{2}-\rho^{2}}-\Re\left\{\frac{1}{\rho_{j}^{2}-\xi^{2}}\right\}\right\}.
\end{equation}
We obtain the result by contour integration along $\partial B(T)$.
\end{proof}

In order to prove the full trace formula we must show
\begin{equation}
\lim_{T\to\infty}\int_{-\i\sigma+T}^{\i\sigma+T}h(\rho)\frac{S'_{\alpha,z_{0}}}{S_{\alpha,z_{0}}}(\tfrac{1}{2}+\i\rho)d\rho=0
\end{equation}
for $T$ such that $\partial B(T)$ does not contain any zeros or poles of $S_{\alpha,z_{0}}$. 

We require the following bound.
\begin{prop}\label{compbound}
There exists an increasing sequence $\lbrace T_{n}\rbrace_{n=0}^{\infty}$, such that $\lim_{n\to\infty}T_{n}=+\infty$, and for any $\epsilon>0$
\begin{equation}
\int_{-\i\sigma+T_{n}}^{\i\sigma+T_{n}}\big|\log|S_{\alpha,z_{0}}(\tfrac{1}{2}+\i\rho)|\big||d\rho|<\!\!<T_{n}^{2+\epsilon}.
\end{equation}
\end{prop}

In order to prove Proposition \ref{compbound} we require a bound on the relative zeta function on a sequence of intervals $[T_{N}-\i\sigma,T_{N}+\i\sigma]$ crossing the strip $|\Im\rho|\leq\sigma$. In fact $S_{\alpha,z_{0}}(\tfrac{1}{2}+\i\rho)$ admits a uniform bound of polynomial growth for a suitably chosen sequence of intervals.
\begin{prop}\label{polybound}
Let $\eta\in\CC$. Then $\exists\lbrace T_{N}\rbrace_{N}$ in $\RR_{+}$, $\lim_{N} T_{N}=+\infty$ such that uniformly $\forall N,\forall t\in[-\sigma,\sigma]$
\begin{equation}\label{polyn}
\sum_{j\geq-M}|\varphi_{j}(z_{0})|^{2}\left|\frac{1}{\lambda_{j}-\mu_{N}(t)}-\frac{1}{\lambda_{j}-\eta}\right|<\!\!<_{\eta,\Gamma}T_{N}^{5}.
\end{equation}
where $\mu_{N}(t)=\tfrac{1}{4}+(T_{N}+\i t)^{2}$.
\end{prop}

Before we give the proof of Proposition \ref{polybound} we state a Lemma which will play a central role in the proof.
\begin{lem}\label{spacing}
Let $\lbrace\varphi_{j}\rbrace_{j}$ be the set of eigenfunctions on $\Gamma\backslash\HH$ with $(\Delta+\lambda_{j})\varphi_{j}=0$. Then there exists a subsequence $\lbrace\lambda_{j_{k}}\rbrace_{k=0}^{\infty}\subset\lbrace\lambda_{j}\rbrace_{j=0}^{\infty}$ and a constant $c_{0}(\Gamma)>0$ s. t. $|\lambda_{j_{k}+1}-\lambda_{j_{k}}|\geq c_{0}(\Gamma)$.
\end{lem}
\begin{proof}
We have the standard upper bound on the number of eigenvalues $\#\lbrace j\mid\lambda_{j}\leq T\rbrace\leq c(\Gamma)T$ for some constant $c(\Gamma)>0$. Let $c_{0}(\Gamma)=\tfrac{1}{3}c(\Gamma)^{-1}$. Pick any $\lambda_{n_{1}}\geq c_{0}(\Gamma)$. If $|\lambda_{n_{1}}-\lambda_{n_{1}+1}|>\lambda_{n_{1}}$ then choose $\lambda_{j_{1}}=\lambda_{n_{1}}$ to be the first member of our subsequence and proceed according to the last line of this proof. So suppose that $|\lambda_{n_{1}}-\lambda_{n_{1}+1}|\leq\lambda_{n_{1}}$. We claim that there exists $\lambda_{j_{1}}\in[\lambda_{n_{1}},2\lambda_{n_{1}}]$ s. t. $|\lambda_{j_{1}}-\lambda_{j_{1}+1}|\geq c_{0}(\Gamma) =\tfrac{1}{3}c(\Gamma)^{-1}$. To see this suppose the contrary for a contradiction. So suppose that for all $\lambda_{j}\in[\lambda_{n_{1}},2\lambda_{n_{1}}]$ we have $|\lambda_{j}-\lambda_{j+1}|<c_{0}(\Gamma)$. Let
\begin{equation}
\lambda_{max,1}=\min\lbrace\lambda_{j}\mid\lambda_{j}\geq2\lambda_{n_{1}}\rbrace.
\end{equation}
We have the estimate
\begin{equation}
\begin{split}
|\lambda_{n_{1}}-\lambda_{max,1}|\leq&\;\sum_{\lambda_{j}\in[\lambda_{n_{1}},2\lambda_{n_{1}}]}|\lambda_{j}-\lambda_{j+1}|\\
<&\;c_{0}(\Gamma)\#\lbrace j\mid\lambda_{j}\in[\lambda_{n_{1}},2\lambda_{n_{1}}]\rbrace\\
\leq&\;2c_{0}(\Gamma)c(\Gamma)\lambda_{n_{1}}=\tfrac{2}{3}\lambda_{n_{1}}
\end{split}
\end{equation}
Which is a contradiction $|\lambda_{n_{1}}-\lambda_{max,1}|\geq|2\lambda_{n_{1}}-\lambda_{n_{1}}|=\lambda_{n_{1}}$. We conclude that there must exist $\lambda_{j_{1}}\in[\lambda_{n_{1}},2\lambda_{n_{1}}]$ s. t. $|\lambda_{j_{1}}-\lambda_{j_{1}+1}|\geq c_{0}(\Gamma)$.

Pick $\lambda_{n_{2}}>2\lambda_{n_{1}}$ and proceed as above. We can continue this procedure indefinitely and in such a way construct the sequence $\lbrace\lambda_{j_{k}}\rbrace_{k=0}^{\infty}\subset\lbrace\rho_{j}\rbrace_{j=0}^{\infty}$ with the desired property.
\end{proof}

We apply Lemma \ref{spacing} to prove the Proposition.\\

\textit{Proof of Proposition \ref{polybound}.}
By Lemma \ref{spacing} we can choose an infinite sequence $T_{N}=\tfrac{1}{2}\lbrace\rho_{N}+\rho_{N+1}\rbrace$ with $T_{N}^{-1}<\!\!<|\rho_{N}-\rho_{N+1}|$. Let $\mu_{N}(t)=\tfrac{1}{4}+(T_{N}+\i t)^{2}$, $t\in[-\sigma,0]$. We have
\begin{equation}
\sum_{j=-M}^{\infty}|\varphi_{j}(z_{0})|^{2}\left|\frac{1}{\lambda_{j}-\mu_{N}(t)}-\frac{1}{\lambda_{j}-\eta}\right|
<\!\!< |\eta-\mu_{N}(t)|\sum_{j=-M}^{\infty}\frac{\lambda_{j}^{1/2}}{|\lambda_{j}-\mu_{N}(t)||\lambda_{j}-\eta|}
\end{equation}
where we have used the bound $|\varphi_{j}(z_{0})|^{2}<\!\!<\lambda_{j}^{1/2}$ (cf. \cite{Iw}, p. 108, (8.3')). Fix $\beta\in(\tfrac{1}{2},1)$. We split the sum into a central part satisfying $\inf_{t\in[-\sigma,0]}|\lambda_{j}-\mu_{N}(t)|<\lambda_{j}^{\beta}$ and a corresponding tail. For convenience we let $I_{N}(\lambda_{j})=\inf_{t\in[-\sigma,0]}|\lambda_{j}-\mu_{N}(t)|$. The first sum is estimated by
\begin{equation}
\sum_{I_{N}(\lambda_{j})<\lambda_{j}^{\beta}}\frac{\lambda_{j}^{1/2}}{|\lambda_{j}-\mu_{N}(t)||\lambda_{j}-\eta|}
\leq\#\lbrace j\mid I_{N}(\lambda_{j})<\lambda_{j}^{\beta}\rbrace\, \max_{I_{N}(\lambda_{j})<\lambda_{j}^{\beta}}\,\sup_{t\in[-\sigma,0]}\left\{\frac{\lambda_{j}^{1/2}}{|\lambda_{j}-\mu_{N}(t)||\lambda_{j}-\eta|}\right\}.
\end{equation}
Now if $\lambda_{j}>\tfrac{1}{4}+T_{N}^{2}$ then $I_{N}(\lambda_{j})=\lambda_{j}-\tfrac{1}{4}-T_{N}^{2}$. It follows
\begin{equation}
\begin{split}
\#\lbrace j\mid I_{N}(\lambda_{j})<\lambda_{j}^{\beta}\rbrace
\leq&\#\lbrace j\mid\lambda_{j}\leq\tfrac{1}{4}+T_{N}^{2}\rbrace
+\#\lbrace j\mid \lambda_{j}-\lambda_{j}^{\beta}<\tfrac{1}{4}+T_{N}^{2}\rbrace
\end{split}
\end{equation}
Let
\begin{equation}
C(\beta)=\#\lbrace j\mid\lambda_{j}\leq 2^{1/(1-\beta)}\rbrace
\end{equation}
and observe that $\lambda_{j}>2^{1/(1-\beta)}$ implies $\lambda_{j}^{\beta-1}<\tfrac{1}{2}$. So $\lambda_{j}>2^{1/(1-\beta)}$ together with $\lambda_{j}(1-\lambda_{j}^{\beta-1})<\tfrac{1}{4}+T_{N}^{2}$ implies
\begin{equation}
\lambda_{j}<2\lambda_{j}(1-\lambda_{j}^{\beta-1})<\tfrac{1}{2}+2T_{N}^{2}.
\end{equation}
Hence
\begin{equation}
\begin{split}
\#\lbrace j\mid\lambda_{j}(1-\lambda_{j}^{\beta-1})<\tfrac{1}{4}+T_{N}^{2}\rbrace
\leq&\,\#\lbrace j\mid\lambda_{j}\leq2^{1/(1-\beta)}\wedge\lambda_{j}(1-\lambda_{j}^{\beta-1})<\tfrac{1}{4}+T_{N}^{2}\rbrace\\
&+\#\lbrace j\mid\lambda_{j}>2^{1/(1-\beta)}\wedge\lambda_{j}(1-\lambda_{j}^{\beta-1})<\tfrac{1}{4}+T_{N}^{2}\rbrace\\
\leq&\,C(\beta)+\#\lbrace j\mid2^{1/(1-\beta)}<\lambda_{j}<\tfrac{1}{2}+2T_{N}^{2}\rbrace\\
\leq&\,C(\beta)+\tfrac{1}{2}c+2cT_{N}^{2}.
\end{split}
\end{equation}
It follows that
\begin{equation}
\#\lbrace j\mid I_{N}(\lambda_{j})<\lambda_{j}^{\beta}\rbrace<\!\!<_{\beta}T_{N}^{2}.
\end{equation}
By the same observations as above we see that $I(\lambda_{j})<\lambda_{j}^{\beta}$ implies $\lambda_{j}\leq\max\lbrace2^{1/(1-\beta)},\tfrac{1}{2}+2T_{N}^{2}\rbrace$. Also for any $j\geq0$ we have $|\rho_{j}-T_{N}|\geq\tfrac{1}{2}|\rho_{N}-\rho_{N+1}|>\!\!>T_{N}^{-1}$ which implies
\begin{equation}
\begin{split}
|\lambda_{j}-\mu_{N}(t)|=|\rho_{j}^{2}-(T_{N}+\i t)^{2}|
=&\;|\rho_{j}-T_{N}-\i t||\rho_{j}+T_{N}+\i t|\\
\geq&\;|\rho_{j}-T_{N}|(\rho_{j}+T_{N})\\
>\!\!>&\;1.
\end{split}
\end{equation}
Since $|\lambda_{j}-\eta|\geq|\Im\eta|>0$ we have
\begin{equation}
\max_{I_{N}(\lambda_{j})<\lambda_{j}^{\beta}}\,\sup_{t\in[-\sigma,0]}\left\{\frac{\lambda_{j}^{1/2}}{|\lambda_{j}-\mu_{N}(t)||\lambda_{j}-\eta|}\right\}
<\!\!<_{\eta}T_{N}.
\end{equation}

The tail can be bounded as follows
\begin{equation}
\begin{split}
\sum_{I(\lambda_{j})\geq\lambda_{j}^{\beta}}\frac{\lambda_{j}^{1/2}}{|\lambda_{j}-\mu_{N}(t)||\lambda_{j}-\eta|}
\leq&\sum_{I(\lambda_{j})\geq\lambda_{j}^{\beta}}\frac{\lambda_{j}^{1/2-\beta}}{|\lambda_{j}-\eta|}\\
\leq&\sum_{j=0}^{\infty}\frac{\lambda_{j}^{1/2-\beta}}{|\lambda_{j}-\eta|}<+\infty.
\end{split}
\end{equation}
Finally note that $|\mu_{N}(t)-\eta|<\!\!<_{\eta}T_{N}^{2}$.
\begin{flushright}
$\square$
\end{flushright}

The following Lemma establishes the existence of a test function in the space $H_{\sigma,\delta}$ with certain properties which we will use in the proof of Proposition \ref{compbound} as well as in the proof of the noncompact case. 
\begin{lem}\label{testf}
Let $\epsilon>0$. There exists $h_{\epsilon}\in H_{\sigma,\delta}$ such that
\begin{equation}\label{sym2}
\overline{h_{\epsilon}(\rho)}=h_{\epsilon}(\overline{\rho})
\end{equation}
and for some subsequence $\lbrace T_{N(j)}\rbrace_{j}\subset\lbrace T_{N}\rbrace_{N}$, $\lim_{j\to\infty}T_{N(j)}=\infty$ and $\rho\in[T_{N(j)},T_{N(j)}-\i\sigma]$
\begin{equation}
|\Re h'_{\epsilon}(\rho)|>\!\!>T_{N(j)}^{-2-\epsilon}
\end{equation}
uniformly in $j$ and $\rho$.
\end{lem}
\begin{proof}
Since $\lim_{T\to\infty}T_{N}=\infty$, we can pick a subsequence $\lbrace T_{N(k)}\rbrace_{k=1}^{\infty}\subset\lbrace T_{N}\rbrace_{N=1}^{\infty}$ such that $2^{n(k)-1}\leq T_{N(k)}+1\leq 2^{n(k)}$ for some integer $n(k)\geq1$ with $n(k+1)>n(k)+1$ for all $k\geq1$. Now let $\sigma_{0}>\sigma$ and consider the test function
\begin{equation}
\begin{split}
h_{\epsilon}(\rho)=\sum_{k=1}^{\infty}2^{-n(k)(2+\epsilon)}\bigg\{
& \frac{1}{(\rho-T_{N(k)}-1-\i\sigma_{0})^{4}}+\frac{1}{(\rho-T_{N(k)}-1+\i\sigma_{0})^{4}}\\
& +\frac{1}{(\rho+T_{N(k)}+1+\i\sigma_{0})^{4}}+\frac{1}{(\rho+T_{N(k)}+1-\i\sigma_{0})^{4}}\bigg\}.
\end{split}
\end{equation}
By construction the property \eqref{sym2} is fulfilled. $h_{\epsilon}$ is even and analytic in the strip $\Im\rho\leq\sigma$. We will show that it also satisfies $h_{\epsilon}(\rho)<\!\!<(1+|\Re\rho|)^{-2-\epsilon}$ uniformly in the strip $\Im\rho\leq\sigma$. Because of evenness we only have to prove the bound for $\Re\rho>0$. We have
\begin{equation}
|h_{\epsilon}(\rho)|\leq2\sum_{k=1}^{\infty}2^{-n(k)(2+\epsilon)}\left\{\frac{1}{|\rho-T_{N(k)}-1+\i\sigma_{0}|^{4}}+\frac{1}{|\rho+T_{N(k)}+1+\i\sigma_{0}|^{4}}\right\}.
\end{equation}
We will estimate the two parts separately. For the second sum we have
\begin{equation}
\begin{split}
\sum_{k=1}^{\infty}2^{-n(k)(2+\epsilon)}\frac{1}{|\rho+T_{N(k)}+1+\i\sigma_{0}|^{4}}
&\leq\sum_{k=1}^{\infty}2^{-n(k)(2+\epsilon)}\frac{1}{(\Re\rho+T_{N(k)}+1)^{4}}\\
&\leq|\Re\rho|^{-4}\sum_{k=1}^{\infty}2^{-k(2+\epsilon)}.
\end{split}
\end{equation}
The first sum is more difficult to estimate. Since $h_{\epsilon}$ is analytic in the strip $\Im\rho\leq\sigma$ it suffices to prove the bound for $\Re\rho\geq1$. So there exist and integer $j\geq1$ such that $2^{j-1}\leq\Re\rho\leq2^{j}$. We obtain
\begin{equation}
\begin{split}
\sum_{k=1}^{\infty}2^{-n(k)(2+\epsilon)}\frac{1}{|\rho-T_{N(k)}-1+\i\sigma_{0}|^{4}}
\leq&\sum_{n(k)\neq j-1,j,j+1}2^{-n(k)(2+\epsilon)}\frac{1}{|\Re\rho-T_{N(k)}-1|^{4}}\\
&+\frac{2^{(-j+1)(2+\epsilon)}+2^{-j(2+\epsilon)}+2^{-(j+1)(2+\epsilon)}}{|\Im\rho+\sigma_{0}|^{4}}.
\end{split}
\end{equation}
Now, since $\Re\rho\geq2^{j-1}$ and $|\Im\rho|\leq\sigma<\sigma_{0}$,
\be
\frac{2^{(-j+1)(2+\epsilon)}+2^{-j(2+\epsilon)}+2^{-(j+1)(2+\epsilon)}}{(\Im\rho+\sigma_{0})^{4}}
\leq\frac{3}{(\sigma_{0}-\sigma)^{4}}|\Re\rho|^{-2-\epsilon}.
\ee
Next we estimate the sum. First observe that for $n(k)\leq j-2$
\be
|T_{N(k)}+1-\Re\rho|\geq2^{j-1}-2^{n(k)}=2^{j-1}(1-2^{n(k)-j+1})\geq2^{j-2}
\ee
and for $n(k)\geq j+2$
\be
|T_{N(k)}+1-\Re\rho|\geq2^{n(k)-1}-2^{j}=2^{j}(2^{n(k)-j-1}-1)\geq2^{j},
\ee
which implies
\begin{equation}\label{estrem}
\begin{split}
\sum_{k\neq j-1,j,j+1}2^{-n(k)(2+\epsilon)}\frac{1}{|\Re\rho-T_{N(k)}-1|^{4}}
\leq &2^{8-4j}\sum_{k\neq j-1,j,j+1}2^{-n(k)(2+\epsilon)}\\
\leq &2^{8-4j}\sum_{k=1}^{\infty}2^{-k(2+\epsilon)}\\
\leq &2^{4}|\Re\rho|^{-4}\sum_{k=1}^{\infty}2^{-k(2+\epsilon)}.
\end{split}
\end{equation}
Next we prove that for $\rho\in[T_{N(j)},T_{N(j)}-\sigma\i]$ we have the lower bound $\Re h_{\epsilon}'(\rho)>\!\!>T_{N(j)}^{-2-\epsilon}$ uniformly in $j$ and $\rho$. We have
\begin{equation}\label{testder}
\begin{split}
h'_{\epsilon}(\rho)=-4\sum_{k=1}^{\infty}2^{-n(k)(2+\epsilon)}\bigg\{
& \frac{1}{(\rho-T_{N(k)}-1-\i\sigma_{0})^{5}}+\frac{1}{(\rho-T_{N(k)}-1+\i\sigma_{0})^{5}}\\
& +\frac{1}{(\rho+T_{N(k)}+1+\i\sigma_{0})^{5}}+\frac{1}{(\rho+T_{N(k)}+1-\i\sigma_{0})^{5}}\bigg\}.
\end{split}
\end{equation}
Since
\begin{equation}\label{realp}
\begin{split}
\Re\left[\frac{1}{(\rho-T_{N(k)}-1-\i\sigma_{0})^{5}}\right]
=&\frac{(\Re\rho-T_{N(k)}-1)^{5}}{((\Re\rho-T_{N(k)}-1)^{2}+(\Im\rho-\sigma_{0})^{2})^{5}}\\
&+\frac{10(\Re\rho-T_{N(k)}-1)^{3}(\Im\rho-\sigma_{0})^{2}}{((\Re\rho-T_{N(k)}-1)^{2}+(\Im\rho-\sigma_{0})^{2})^{5}}\\
&+\frac{5(\Re\rho-T_{N(k)}-1)(\Im\rho-\sigma_{0})^{4}}{((\Re\rho-T_{N(k)}-1)^{2}+(\Im\rho-\sigma_{0})^{2})^{5}}
\end{split}
\end{equation}
we have for $\rho\in[T_{N(j)},T_{N(j)-\i\sigma}]$
\begin{equation}
\begin{split}
&\sum_{k=1}^{\infty}2^{-n(k)(2+\epsilon)}\Re\left[\frac{1}{(\rho-T_{N(k)}-1-\i\sigma_{0})^{5}}\right]\\
=\;&-2^{-n(j)(2+\epsilon)}\frac{1+10(\Im\rho-\sigma_{0})^{2}+5(\Im\rho-\sigma_{0})^{4}}{(1+(\Im\rho-\sigma_{0})^{2})^{5}}\\
&+\sum_{k\neq j}2^{-n(k)(2+\epsilon)}\Re\left[\frac{1}{(\rho-T_{N(k)}-1-\i\sigma_{0})^{5}}\right].
\end{split}
\end{equation}
Now it follows from \eqref{realp} and $\Re\rho=T_{N(j)}$ that
\begin{equation}
\left|\Re\left[\frac{1}{(\rho-T_{N(k)}-1-\i\sigma_{0})^{5}}\right]\right|<\!\!<|T_{N(j)}-T_{N(k)-1}|^{-5}
\end{equation}
and thus
\begin{equation}
\left|\sum_{k\neq j}2^{-n(k)(2+\epsilon)}\Re\left[\frac{1}{(\rho-T_{N(k)}-1-\i\sigma_{0})^{5}}\right]\right|
<\!\!<\sum_{k\neq j}2^{-n(k)(2+\epsilon)}|T_{N(j)}-T_{N(k)-1}|^{-5}.
\end{equation}
To bound this we distinguish two cases. First assume $j>k$. Then $n(j)>n(k)+1$ and in particular $n(j)>2$. So we have $T_{N(k)}+1\leq2^{n(k)}<2^{n(j)-1}-1\leq T_{N(j)}$. This implies
\begin{equation}
|T_{N(k)}+1-T_{N(j)}|\geq 2^{n(j)-1}-2^{n(k)}-1=2^{n(j)-1}(1-2^{n(k)-n(j)+1}-2^{-n(j)+1})\geq2^{n(j)-3}.
\end{equation}
Now assume $j<k$. Then $n(k)>n(j)+1$. In this case we have $T_{N(j)}\leq2^{n(j)}-1<2^{n(j)}<2^{n(k)-1}\leq T_{N(k)}+1$. This implies
\begin{equation}
|T_{N(k)}+1-T_{N(j)}|\geq2^{n(k)-1}-2^{n(j)}=2^{n(j)}(2^{n(k)-n(j)-1})-1)>2^{n(j)}.
\end{equation}
It follows from $T_{N(j)}\geq 2^{n(j)-1}-1>2^{n(j)-2}$ (for $n(j)>2$) that
\begin{equation}
\sum_{k\neq j}2^{-n(k)(2+\epsilon)}|T_{N(j)}-T_{N(k)-1}|^{-5}
\leq2^{15}2^{-5n(j)}\sum_{k\neq j}2^{-n(k)(2+\epsilon)}
\leq2^{15}T_{N(j)}^{-5}\sum_{k=1}^{\infty}2^{-k(2+\epsilon)}.
\end{equation}
So for large $T_{N(j)}$ we have
\begin{equation}\label{est}
\begin{split}
&\left|\sum_{k=1}^{\infty}2^{-n(k)(2+\epsilon)}\Re\left[\frac{1}{(\rho-T_{N(k)}-1-\i\sigma_{0})^{5}}\right]\right|\\
&=2^{-n(j)(2+\epsilon)}\frac{1+10(\Im\rho-\sigma_{0})^{2}+5(\Im\rho-\sigma_{0})^{4}}{(1+(\Im\rho-\sigma_{0})^{2})^{5}}+O(T_{N(j)}^{-5}),
\end{split}
\end{equation}
and thus, since $T_{N(j)}\leq2^{n(j)}-1<2^{n(j)}$,
\begin{equation}
2^{-n(j)(2+\epsilon)}\frac{1+10(\Im\rho-\sigma_{0})^{2}+5(\Im\rho-\sigma_{0})^{4}}{(1+(\Im\rho-\sigma_{0})^{2})^{5}}
\geq 2^{-2-\epsilon}T_{N(j)}^{-2-\epsilon}\frac{1+10(\sigma-\sigma_{0})^{2}+5(\sigma-\sigma_{0})^{4}}
{(1+4\sigma_{0}^{2})^{5}},
\end{equation}
which implies together with \eqref{est}
\begin{equation}
\left|\sum_{k=1}^{\infty}2^{-n(k)(2+\epsilon)}\Re\left[\frac{1}{(\rho-T_{N(k)}-1-\i\sigma_{0})^{5}}\right]\right|
>\!\!>T_{N(j)}^{-2-\epsilon}
\end{equation}
uniformly in $j$ and $\rho$. The other three terms in \eqref{testder} can be estimated in an analogous way. It follows that for $T_{N}(j)$ large and $\rho\in[T_{N}(j),T_{N(j)}-\i\sigma]$ we have $|\Re h'_{\epsilon}(\rho)|>\!\!>T_{N(j)}^{-2-\epsilon}$ uniformly in $j$ and $\rho$.
\end{proof}

We can now use Lemma \ref{testf} to prove the Proposition.\\

\textit{Proof of Proposition \ref{compbound}.} 
There exists (cf. Proposition \ref{polybound}) $\lbrace T_{N}\rbrace_{n=0}^{\infty}$ such that $|S_{\alpha,z_{0}}(\tfrac{1}{2}+\i\rho)|< c(\Gamma)T_{N}^{5}$ for $\rho\in[T_{N},T_{N}-\i\sigma]$. We conclude from \eqref{trunc} that
\begin{equation}
\lim_{N\to\infty}\int_{-\i\sigma+T_{N}}^{\i\sigma+T_{N}}h(\rho)\frac{S'_{\alpha,z_{0}}}{S_{\alpha,z_{0}}}(\tfrac{1}{2}+\i\rho)d\rho
\end{equation}
exists. Now $\arg(S_{\alpha,z_{0}}(\tfrac{1}{2}+\i\rho))<\!\!<1$ along the edges of $B(T)$ because the winding number about the poles is less than $1$. Thus, for $N\to\infty$,
\begin{equation}
\log S_{\alpha,z_{0}}(\tfrac{1}{2}\pm\sigma+\i T_{N})
<\!\!<\log|S_{\alpha,z_{0}}(\tfrac{1}{2}\pm\sigma+\i T_{N})|\sim\log\log T_{N}
\end{equation}
which follows from
\begin{equation}
\begin{split}
|S_{\alpha,z_{0}}(\tfrac{1}{2}-\sigma+\i T_{N})|
=|S_{\alpha,z_{0}}(\tfrac{1}{2}+\sigma-\i T_{N})|
=&|\overline{S_{\alpha,z_{0}}(\tfrac{1}{2}+\sigma-\i T_{N})}|\\
=&|S_{\alpha,z_{0}}(\tfrac{1}{2}+\sigma+\i T_{n})|\sim\log T_{N}.
\end{split}
\end{equation} 
Since $h(\pm\i\sigma+T_{N})<\!\!<T_{N}^{-2-\delta}$, we conclude by integration by parts that
\begin{equation}
\lim_{N\to\infty}\int_{-\i\sigma+T_{N}}^{\i\sigma+T_{N}}h'(\rho)\log|S_{\alpha,z_{0}}(\tfrac{1}{2}+\i\rho)|d\rho
\end{equation}
exists. Now we know (cf. Lemma \ref{testf}) that there is a subsequence $\lbrace T_{N(j)}\rbrace_{j}\subset\lbrace T_{N}\rbrace_{N}$ and $h_{\epsilon}\in H_{\sigma,\epsilon}$ such that
\begin{equation}\label{sym}
h'_{\epsilon}(\bar{\rho})=-h'_{\epsilon}(\rho)
\end{equation}
and for $\rho\in[T_{N(j)},T_{N(j)}-\i\sigma]$ we have
\begin{equation}
\Re h'_{\epsilon}(\rho)>\!\!>T_{N(j)}^{-2-\epsilon}.
\end{equation}
We have, since the integrand is odd and because of \eqref{sym},
\begin{equation}\label{id}
\begin{split}
&\int_{-\i\sigma+T_{N(j)}}^{\i\sigma+T_{N(j)}}h'_{\epsilon}(\rho)\log|S_{\alpha,z_{0}}(\tfrac{1}{2}+\i\rho)|d\rho\\
=&\left\{\int_{-T_{N(j)}}^{-\i\sigma-T_{N(j)}}+\int_{-\i\sigma+T_{N(j)}}^{T_{N(j)}}\right\}h'_{\epsilon}(\rho)\log|S_{\alpha,z_{0}}(\tfrac{1}{2}+\i\rho)|d\rho\\
=&-2\i\int_{-\i\sigma+T_{N(j)}}^{T_{N(j)}}\Re h'_{\epsilon}(\rho)\log|S_{\alpha,z_{0}}(\tfrac{1}{2}+\i\rho)|d\rho.
\end{split}
\end{equation}
Since $|\Re h'_{\epsilon}(\rho)|<\!\!<(1+|\Re\rho|)^{-2-\epsilon}$ we conclude that
\begin{equation}
\lim_{j\to\infty}\int_{-\i\sigma+T_{N(j)}}^{T_{N(j)}}\Re h'_{\epsilon}(\rho)\log\left\{ c(\Gamma)^{-1}T_{N(j)}^{-5}|S_{\alpha,z_{0}}(\tfrac{1}{2}+\i\rho)|\right\}d\rho
\end{equation}
exists. And since for all $\rho\in[T_{N(j)},T_{N(j)}-\i\sigma]$
\begin{equation}
\Re h'_{\epsilon}(\rho)\log\left\{c(\Gamma)^{-1}T_{N(j)}^{-5}|S_{\alpha,z_{0}}(\tfrac{1}{2}+\i\rho)|\right\}<0
\end{equation}
we have
\begin{equation}\label{last}
\begin{split}
&T_{N(j)}^{-2-\epsilon}\int_{0}^{\sigma}\log\left\{c(\Gamma)^{-1}T_{N(j)}^{-5}|S_{\alpha,z_{0}}(\tfrac{1}{2}+\sigma-t+\i T_{N(j)})|\right\}dt\\
&<\!\!<\int_{0}^{\sigma}|\Re h'_{\epsilon}(\i(t-\sigma)+T_{N(j)})|\big|\log c(\Gamma)^{-1}T_{N(j)}^{-5}|S_{\alpha,z_{0}}(\tfrac{1}{2}+\sigma-t+\i T_{N(j)})|\big|dt\\
&<\!\!<\left|\int_{0}^{\sigma}\Re h'_{\epsilon}(\i(t-\sigma)+T_{N(j)})\log|S_{\alpha,z_{0}}(\tfrac{1}{2}+\sigma-t+\i T_{N(j)})|dt\right|+O(T_{N(j)}^{-2-\epsilon}\log T_{N(j)}).
\end{split}
\end{equation}
The first term on the RHS of \eqref{last} converges as $j\to\infty$, which implies
\begin{equation}
\int_{0}^{\sigma}\log|S_{\alpha,z_{0}}(\tfrac{1}{2}+\sigma-t+\i T_{N(j)})|dt<\!\!<T_{N(j)}^{2+\epsilon}.\\
\end{equation}
\begin{flushright}
$\square$
\end{flushright}

We can now apply Proposition \ref{compbound} to derive the vanishing of the sequence of boundary terms \eqref{bdterms}.
\begin{thm}\label{van}
Let $\delta,\epsilon>0$ and $\lbrace T_{N(j)}\rbrace_{N(j)}$ as above. Then for any $h\in H_{\sigma+\epsilon,\delta}$
\begin{equation}
\lim_{T_{N(j)}\to\infty}\int_{T_{N(j)}-\i\sigma}^{T_{N(j)}}h(\rho)\frac{S'_{\alpha,z_{0}}}{S_{\alpha,z_{0}}}{(\tfrac{1}{2}+\i\rho)}d\rho=0.
\end{equation}
\end{thm}
\begin{proof}
We follow the same lines as in Proposition \ref{compbound}. In exactly the same way as in the proof above we obtain the identity \eqref{id} for $h\in H_{\sigma+\epsilon,\delta}\subset H_{\sigma,\delta}$. So
\begin{equation}
\begin{split}
&\lim_{j\to\infty}\left\{\int_{T_{N(j)}-\i\sigma}^{T_{N(j)}}+\int_{-T_{N(j)}}^{-T_{N(j)}-\i\sigma}\right\}h(\rho)\frac{d}{d\rho}\log S_{\alpha,z_{0}}(\tfrac{1}{2}+\i\rho)d\rho\\
=&\lim_{j\to\infty}\int_{T_{N(j)}-\i\sigma}^{T_{N(j)}}\lbrace h'(\rho)-h'(-\bar{\rho})\rbrace\log|S_{\alpha,z_{0}}(\tfrac{1}{2}+\i\rho)|d\rho.
\end{split}
\end{equation}
The limit vanishes since
\begin{equation}
\begin{split}
&\int_{T_{N(j)}-\i\sigma}^{T_{N(j)}}|\lbrace h'(\rho)-h'(-\bar{\rho})\rbrace|\big|\log|S_{\alpha,z_{0}}(\tfrac{1}{2}+\i\rho)|\big||d\rho|\\
<\!\!<\;&T_{N(j)}^{-2-\delta}\int_{T_{N(j)}-\i\sigma}^{T_{N(j)}}\big|\log|S_{\alpha,z_{0}}(\tfrac{1}{2}+\i\rho)|\big|d\rho|.
\end{split}
\end{equation}
where we have used Proposition \ref{compbound} and observe that by Cauchy's theorem $h\in H_{\sigma+\epsilon,\delta}$ implies
\begin{equation}
h'(\rho)<\!\!<(1+|\Re\rho|)^{-2-\delta}
\end{equation}
uniformly in $|\Im\rho|\leq\sigma$.
\end{proof}

As an application of Theorem \ref{van} we can now prove the trace formula. Let $h\in H_{\sigma,\delta}$ for any $\sigma>\tfrac{1}{2}$ and $\delta>0$. We will make a specific choice of $\sigma$ below. We observe that
\begin{equation}
\begin{split}
\int_{-\i\sigma-T}^{-\i\sigma+T}h(\rho)\frac{S'_{\alpha,\,z_{0}}}{S_{\alpha,\,z_{0}}}(\tfrac{1}{2}+\i\rho)d\rho
=&h(-\i\sigma+T)\log S_{\alpha,z_{0}}(\tfrac{1}{2}+\sigma+\i T)\\
&-h(-\i\sigma-T)\log S_{\alpha,z_{0}}(\tfrac{1}{2}+\sigma-\i T)\\
&-\int_{-\i\sigma-T}^{-\i\sigma+T}h'(\rho)\log S_{\alpha,\,z_{0}}(\tfrac{1}{2}+\i\rho)d\rho
\end{split}
\end{equation}
and, because of \eqref{log}, \eqref{Im} and the decay of $h$,
\begin{equation}
\lim_{T\to\infty}h(-\i\sigma\pm T)\log S_{\alpha,z_{0}}(\tfrac{1}{2}+\sigma\pm\i T)=0.
\end{equation}
Thus we can take $T\to\infty$ in \eqref{trunc} to obtain
\begin{equation}\label{comppretrace}
\begin{split}
\sum_{j\geq-M}h(\rho^{\alpha}_{j})-\sum_{j\geq-M}h(\rho_{j})
=\;&-\frac{1}{2\pi}\int_{-\i\sigma-\infty}^{-\i\sigma+\infty}h'(\rho)\log S_{\alpha,z_{0}}(\tfrac{1}{2}+\i\rho)d\rho.\\
\end{split}
\end{equation}
where we have used Theorem \ref{van}, the existence of the limit \eqref{limbound} and Weyl's law for the unperturbed spectrum.

In order to complete the proof of Theorem \ref{thm1} we need to expand the RHS of \eqref{comppretrace} into an identity term and diffractive orbit terms. We define
\begin{equation}
G_{s}^{\Gamma\backslash\scrI}(z,w)=\sum_{\gamma\in\Gamma\backslash\scrI}G_{s}(z,\gamma w).
\end{equation}
In the first term of \eqref{comppretrace} we can expand the logarithm in a power series
\begin{equation}\label{series}
\log\left[1+m\beta\psi(\tfrac{1}{2}+\i\rho)+\beta G^{\Gamma\backslash\scrI}_{\tfrac{1}{2}+\i\rho}\left(z_{0}, z_{0}\right)\right]\\
=\log\left[1+m\beta\psi(\tfrac{1}{2}+\i\rho)\right]-\sum_{k=1}^{\infty}\frac{(-1)^{k}}{k}\left(\frac{\beta G^{\Gamma\backslash\scrI}_{\tfrac{1}{2}+\i\rho}(z_{0}, z_{0})}{1+m\beta\psi(\tfrac{1}{2}+\i\rho)}\right)^{k}.
\end{equation}
This series converges absolutely and uniformly for all $\rho\in\CC$ with $\Im\rho=-\sigma$ if $\sigma$ is sufficiently large. To see this consider the following estimate
\begin{equation}
\begin{split}
&\left|\beta^{-1}+m\psi(\tfrac{1}{2}+\sigma+\i t)\right|\\
\geq& -\left|\beta\right|^{-1}+m\left|\psi(\tfrac{1}{2}+\sigma+\i t)\right|\\
\geq&-\left|\beta\right|^{-1}+m\left|\log(\tfrac{1}{2}+\sigma+\i t)\right|-\frac{m}{2\left|\tfrac{1}{2}+\sigma+\i t\right|}-m\sum_{n=1}^{\infty}\frac{B_{2n}}{n!}\left|\tfrac{1}{2}+\sigma+\i t\right|^{-2n}\\
\geq&-\left|\beta\right|^{-1}+m\log(\tfrac{1}{2}+\sigma)-\frac{m\pi}{2}-\frac{m}{1+2\sigma}-m\sum_{n=1}^{\infty}\frac{B_{2n}}{n!}(\tfrac{1}{2}+\sigma)^{-2n}.
\end{split}
\end{equation}
Since
\begin{equation}
\frac{1}{1+2\sigma}+\sum_{n=1}^{\infty}\frac{B_{2n}}{n!}\sigma^{-2n}=O(\sigma^{-1})
\end{equation}
we infer from the above for large enough $\sigma>\tfrac{1}{2}$ that
\begin{equation}
\left|\beta^{-1}+m\psi(\tfrac{1}{2}+\sigma+\i t)\right|\geq q(\beta)\log(\tfrac{1}{2}+\sigma)
\end{equation}
for some constant $0<q(\beta)<m$. Combining this with Lemma \ref{Greenbound} we obtain for $\Im\rho=-\sigma$ the estimate
\begin{equation}\label{bigbound}
\left|\frac{\beta G^{\Gamma\backslash\scrI}_{\tfrac{1}{2}+\i\rho}(z_{0},z_{0})}{1+m\beta\psi(\tfrac{1}{2}+\i\rho)}\right|
\leq\frac{q(\beta)^{-1}C(\Gamma,z_{0})}{(\tfrac{1}{2}+\sigma)^{1/2}\log(\tfrac{1}{2}+\sigma)}.
\end{equation}
We can now choose $\tilde{\sigma}(\beta)>\tfrac{1}{2}$ large enough such that
\begin{equation}\label{large1}
\frac{q(\beta)^{-1}C(\Gamma,z_{0})}{(\tfrac{1}{2}+\tilde{\sigma}(\beta))^{1/2}\log(\tfrac{1}{2}+\tilde{\sigma}(\beta))}<1
\end{equation}
and $\tilde{\sigma}(\beta)>\left|\Im\rho^{\alpha}_{-M}\right|$. This choice ensures that the series \eqref{series} converges absolutely and uniformly for all $\rho\in\CC$ with $\Im\rho=-\tilde{\sigma}(\beta)$.
Let $\sigma(\beta)=\tilde{\sigma}(\beta)+\epsilon$ for some $\epsilon>0$ and $h\in H_{\sigma(\beta),\delta}$. We now have for the RHS of \eqref{pretrace}
\begin{equation}\label{tr2}
\begin{split}
&\frac{1}{2\pi\i}\int_{-\i\tilde{\sigma}-\infty}^{-\i\tilde{\sigma}+\infty}h(\rho)\frac{m\beta\psi'(\tfrac{1}{2}+\i\rho)}{1+m\beta\psi(\tfrac{1}{2}+\i\rho)}d\rho\\
+&\frac{1}{2\pi\i}\sum_{k=1}^{\infty}\frac{(-\beta)^{k}}{k}\sum_{\gamma_{1},...,\gamma_{k}\in\Gamma\backslash\scrI}\int_{-\i\tilde{\sigma}-\infty}^{-\i\tilde{\sigma}+\infty}h'(\rho)\frac{\prod_{j=1}^{k}G_{\tfrac{1}{2}+\i\rho}(z_{0},\gamma_{j}z_{0})}{(1+m\beta\psi(\tfrac{1}{2}+\i\rho))^{k}}d\rho.
\end{split}
\end{equation}
Substituting the integral representation \eqref{intrep} of the free Green's function and changing order of integration gives
\begin{equation}\label{tr3}
\begin{split}
&\frac{1}{2\pi\i}\int_{-\i\tilde{\sigma}-\infty}^{-\i\tilde{\sigma}+\infty}h(\rho)\frac{m\beta\psi'(\tfrac{1}{2}+\i\rho)}{1+m\beta\psi(\tfrac{1}{2}+\i\rho)}d\rho
+\frac{1}{2\pi\i}\sum_{k=1}^{\infty}\frac{(-\beta)^{k}}{k}\\ &\times\sum_{\gamma_{1},\cdots,\gamma_{k}\in\Gamma\backslash\scrI}\int_{-\i\tilde{\sigma}-\infty}^{-\i\tilde{\sigma}+\infty}\int_{\tau_{1}}^{\infty}\cdots\int_{\tau_{k}}^{\infty}\frac{h'(\rho)}{(1+m\beta\psi(\tfrac{1}{2}+\i\rho))^{k}}\prod_{j=1}^{k}\frac{\,\e^{-\i\rho t_{j}}}{\sqrt{\cosh t_{j}-\cosh\tau_{j}}}\prod_{j=1}^{k}dt_{j}\,d\rho
\end{split}
\end{equation}
Note that we have the bound
\begin{equation}
\begin{split}
&\frac{|h'(\rho)|}{|1+m\beta\psi(\tfrac{1}{2}+\i\rho)|^{k}}\prod_{j=1}^{k}\left|\frac{\,\e^{-\i\rho t_{j}}}{\sqrt{\cosh t_{j}-\cosh\tau_{j}}}\right|\\
<\!\!<&\e^{-(\sigma+\tfrac{1}{4})(t_{1}+\cdots+t_{k})}\,(1+\left|\Re\rho\right|)^{-2-\delta}(\log\left|\Re\rho\right|)^{-k}
\end{split}
\end{equation}
on the integrand. So we can exchange integration by Fubini's Theorem. Recall (cf. section 1.2.) that the only zero of $1+m\beta\psi(\tfrac{1}{2}+\i\rho)$ in the halflane $\Im\rho<0$ is given by $\rho=-\i v_{\beta}$. So we have by shifting the contour from $\Im\rho=-\tilde{\sigma}$ to $\Im\rho=-\nu$
\begin{equation}
g_{\beta,k}(t)=\frac{(-1)^{k}}{2\pi\i k}\int_{-\i\tilde{\sigma}-\infty}^{-\i\tilde{\sigma}+\infty}h'(\rho)\frac{\e^{-\i\rho t}}{(1+m\beta\psi(\tfrac{1}{2}+\i\rho))^{k}}d\rho.
\end{equation}
So \eqref{tr3} equals
\begin{equation}
\begin{split}
&\frac{1}{2\pi\i}\int_{-\i\tilde{\sigma}-\infty}^{-\i\tilde{\sigma}+\infty}h(\rho)\frac{m\beta\psi'(\tfrac{1}{2}+\i\rho)}{1+m\beta\psi(\tfrac{1}{2}+\i\rho)}d\rho\\
+&\frac{1}{2\pi\i}\sum_{k=1}^{\infty}\frac{(-\beta)^{k}}{k} \sum_{\gamma_{1},\cdots,\gamma_{k}\in\Gamma\backslash\scrI}\int_{\tau_{1}}^{\infty}\cdots\int_{\tau_{k}}^{\infty}\prod_{j=1}^{k}\frac{g_{\beta,k}(t_{1}+\cdots+t_{k})}{\sqrt{\cosh t_{j}-\cosh\tau_{j}}}\prod_{j=1}^{k}dt_{j}.
\end{split}
\end{equation}

We have the bound
\begin{equation}
|g_{\beta,k}(t)|\leq\frac{\e^{-\tilde{\sigma} t}\int_{-\infty}^{\infty}|h'(\rho)|d\rho}{2\pi k\beta^{k}q(\beta)^{k}(\log(\tfrac{1}{2}+\tilde{\sigma}))^{k}}.
\end{equation}
We define the subset $\Gamma_{r,z_{0}}=\left\{\gamma\in\Gamma\backslash\scrI\mid d(z_{0},\gamma z_{0})<r\right\}$ of $\Gamma$ which is clearly finite because of discreteness of the group. For some $r>0$ we obtain the estimate
\begin{equation}
\begin{split}
&\sum_{\gamma_{1},\cdots,\gamma_{k}\in\Gamma_{r,z_{0}}}\int_{\tau_{1}}^{\infty}\cdots\int_{\tau_{k}}^{\infty}\frac{\left|g_{\beta,k}(t_{1}+\cdots+t_{k})\right|\prod_{n=1}^{k}dt_{n}}{\prod_{n=1}^{k}\sqrt{\cosh t_{n}-\cosh\tau_{n}}}\\
<\!\!<&\;\beta^{-k}(q(\beta)\log(\tfrac{1}{2}+\tilde{\sigma}))^{-k}\left[\sum_{\gamma\in\Gamma_{r,z_{0}}}\int_{\tau_{\gamma}}^{\infty}\frac{\e^{-\tilde{\sigma} t}dt}{\sqrt{\cosh t-\cosh\tau_{\gamma}}}\right]^{k}\\
\leq&\;\beta^{-k}(q(\beta)\log(\tfrac{1}{2}+\tilde{\sigma}))^{-k}\left[\sum_{\gamma\in\Gamma\backslash\scrI}\int_{\tau_{\gamma}}^{\infty}\frac{\e^{-\tilde{\sigma} t}dt}{\sqrt{\cosh t-\cosh\tau_{\gamma}}}\right]^{k}\\
\leq&\;\beta^{-k}\left(\frac{q(\beta)^{-1}C(\Gamma,z_{0})}{(\tfrac{1}{2}+\tilde{\sigma})^{1/2}
\log(\tfrac{1}{2}+\tilde{\sigma})}\right)^{k}<\;\beta^{-k}
\end{split}
\end{equation}
where we follow the same lines as in the proof of Lemma \ref{Greenbound}. So independently of $r>0$ the sum over $k$ converges absolutely. Taking $r\to\infty$ we see that
\begin{equation}\label{expansion}
\begin{split}
&\frac{1}{2\pi\i}\int_{-\i\tilde{\sigma}-\infty}^{-\i\tilde{\sigma}+\infty}h(\rho)\frac{m\beta\psi'(\tfrac{1}{2}+\i\rho)}{1+m\beta\psi(\tfrac{1}{2}+\i\rho)}d\rho\\
+&\frac{1}{2\pi\i}\sum_{k=1}^{\infty}\frac{(-\beta)^{k}}{k}\sum_{\gamma_{1},\cdots,\gamma_{k}\in\Gamma\backslash\scrI}\int_{\tau_{1}}^{\infty}\cdots\int_{\tau_{k}}^{\infty}\frac{g_{\beta,k}(t_{1}+\cdots+t_{n})\prod_{k=1}^{n}dt_{n}}{\prod_{n=1}^{k}\sqrt{\cosh t_{n}-\cosh\tau_{n}}}
\end{split}
\end{equation}
converges absolutely.

\section{Perturbed Eisenstein series}

In the remainder of this paper we will be concerned with $\Gamma\backslash\HH$ having a single cusp. The traditional Eisenstein series $E(z,s)$ are generalised eigenfunctions of the Laplacian on such a surface. In analogy with this definition we introduce an analogue of the Eisenstein series for the operator $\Delta_{\alpha,z_{0}}$. In order to admit non-$L^{2}$ eigenfunctions we enrich the domain $D_{\varphi(\alpha)}$ of $\Delta_{\alpha,z_{0}}$ by introducing the larger space
\begin{equation}
D^{*}_{0}=\left\{f\in C^{2}(\Gamma\backslash\HH)|f(z_{0})=0\right\},\qquad D^{*}_{0}\supset D_{0},
\end{equation}
where we have dropped the condition of square-integrability. We introduce the larger domain
\begin{equation}
D^{*}_{\varphi(\alpha)}=\lbrace g+cG^{\Gamma}_{\,t}(\cdot,z_{0})+c\,\e^{\i\varphi(\alpha)}G^{\Gamma}_{\,\bar{t}}(\cdot,z_{0})|g\in\D^{*}_{0}, c\in\CC\rbrace
\end{equation}
As before $\Delta_{\alpha,z_{0}}$ acts on $f\in D^{*}_{\varphi}$ as
\begin{equation}
\Delta_{\varphi}f=\Delta g - c\,t(1-t)G^{\Gamma}_{\,t}(z,z_{0})
-c\,\bar{t}(1-\bar{t})\e^{\i\varphi}G^{\Gamma}_{\,\bar{t}}(z,z_{0}).
\end{equation}

We define the perturbed Eisenstein series on $\Gamma\backslash\HH$ to be the solution to 
\begin{equation}\label{Eisenstein}
(\Delta_{\alpha,z_{0}}+s(1-s))\E^{\,\alpha,\,z_{0}}(z,s)=0,
\end{equation}
for any $s\in\CC$ with
\begin{equation}\label{dom}
E^{\,\alpha,\,z_{0}}(\cdot,s)\in D^{*}_{\varphi(\alpha)}.
\end{equation}
In the cusp we have the asymptotics
\begin{equation}\label{asmpt}
E^{\,\alpha,\,z_{0}}(x+\i y,s)=y^{s}+\varphi_{\,\alpha,z_{0}}(s)y^{1-s}+O(\e^{-2\pi y}),\qquad y\to\infty,
\end{equation}
for some $\varphi_{\,\alpha,z_{0}}(s)\in\CC$ which will be a meromorphic function in $s$. A standard argument shows that $\varphi_{\,\alpha,z_{0}}(s)$ is well-defined:
\begin{lem}
A solution to \eqref{Eisenstein} and \eqref{dom} with asymptotics \eqref{asmpt} uniquely defines the scattering coefficient $\varphi_{\alpha,z_{0}}(s)$.
\end{lem}
\begin{proof}
Suppose there exist two solutions $F(z,s)$ and $G(z,s)$ to \eqref{Eisenstein} and \eqref{asmpt} with asymptotics
\begin{equation}
\begin{split}
&F(z,s)=y^{s}+\varphi_{F}(s)y^{1-s}+O(\e^{-2\pi y})\\
&G(z,s)=y^{s}+\varphi_{G}(s)y^{1-s}+O(\e^{-2\pi y})
\end{split}
\end{equation}
in the cusp. Now assume $\Re s>\tfrac{1}{2}$, $\Im s>0$ and consider $H=F-G$. Then $H\in H^{2}(\Gamma\backslash\HH)$ and therefore $H\in D_{\varphi(\alpha)}$. Also \eqref{Eisenstein} implies that $H$ is an eigenfunction of $\Delta_{\alpha,z_{0}}$ with eigenvalue $s(1-s)\notin\RR$. Because $\Delta_{\alpha,z_{0}}$ is self-adjoint on $D_{\varphi\left(\alpha\right)}$, it follows that $H(\cdot,s)=0$. In particular
\begin{equation}
0=\lim_{y\to+\infty}[(\varphi_{F}(s)-\varphi_{G}(s))y^{1-s}+O(\e^{-2\pi y})]
=\lim_{y\to+\infty}[(\varphi_{F}(s)-\varphi_{G}(s))y^{1-s}]
\end{equation}
which shows that $\varphi_{F}(s)=\varphi_{G}(s)$ for $\Re s>\tfrac{1}{2}$ and $\Im s>0$. The result follows for all $s\in\CC$ by analytic continuation.
\end{proof}

We now derive $E^{\,\alpha,\,z_{0}}(z,s)$ rigorously and prove a perturbative analogue of the functional equation \eqref{Greenfe} which explains the connection between $E^{\,\alpha,\,z_{0}}(z,s)$ and the automorphic Green's function.
\begin{lem}
The perturbed Eisenstein series is given by
\begin{equation}\label{eisen}
E^{\,\alpha,\,z_{0}}(z,s)=E(z,s)-\lbrace1+\e^{\i\varphi\left(\alpha\right)}\rbrace\frac{E(z_{0},s)}{S_{\,\alpha,\,z_{0}}(s)}\,G^{\Gamma}_{s}(z,z_{0})
\end{equation}
and satisfies in analogy with \eqref{Greenfe}
\begin{equation}\label{Greenfe2}
G^{\Gamma}_{1-s}(z,z_{0})=\theta_{\,\alpha,\,z_{0}}(s)G^{\Gamma}_{s}(z,z_{0})-\frac{E(z_{0},s)E^{\,\alpha,\,z_{0}}(z,1-s)}{1-2s},
\end{equation}
where
\begin{equation}
\theta_{\,\alpha,\,z_{0}}(s)=\frac{S_{\,\alpha,\,z_{0}}(1-s)}{S_{\,\alpha,\,z_{0}}(s)}.
\end{equation}
\end{lem}
\begin{proof}
Let $s\in\CC$ and $S_{\alpha,\,z_{0}}(s)\neq0$. First of all we show that $E^{\,\alpha,\,z_{0}}(z,s)$ as given by \eqref{eisen} is in $D^{*}_{\varphi(\alpha)}$. For the automorphic Green's function we recall the decomposition
\begin{equation}
G^{\Gamma}_{s}(z,z_{0})=
\frac{1}{1+\e^{\i\varphi(\alpha)}}\left\{S_{\,\alpha,\,z_{0}}(z,s)+G^{\Gamma}_{t}(z,z_{0})+\e^{\i\varphi(\alpha)}G^{\Gamma}_{\bar{t}}(z,z_{0})\right\}.
\end{equation}
Substituting this back into \eqref{eisen} we obtain the decomposition of $E^{\,\alpha,\,z_{0}}\left(z,s\right)$
\begin{equation}
E^{\,\alpha,\,z_{0}}(z,s)=E(z,s)-E(z_{0},s)\frac{S_{\,\alpha,\,z_{0}}(z,s)}{S_{\,\alpha,\,z_{0}}(s)}
-\frac{E(z_{0},s)}{S_{\,\alpha,\,z_{0}}(s)}\,\left\{G^{\Gamma}_{t}(z,z_{0})+\e^{\i\varphi(\alpha)}G^{\Gamma}_{\bar{t}}(z,z_{0})\right\},
\end{equation}
which shows that $E^{\,\alpha,\,z_{0}}(\cdot,s)\in D^{*}_{\varphi(\alpha)}$ since
\begin{equation}
\lim_{z\to z_{0}}\left\{E(z,s)-E(z_{0},s)\frac{S_{\,\alpha,\,z_{0}}(z,s)}{S_{\,\alpha,\,z_{0}}(s)}\right\}
=\,E(z_{0},s)\left\{1-\frac{\lim_{z\to z_{0}}S_{\,\alpha,\,z_{0}}(z,s)}{S_{\alpha,\,z_{0}}(s)}\right\}=0.
\end{equation}

We also see that
\begin{equation}
\begin{split}
(\Delta_{\alpha,z_{0}}+s(1-s))E^{\,\alpha,\,z_{0}}(z,s)
=&\,(\Delta+s(1-s))\left\{E(z,s)-E(z_{0},s)\frac{S_{\,\alpha,\,z_{0}}(z,s)}{S_{\,\alpha,\,z_{0}}(s)}\right\}\\
&-\frac{E(z_{0},s)}{S_{\,\alpha,\,z_{0}}(s)}(s(1-s)-t(1-t))G^{\Gamma}_{t}(z,z_{0})\\
&-\e^{\i\varphi(\alpha)}\frac{E(z_{0},s)}{S_{\,\alpha,\,z_{0}}(s)}(s(1-s)-\bar{t}(1-\bar{t}))G^{\Gamma}_{\bar{t}}(z,z_{0}),
\end{split}
\end{equation}
and because of
\begin{equation}
(\Delta+s(1-s))E(z,s)=0,
\end{equation}
and
\begin{equation}
\begin{split}
(\Delta+s(1-s))\,S_{\alpha,z_{0}}(z,s)=&(-s(1-s)+t(1-t))G^{\Gamma}_{t}(z,z_{0})\\
&+\e^{\i\varphi(\alpha)}(-s(1-s)+\bar{t}(1-\bar{t}))G^{\Gamma}_{\bar{t}}(z,z_{0}),
\end{split}
\end{equation}
which follows from Lemma \ref{Hilbert}, we have
\begin{equation}
(\Delta_{\alpha,z_{0}}+s(1-s))E^{\alpha,\,z_{0}}(z,s)=0.
\end{equation}

In order to prove \eqref{Greenfe2} we rewrite $E^{\,\alpha,\,z_{0}}(z,s)$ as
\begin{equation}
\begin{split}
E^{\,\alpha,\,z_{0}}(z,s)&=E(z,s)-\frac{E(z_{0},s)}{F_{\,\alpha,\,z_{0}}(s)}G^{\Gamma}_{s}(z,z_{0})\\
&=\left\{\frac{1-2s}{E(z_{0},1-s)}-\frac{E(z_{0},s)}{F_{\,\alpha,\,z_{0}}(s)}\right\}G^{\Gamma}_{s}(z,z_{0})-\frac{1-2s}{E(z_{0},1-s)}G^{\Gamma}_{1-s}(z,z_{0}),
\end{split}
\end{equation}
where \eqref{Greenfe} was substituted and for convenience we let
\begin{equation}
F_{\,\alpha,\,z_{0}}(s)=\frac{S_{\,\alpha,\,z_{0}}(s)}{1+\e^{\i\varphi(\alpha)}}.
\end{equation}
Consequently one gets the equation
\begin{equation}
\frac{E^{\,\alpha,\,z_{0}}(z,s)E(z_{0},1-s)}{1-2s}=\left\{1-\frac{E(z_{0},1-s)E(z_{0},s)}{(1-2s)F_{\,\alpha,\,z_{0}}(s)}\right\}G^{\Gamma}_{s}(z,z_{0})-G^{\Gamma}_{1-s}(z,z_{0}).
\end{equation}
So let
\begin{equation}
\begin{split}
\theta_{\,\alpha,\,z_{0}}(s)&=1-\frac{E(z_{0},1-s)E(z_{0},s)}{(1-2s)F_{\,\alpha,\,z_{0}}(s)}\\
&=1-\frac{\lim_{z\to z_{0}}(G^{\Gamma}_{s}-G^{\Gamma}_{1-s})(z,z_{0})}{F_{\,\alpha,\,z_{0}}(s)}\\
&=\frac{F_{\,\alpha,\,z_{0}}(s)-\lim_{z\to z_{0}} (G^{\Gamma}_{s}-G^{\Gamma}_{1-s})(z,z_{0})}{F_{\,\alpha,\,z_{0}}(s)}\\
&=\frac{F_{\,\alpha,\,z_{0}}(1-s)}{F_{\,\alpha,\,z_{0}}(s)}\\
&=\frac{S_{\,\alpha,\,z_{0}}(1-s)}{S_{\,\alpha,\,z_{0}}(s)}.
\end{split}
\end{equation}
\end{proof}
\begin{remark}
We note from \eqref{eisen} that the eigenvalues in the residual spectrum $R_{\alpha,z_{0}}(\Gamma\backslash\HH)$ have an interpretation as residues of $E^{\,\alpha,\,z_{0}}(z,s)$ provided $E(z_{0},s)\neq0$.
\end{remark}
It is now a straightforward corollary of the previous lemma to derive the functional equation and therefore the scattering coefficient $\varphi_{\,\alpha,\,z_{0}}\left(s,z_{0}\right)$ for the perturbed Eisenstein series $\E^{\,\alpha,\,z_{0}}\left(z,s\right)$.
\begin{cor}
\begin{equation}
\E^{\alpha,\,z_{0}}(z,s)=\varphi_{\,\alpha,\,z_{0}}(s)\E^{\alpha,\,z_{0}}(z,1-s),
\end{equation}
where the scattering coefficient is given by
\begin{equation}\label{scatt}
\varphi_{\,\alpha,\,z_{0}}(s)=\varphi(s)\frac{S_{\,\alpha,\,z_{0}}(1-s,z_{0})}{S_{\,\alpha,\,z_{0}}(s,z_{0})}.
\end{equation}
\end{cor}
\begin{proof}
We use the relation \eqref{Greenfe2} to derive the functional equation for $E^{\alpha,z_{0}}(z,s)$.
\begin{equation}
\begin{split}
E^{\alpha,\,z_{0}}(z,s)&=-\frac{1-2s}{E(z_{0},1-s)}\left[\theta^{\,\alpha,\,z_{0}}(s)G^{\Gamma}_{s}(z,z_{0})-G^{\Gamma}_{1-s}(z,z_{0})\right]\\
&=-\frac{1-2s}{\varphi(1-s)E(z_{0},s)}\left[\theta^{\,\alpha,\,z_{0}}(s)G^{\Gamma}_{s}(z,z_{0})-G^{\Gamma}_{1-s}(z,z_{0})\right]\\
&=\varphi(s)\frac{1-2s}{E(z_{0},s)}\left[\theta^{\,\alpha,\,z_{0}}(s)\theta^{\,\alpha,\,z_{0}}(1-s)G^{\Gamma}_{1-s}(z,z_{0})-\theta_{\alpha,\,z_{0}}(s)G^{\Gamma}_{s}(z,z_{0})\right]\\
&=\varphi(s)\theta_{\,\alpha,\,z_{0}}\left(s\right)E^{\alpha,\,z_{0}}(z,1-s).
\end{split}
\end{equation}
\end{proof}

\section{Bounds on the relative zeta function for surfaces with one cusp}

We continue with the construction of a sequence of intervals $\lbrace[T_{N},T_{N}-\i\sigma]\rbrace_{N\in\NN}$ crossing the strip $\RR\times[0,-\i\sigma]$ on which the relative zeta function $S_{\alpha,z_{0}}(\tfrac{1}{2}+\i\rho)$ admits a uniform upper bound of order $e^{c(\Gamma,\alpha,z_{0})T_{N}^{2}\ln T_{N}}$ for some positive constant $c(\Gamma,\alpha,z_{0})$. We suspect that at least in arithmetic cases it is possible to do much better and in fact achieve a polynomial bound. However, for a generic non-compact surface the general bound  on the respective Eisenstein series (cf. Thm. 12.9(d) in \cite{Hj3}) gives the above bound which will suffice for our purposes. 

As in the compact case the proof uses the spectral expansion of $S_{\alpha,z_{0}}(s)$. The discrete part of the relative zeta function admits a uniform bound of polynomial growth for a suitably chosen sequence of intervals and the proof follows exactly the same lines as the proof of Proposition \ref{polybound}. The case of the existence of only finitely many cusp forms is trivial.

In order to obtain the analogous bound for the continuous part we require a bound on the scattering coefficient $\varphi(s)$ close to the critical line.
\begin{lem}\label{scatbound}
Let $\xi=s-\tfrac{1}{2}$ and $\sigma\geq\tfrac{1}{2}$. If $\;0\leq-\Re\xi\leq\sigma$ and $|\Im\xi|\to\infty$ then there exists a positive constant $C_{1}(\Gamma)$ s. t.
\begin{equation}
|\varphi(s)|<\!\!<_{\Gamma,\sigma} e^{C_{1}(\Gamma)(\sigma+|\Im\xi|)}\prod_{\gamma_{j}<2|Im\xi|+2}\left|\frac{\xi-\eta_{j}-\i\gamma_{j}}{\xi+\eta_{j}+\i\gamma_{j}}\right|\left|\frac{\xi-\eta_{j}+\i\gamma_{j}}{\xi+\eta_{j}-\i\gamma_{j}}\right|.
\end{equation}
\end{lem}
\begin{proof}
The scattering coefficient $\varphi(s)$ is (see \cite{Hj3}, Prop. 12.6, p. 157) of the general form
\begin{equation}\label{prod}
\varphi(s)=\varphi(\tfrac{1}{2})e^{A\xi}\prod_{k=1}^{M}\frac{\xi-\i\rho_{-k}}{\xi+\i\rho_{-k}}\prod_{j=0}^{\infty}\frac{\xi-\eta_{j}-\i\gamma_{j}}{\xi+\eta_{j}+\i\gamma_{j}}\;\frac{\xi-\eta_{j}+\i\gamma_{j}}{\xi+\eta_{j}-\i\gamma_{j}},\;A>0,
\end{equation}
where the infinite product is convergent whenever $\xi\neq\eta_{j}\pm\i\gamma_{j}$. For all $j$ we have $\eta_{j}>0$ and $\gamma_{j}\geq0$. We also have (cf. Prop. 12.5, p. 156 \cite{Hj3})
\begin{equation}
\sum_{j=0}^{+\infty}\frac{\eta_{j}}{1+\gamma_{j}^{2}}<+\infty.
\end{equation}
The first two terms are trivially bounded. We split the infinite product into two parts. The first term is estimated as follows (cf. \cite{Hj3}, p. 159, (**) and the fourth line from the bottom)
\begin{equation}
\left|\log\prod_{\gamma_{j}\geq2|Im\xi|+2}\left|\frac{\xi-\eta_{j}-\i\gamma_{j}}{\xi+\eta_{j}+\i\gamma_{j}}\;\frac{\xi-\eta_{j}+\i\gamma_{j}}{\xi+\eta_{j}-\i\gamma_{j}}\right|\right|
\leq\sum_{\gamma_{j}\geq2|Im\xi|+2}\left|\log\left(1-\frac{4\xi\eta_{j}}{(\eta_{j}+\xi)^{2}+\gamma_{j}^{2}}\right)\right|.
\end{equation}
We note that for $|\Im\xi|$ large we have
\begin{equation}
\begin{split}
\frac{4|\xi|\eta_{j}}{|(\eta_{j}+\xi)^{2}+\gamma_{j}^{2}|}
&\leq\frac{4|\xi|\eta_{j}}{|\Re\lbrace(\eta_{j}+\xi)^{2}+\gamma_{j}^{2}\rbrace|}\\
&\leq\frac{4|\xi|\eta_{j}}{(\eta_{j}+\Re\xi)^{2}+\gamma_{j}^{2}-(\Im\xi)^{2}}\\
&\leq\frac{4|\xi|\eta_{j}}{\gamma_{j}^{2}-(\Im\xi)^{2}}
\leq\frac{4|\xi|\eta_{j}}{\tfrac{3}{4}\gamma_{j}^{2}+1}\\
&\leq\frac{4|\xi|\eta_{j}}{3(\Im\xi)^{2}+4}\leq\frac{4(\sigma+|\Im\xi|)}{3(\Im\xi)^{2}+4}<\tfrac{1}{2}
\end{split}
\end{equation}
where we have used that $\gamma_{j}\geq2|\Im\xi|+2$ implies $\tfrac{1}{4}\gamma_{j}^{2}\geq(\Im\xi)^{2}+1$. Note that for $|z|<\tfrac{1}{2}$ we have
\begin{equation}
|\log(1-z)|\leq\sum_{k=1}^{\infty}\frac{|z|^{k}}{k}\leq2|z|,
\end{equation}
and as a consequence we have the estimate
\begin{equation}
\begin{split}
\sum_{\gamma_{j}\geq2|Im\xi|+2}\left|\log\left(1-\frac{4\xi\eta_{j}}{(\eta_{j}+\xi)^{2}+\gamma_{j}^{2}}\right)\right|
&\leq8|\xi|\sum_{\gamma_{j}\geq2|Im\xi|+2}\frac{\eta_{j}}{|(\eta_{j}+\xi)^{2}+\gamma_{j}^{2}|}\\
&\leq8|\xi|\sum_{\gamma_{j}\geq2|Im\xi|+2}\frac{\eta_{j}}{1+\tfrac{3}{4}\gamma_{j}^{2}}\\
&\leq\frac{32}{3}|\xi|\sum_{j=0}^{+\infty}\frac{\eta_{j}}{1+\gamma_{j}^{2}}.
\end{split}
\end{equation}
Let $C_{1}(\Gamma)=\tfrac{32}{3}\sum_{j=0}^{+\infty}\frac{\eta_{j}}{1+\gamma_{j}^{2}}$. Then
\begin{equation}
\prod_{\gamma_{j}\geq2|Im\xi|+2}\left|\frac{\xi-\eta_{j}-\i\gamma_{j}}{\xi+\eta_{j}+\i\gamma_{j}}\;\frac{\xi-\eta_{j}+\i\gamma_{j}}{\xi+\eta_{j}-\i\gamma_{j}}\right|\leq e^{C_{1}(\Gamma)|\xi|}\leq e^{C_{1}(\Gamma)(\sigma+|\Im\xi|)}.
\end{equation}
\end{proof}

To be able to control a residual term which, as we will see, arises from the continous part of the relative zeta function we also require a bound on the Eisenstein series on subsets of the strip $\tfrac{1}{2}\leq\Re s\leq\tfrac{3}{2}$.
\begin{lem}\label{Ebound}
Let $c_{0}>0$. For any $T\in\RR$ there exists $T_{0}$ with $|T-T_{0}|\leq c_{0}T_{0}^{-2}$ such that
\begin{equation}
|E(z_{0},\tfrac{1}{2}+t+\i T_{0})|<\!\!<T_{0}^{2}e^{3T_{0}}
\end{equation}
for all $t\in[0,1]$.
\end{lem}
\begin{proof}
We have the following bounds (cf. Thm. 12.9.(d), p. 164, Prop. 12.7, p. 161 \cite{Hj3}, (10.13), p. 142 \cite{Iw})
\begin{equation}\label{bound1}
|E(z_{0},\sigma+\i t)|<\!\!<\sqrt{w(t)}\,e^{3|t|},
\end{equation}
uniformly for $\tfrac{1}{2}\leq\sigma\leq\tfrac{3}{2}$, where $w$ satisfies 
\begin{equation}
\forall t\in\RR:\,w(t)\geq1
\end{equation}
and, making use of (10.13) in \cite{Iw} and running through the same argument as on pp. 161-62 in \cite{Hj3}, we have
\begin{equation}\label{bound2}
\int^{R}_{-R}w(t)dt<\!\!<R^{2}.
\end{equation}
Using bound \eqref{bound2} we infer
\begin{equation}
\int_{T-c_{0}T^{-2}}^{T+c_{0}T^{-2}}w(r)dr\leq c_{1}T^{2}
\end{equation}
for some uniform constant $c_{1}>0$. Now suppose for a contradiction that for all $r$ with $|r-T|\leq c_{0}T^{-2}$ we have $w(r)>\frac{c_{1}c_{0}^{-1}}{2}T^{4}$. This implies
\begin{equation}
\int_{T-c_{0}T^{-2}}^{T+c_{0}T^{-2}}w(r)dr>2c_{0}T^{-2}\,\frac{c_{1}c_{0}^{-1}}{2}T^{4}=c_{1}T^{2}
\end{equation}
which is a contradiction to the above bound. So we conclude that there exists $T_{0}$ with $|T-T_{0}|\leq c_{0}T^{-2}$ such that $w(T_{0})\leq\frac{c_{1}c_{0}^{-1}}{2}T^{4}$. This implies by bound \eqref{bound1}
\begin{equation}
|E(z_{0},\tfrac{1}{2}+t+\i T_{0})|<\!\!<T_{0}^{2}e^{3T_{0}}
\end{equation}
for $t\in[0,1]$.
\end{proof}

Using Proposition \ref{polybound} and Lemmas \ref{scatbound} and \ref{Ebound} we can establish a bound on the relative zeta function on intervals crossing the strip $|\Im\rho|\leq\sigma$ in between its poles and zeros.
\begin{prop}\label{seqbound}
There exists a sequence $\lbrace T_{N}\rbrace_{N\in\NN}$ in $\RR_{+}$, $\lim_{N\to\infty}T_{N}=+\infty$, such that for all $N$ and $t\in[0,\sigma]$ we have the uniform bound
\begin{equation}
|S_{\alpha,z_{0}}(\tfrac{1}{2}+\i\rho_{N}(t))|<\!\!<e^{C_{2}(\Gamma)T_{N}^{2}\ln T_{N}},
\end{equation}
where $\rho_{N}(t)=T_{N}-\i t$ and $C_{2}(\Gamma)$ is some positive constant.
\end{prop}
\begin{proof}
Recall that for $\Im\rho<0$ the relative zeta function $S_{\alpha,z_{0}}(\tfrac{1}{2}+\i\rho)$ is given by
\begin{equation}
\begin{split}
S_{\alpha,z_{0}}(\tfrac{1}{2}+\i\rho)=&\alpha^{-1}+\sum_{j}|\varphi_{j}(z_{0})|^{2}\left(\frac{1}{\rho_{j}^{2}-\rho^{2}}-\Re\left[\frac{1}{\rho_{j}^{2}-\rho(\eta)^{2}}\right]\right)\\
&+\frac{1}{4\pi}\int_{-\infty}^{+\infty}E(z_{0},\tfrac{1}{2}+\i r)E(z_{0},\tfrac{1}{2}-\i r)\left(\frac{1}{r^{2}-\rho^{2}}-\Re\left[\frac{1}{r^{2}-\rho(\eta)^{2}}\right]\right)dr.
\end{split}
\end{equation}
$S_{\alpha,z_{0}}(\tfrac{1}{2}+\i\rho)$ can be continued meromorphically to the full complex plane by shifting the contour of integration and collecting a residue. We will require a continuation to the real line in order to establish the desired bound. We first have to introduce some notation.

Define $\gamma:[-1,1]\mapsto\RR_{-}$
\begin{equation}
\gamma(t)=
\begin{cases}
\begin{split}
&-1, &\text{if}\;t&\in[-\tfrac{1}{3},\tfrac{1}{3}],\\
&-1+\tfrac{3}{2}(t-\tfrac{1}{3}), &\text{if}\;t&\in[\tfrac{1}{3},1],\\
&-1+\tfrac{3}{2}(\tfrac{1}{3}+t), &\text{if}\;t&\in[-1,-\tfrac{1}{3}].
\end{split}
\end{cases}
\end{equation}
We define a contour $\Gamma_{N}(t)$, where $\rho_{N}(t)=T_{N}-\i t$, $t\in[0,\sigma]$, by the parameterisation
\begin{equation}
\Gamma_{N}(t)(u)=
\begin{cases}
\begin{split}
u,\qquad&\text{if}\; |u-T_{N}|>\frac{1}{6}cT_{N}^{-1}\\
&\text{or if}\; |u-T_{N}|\leq\frac{1}{6}cT_{N}^{-1} \;\text{and}\; t\geq\tfrac{1}{12}cT_{N}^{-1},\\
u+\i T_{N}^{-1}&\gamma(6c^{-1}T_{N}(u-T_{N})),\qquad\text{if}\;|u-T_{N}|\leq \frac{1}{6}cT_{N}^{-1}\;\text{and}\; 0\leq t<\tfrac{1}{12}cT_{N}^{-1},
\end{split}
\end{cases}
\end{equation}
for $u\in\RR_{+}$ and $c>0$ is a constant which we will determine later on in the proof. To simplify notation we will also denote this decomposition by
\begin{equation}
\Gamma_{N}(t)=(\RR_{+}\backslash[T_{N}-\tfrac{1}{6}cT_{N}^{-1},T_{N}+\tfrac{1}{6}cT_{N}^{-1}])\cup\gamma_{N}(t).
\end{equation}

Let 
\begin{equation}
\delta_{N}(t)=
\begin{cases}
1,\;\text{if}\;0\leq t<\tfrac{1}{12}cT_{N}^{-1}\\
0,\;\text{if\,}\; t\geq\tfrac{1}{12}cT_{N}^{-1}.
\end{cases}
\end{equation}
Then we have
\begin{equation}\label{rep}
\begin{split}
S_{\alpha,z_{0}}(\tfrac{1}{2}+\i\rho_{N}(t))=&\,\alpha^{-1}+\sum_{j}|\varphi_{j}(z_{0})|^{2}\left(\frac{1}{\rho_{j}^{2}-\rho_{N}(t)^{2}}-\Re\left[\frac{1}{\rho_{j}^{2}-\rho(\eta)^{2}}\right]\right)\\
&-\delta_{N}(t)\frac{E(z_{0},\tfrac{1}{2}+\i\rho_{N}(t))E(z_{0},\tfrac{1}{2}-\i\rho_{N}(t))}{2\i\rho_{N}(t)}\\
&+\frac{1}{2\pi}\int_{\Gamma_{N}(t)}E(z_{0},\tfrac{1}{2}+\i r)E(z_{0},\tfrac{1}{2}-\i r) \left(\frac{1}{r^{2}-\rho_{N}(t)^{2}}-\Re\left[\frac{1}{r^{2}-\rho(\eta)^{2}}\right]\right)dr.
\end{split}
\end{equation}

Recall that the set $\lbrace\gamma_{j}\rbrace_{j=0}^{\infty}$ denotes the real parts of the resonances of $E(z,\tfrac{1}{2}+\i\rho)$ (since $i\rho=\xi=s-\tfrac{1}{2}$). Let $K=\lbrace\gamma_{j}\rbrace_{j=0}^{\infty}\cup\lbrace\rho_{i}\rbrace_{i=0}^{\infty}
=\lbrace\kappa_{l}\rbrace_{l=0}^{\infty}$. One has the upper bound (cf. (7.11), p. 101 \cite{Iw})
\begin{equation}\label{resbound}
\#\lbrace j\mid \gamma_{j}\leq R\rbrace <\!\!< R^{2},
\end{equation}
which implies
\begin{equation}
\#\lbrace l\mid\kappa_{l}\leq R\rbrace=\#\lbrace j\mid \gamma_{j}\leq R\rbrace+\#\lbrace i\mid \rho_{i}\leq R\rbrace<\!\!<R^{2}.
\end{equation}
Therefore, along the lines of the proof of Lemma \eqref{spacing}, we can pick an infinite sequence $\tau_{N}=\tfrac{1}{2}(\kappa_{N}+\kappa_{N+1})$ with $|\kappa_{N}-\kappa_{N+1}|\geq c_{1}\tau_{N}^{-1}$. Take $c_{0}=\tfrac{1}{6}c_{1}$ in Lemma \ref{Ebound}. So for every $N$ there exists $T_{N}$ with $|T_{N}-\tau_{N}|\leq\tfrac{1}{6}c_{1}T_{N}^{-1}$ such that
\begin{equation}\label{bound3}
|E(z_{0},\tfrac{1}{2}+t+\i T_{N})|<\!\!<T_{N}^{2}e^{3T_{N}}.
\end{equation}
We also note $|\kappa_{N}-\kappa_{N+1}|\geq cT_{N}^{-1}$ for some $0<c<c_{1}$ and $T_{N}$ sufficiently large. This implies for all $j$
\begin{equation}
|\rho_{N}(t)-\kappa_{j}|\geq|T_{N}-\kappa_{j}|\geq\tfrac{1}{3}|\kappa_{N}-\kappa_{N+1}|\geq\tfrac{1}{3}cT_{N}^{-1}.
\end{equation}

We now turn to estimating $S_{\alpha,z_{0}}(\tfrac{1}{2}+\i\rho_{N}(t))$ where we use the representation \eqref{rep}. We split the integral into two parts, where we can drop the real part in the regularisation term for simplicity, 
\begin{equation}
\begin{split}
&\int_{\Gamma_{N}(t)}E(z_{0},\tfrac{1}{2}+\i r)E(z_{0},\tfrac{1}{2}-\i r)\left(\frac{1}{r^{2}-\rho_{N}(t)^{2}}-\frac{1}{r^{2}-\rho(\eta)^{2}}\right)dr\\
=&\left\{\int_{|r-T_{N}|>cT_{N}^{-1}/6}+\int_{\gamma_{N}(t)}\right\}E(z_{0},\tfrac{1}{2}+\i r)E(z_{0},\tfrac{1}{2}-\i r)\left(\frac{1}{r^{2}-\rho_{N}(t)^{2}}-\frac{1}{r^{2}-\rho(\eta)^{2}}\right)dr.
\end{split}
\end{equation}
We bound the tail as follows
\begin{equation}
\begin{split}
&|\rho_{N}(t)^{2}-\rho(\eta)^{2}|\int_{T_{N}+cT_{N}^{-1}/6}^{\infty}\frac{\left|E(z_{0},\tfrac{1}{2}+\i r)\right|^{2}dr}{|r^{2}-\rho_{N}(t)^{2}||r^{2}-\rho(\eta)^{2}|}\\
<\!\!<&\,T_{N}^{2}\left\{\int_{T_{N}+\tfrac{1}{6}cT_{N}^{-1}}^{T_{N}+1}+\sum_{k=1}^{\infty}\int_{T_{N}+k}^{T_{N}+k+1}\right\}\frac{\left|E(z_{0},\tfrac{1}{2}+\i r)\right|^{2}dr}{|r^{2}-\rho_{N}(t)^{2}||r^{2}-\rho(\eta)^{2}|}\\
<\!\!<&\,T_{N}^{2}+T_{N}^{2}\sum_{k=2}^{\infty}\frac{(T_{N}+k)^{2}}
{|(T_{N}+k)^{2}-\rho_{N}(t)^{2}||(T_{N}+k)^{2}-\rho(\eta)^{2}|}
\end{split}
\end{equation}
where we have used the mean value bound on Eisenstein series
\begin{equation}
\int_{0}^{T}|E(z_{0},\tfrac{1}{2}+\i t)|^{2}dt<\!\!<T^{2}
\end{equation}
which follows straight away from Prop. 7.2., p. 101 in \cite{Iw}. We continue with the estimate as follows
\begin{equation}
\begin{split}
&T_{N}^{2}\left\{\sum_{k=2}^{\left\lfloor T_{N}\right\rfloor}+\sum_{k=\left\lceil T_{N}\right\rceil}^{\infty}\right\}\frac{(T_{N}+k)^{2}}{|(T_{N}+k)^{2}-\rho_{N}(t)^{2}| |(T_{N}+k)^{2}-\rho(\eta)^{2}|}\\
\leq&\;T_{N}^{3}\max_{2\leq k\leq\left\lfloor T_{N}\right\rfloor}\frac{(T_{N}+k)^{2}}{|(T_{N}+k)^{2}-\Re\lbrace\rho_{N}(t)^{2}\rbrace||(T_{N}+k)^{2}-\Re\lbrace\rho(\eta)^{2}\rbrace|}\\
&\;+T_{N}^{2}\sum_{k=\left\lceil T_{N}\right\rceil}^{\infty}\frac{(T_{N}+k)^{2}} {|(T_{N}+k)^{2}-\Re\lbrace\rho_{N}(t)^{2}\rbrace||(T_{N}+k)^{2}-T_{N}^{2}+t^{2}]|}\\
<\!\!<&\;T_{N}^{2}+T_{N}^{2}\int_{\left\lceil T_{N}\right\rceil}^{\infty}\frac{(T_{N}+x)^{2}dx}{|(T_{N}+x)^{2}-T_{N}^{2}+t^{2}| |(T_{N}+x)^{2}-\Re\lbrace\rho(\eta)^{2}\rbrace|}\\
=&\;T_{N}^{2}+T_{N}\int_{T_{N}^{-1}\left\lceil T_{N}\right\rceil}^{\infty}\frac{(1+y)^{2}dy} {|(1+y)^{2}-1+T_{N}^{-2}t^{2}||(1+y)^{2}-1+T_{N}^{-2}\Re\lbrace\rho(\eta)^{2}\rbrace|}\\
<\!\!<&\;T_{N}^{2}.
\end{split}
\end{equation}
We bound the central part by
\begin{equation}
\begin{split}
&\int_{-1}^{1}|E(z_{0},\tfrac{1}{2}+\i[\gamma_{N}(t)](r))|^{2}\frac{|\varphi(\tfrac{1}{2}-\i[\gamma_{N}(t)](r))|}
{|[\gamma_{N}(t)](r)^{2}-\rho_{N}(t)^{2}||[\gamma_{N}(t)](r)^{2}-\rho(\eta)^{2}|}\left|\frac{d[\gamma_{N}(t)](r)}{dr}\right|dt\\
<\!\!<&\,T_{N}^{2}\,\int_{T_{N}-cT_{N}^{-1}/6}^{T_{N}+cT_{N}^{-1}/6}w(r)e^{6|r|}dr\,\sup_{r\in \,\gamma_{N}(t)}\frac{|\varphi(\tfrac{1}{2}-\i r)|}{|r^{2}-\rho_{N}(t)^{2}||r^{2}-\rho(\eta)^{2}|}\\
<\!\!<&\,T_{N}^{4}e^{6 T_{N}}\sup_{r\in \,\gamma_{N}(t)}\frac{|\varphi(\tfrac{1}{2}-\i r)|}{|r^{2}-\rho_{N}(t)^{2}||r^{2}-\rho(\eta)^{2}|}
\end{split}
\end{equation}
which follows from \eqref{bound1} and \eqref{bound2}.

For $r\in\gamma_{N}(t)$ we have
\begin{equation}
|r-\rho_{N}(t)|\geq\tfrac{1}{12}cT_{N}^{-1}
\end{equation}
and for $T_{N}$ large
\begin{equation}
|r+\rho_{N}(t)|=T_{N}+O(T_{N}^{-1}),\quad|r^{2}-\rho(\eta)^{2}|=T_{N}^{2}+O(1).
\end{equation}
Hence
\begin{equation}
\sup_{r\in\gamma_{N}(t)}\frac{1}{|r^{2}-\rho_{N}(t)^{2}||r^{2}-\rho(\eta)^{2}|}<\!\!<T_{N}.
\end{equation}

Let $\xi=-\i r$ and $r\in\gamma_{N}(t)$. Assume $\gamma_{j}<2|\Im\xi|+2$. Then we have
\begin{equation}
|\xi+\eta_{j}\pm\i\gamma_{j}|=|\Re\xi+\eta_{j}+\i\Im\xi\pm\i\gamma_{j}|
\geq|\Im\xi\pm\gamma_{j}|\geq\tfrac{1}{3}cT_{N}^{-1}
\end{equation}
by the choice of the sequence $\lbrace T_{N}\rbrace_{N}$. Let $\eta=\sup_{j}\eta_{j}$. We know that $\eta<+\infty$ (cf. \cite{Hj3}, last line in the proof of Prop. 12.5, p. 157). We thus have
\begin{equation}
|\xi-\eta_{j}+\i\Im\xi\pm\i\gamma_{j}|\leq|\Re\xi|+\eta_{j}+|\Im\xi|+|\gamma_{j}|
\leq\sigma+\eta+2|\Im\xi|+2<\!\!<T_{N}.
\end{equation}
This implies
\begin{equation}\label{bound4}
\begin{split}
\prod_{\gamma_{j}<2|Im\xi|+2}\left|\frac{\xi-\eta_{j}-\i\gamma_{j}}{\xi+\eta_{j}+\i\gamma_{j}}\right|\left|\frac{\xi-\eta_{j}+\i\gamma_{j}}{\xi+\eta_{j}-\i\gamma_{j}}\right|
<\!\!<&\,T_{N}^{4\#\lbrace j\mid\gamma_{j}<2|Im\xi|+2\rbrace}\\
\leq&\,T_{N}^{4\tilde{c}(\Gamma)(2(T_{N}+cT_{N}^{-1}/6)+2)^{2}}\\
<\!\!<&\,e^{16\tilde{c}(\Gamma)T_{N}^{2}\ln T_{N}}
\end{split}
\end{equation}
for some positive constant $\tilde{c}(\Gamma)$, where the bound in the second line follows from \eqref{resbound}. Together with Lemma \ref{scatbound} we obtain
\begin{equation}
\sup_{r\in\gamma_{N}(t)}|\varphi(\tfrac{1}{2}-\i r)|<\!\!<e^{16cT_{N}^{2}\ln T_{N}}.
\end{equation}
For the residual term we have
\begin{equation}
\begin{split}
\left|\frac{E(z_{0},\tfrac{1}{2}+\i\rho_{N}(t))E(z_{0},\tfrac{1}{2}-\i\rho_{N}(t))}{2\rho_{N}(t)}\right|
=&\;\frac{|E(z_{0},\tfrac{1}{2}+t+\i T_{N})|^{2}|\varphi(\tfrac{1}{2}-t-\i T_{N})|}{2|T_{N}-\i t|}\\
<\!\!<&\; T_{N}^{3}e^{6T_{N}+16cT_{N}^{2}\ln T_{N}}
\end{split}
\end{equation}
where we have used \eqref{bound3} and \eqref{bound4}. Since we know that the sum in \eqref{rep} is polynomially bounded in $T_{N}$ the result follows straightaway.
\end{proof}

\section{The trace formula for surfaces with one cusp}

As before the key idea in the proof of the trace formula is to exploit the properties of the relative zeta function $S^{\alpha,z_{0}}(s)$. In particular this means the zeros and poles of $S^{\alpha,z_{0}}(s)$ which are associated with the perturbed and unperturbed discrete spectrum. We also need to make use of its link with the perturbed Eisenstein series via the scattering coefficient 
\begin{equation}
\varphi_{\alpha,z_{0}}(s)=\varphi(s)\frac{S_{\alpha,z_{0}}(1-s)}{S_{\alpha,z_{0}}(s)}.
\end{equation}
Let $B(T)=[\i\sigma,-\i\sigma]\times[-T,T]$ such that $\partial B(T)$ does not contain any zeros or poles of $S_{\alpha,z_{0}}(s)$. Similarly to the compact case the strategy in the proof of the trace formula is to first of all prove a truncated version. We then absorb the resonances into an integral along the critical line. Finally we prove the necessary bounds to justify the limit $T\to\infty$. Throughout the proof we shall assume that $\alpha\in\RR\backslash\lbrace0\rbrace$ is generic in the sense that there are only finitely many new perturbed eigenvalues (i. e. ignoring possible inherited degenerate eigenvalues of $\Delta$). Proposition \ref{finite} shows that the exceptional set is at most countable. So once we have obtained the proof of the trace formula for generic $\alpha$ we can extend the result to any $\alpha$ by continuity. Here it will be crucial that any resonances near the critical line are controlled by Proposition \ref{lim}.

Let us begin by deriving an expression involving the resonances. We can express the sum over perturbed and unperturbed resonances in the strip $0<\Im\rho<\sigma$ in terms of the following integral. We also pick up the small perturbed eigenvalues, since those are encoded in the perturbation factor of the perturbed scattering coefficient $\varphi_{\alpha,z_{0}}(s)/\varphi(s)=\theta_{\alpha,z_{0}}(s)$.
\begin{lem}\label{formula2}
\begin{equation}
\begin{split}
&-\sum_{r_{j}\in B(T)}h(r_{j})+\sum_{r^{\alpha}_{j}\in B(T)}h(r^{\alpha}_{j}) -\sum_{\rho^{\alpha}_{j}\in(0,-\i\sigma)}h(\rho^{\alpha}_{j})\\
=\,&\frac{1}{2\pi\i}\left[\int_{-\i\sigma-T}^{-\i\sigma+T}-\int_{-T}^{T}+\int_{-T}^{-\i\sigma-T}+\int_{-\i\sigma+T}^{T}\right]h(\rho)\frac{d}{d\rho}\log\theta_{\alpha,z_{0}}(\tfrac{1}{2}+\i\rho)d\rho
\end{split}
\end{equation}
\end{lem}
\begin{proof}
We prove the result by contour integration along the boundary of the box $[0,-\i\sigma]\times[-T,T]$. Since we are integrating $h$ against the logarithmic derivative of the meromorphic function $\theta^{\alpha,z_{0}}\left(\tfrac{1}{2}+\i\rho\right)$, when we shift across the contour from the line $\Im\rho=-\sigma$ to the real line we collect a residue $-kh(\rho')$ at every pole $\rho'$ of order $k$ of $\theta^{\alpha,z_{0}}\left(\tfrac{1}{2}+\i\rho\right)$ and a residue $kh(\rho')$ at every zero $\rho'$ of order $k$. The function
\begin{equation}
\theta_{\alpha,z_{0}}(\tfrac{1}{2}+\i\rho)=\frac{S_{\alpha,\,z_{0}}(\tfrac{1}{2}-\i\rho)}{S_{\alpha,\,z_{0}}(\tfrac{1}{2}+\i\rho)}
\end{equation}
has no zeros or poles corresponding to the unperturbed eigenvalues $\lbrace\rho_{j}\rbrace_{j}$ since the corresponding poles cancel because of symmetry
\begin{equation}
0=S_{\alpha,\,z_{0}}(\tfrac{1}{2}+\i\rho_{j})^{-1}=S_{\alpha,\,z_{0}}(\tfrac{1}{2}-\i\rho_{j})^{-1}
\end{equation}
where the zeros are of order one. $S_{\alpha,\,z_{0}}(\tfrac{1}{2}-\i\rho)$ has poles corresponding to the unperturbed resonances at $\lbrace-r_{j}\rbrace_{j}$ inside the box $[0,-\i\sigma]\times[-T,T]$ and for each of them we collect a residue $-h(-r_{j})=-h(r_{j})$. It also has zeros corresponding to perturbed resonances at $\lbrace -r^{\alpha}_{j}\rbrace_{j}$ inside the box $[0,-\i\sigma]\times[-T,T]$. For each of them we pick up a residue $h(-r^{\alpha}_{j})=h(r^{\alpha}_{j})$. The function $S_{\alpha,\,z_{0}}(\tfrac{1}{2}+\i\rho)^{-1}$ has simple poles in the interval $(0,-\i\sigma)$ corresponding to the small perturbed eigenvalues at $\lbrace\rho^{\alpha}_{j}\rbrace_{j=-M}^{-1}$. For each of them we pick up a residue $-h(-\rho^{\alpha}_{j})=-h(\rho^{\alpha}_{j})$. The result follows from Cauchy's residue theorem where we count multiplicities separately.
\end{proof}

We use the previous Lemma to prove a truncated trace formula. In the following let $\sigma\geq\tfrac{1}{2}$, $\delta>0$ and $h\in H_{\sigma,\delta}$.
\begin{prop}
Let $\delta_{\Gamma}=1$ if $\lambda=\tfrac{1}{4}$ is not an eigenvalue of the Laplacian and $\delta_{\Gamma}=0$ otherwise. Then we have
\begin{equation}\label{prop16}
\begin{split}
\sum_{\rho^{\alpha}_{j}\in B(T)}h(\rho^{\alpha}_{j})-\sum_{\rho_{j}\in B(T)}h(\rho_{j})
=\,&\frac{1}{2\pi}\left[\int_{-\i\sigma-T}^{-\i\sigma+T}+\int_{-T}^{-T-\i\sigma}+\int_{T-\i\sigma}^{T}\right]h(\rho) \frac{S'_{\alpha,\,z_{0}}}{S_{\alpha,\,z_{0}}}(\tfrac{1}{2}+\i\rho)d\rho\\
&+\tfrac{1}{2}\delta_{\Gamma}h(0)+\frac{1}{4\pi}\int_{-T}^{T}h(\rho)\frac{\theta'_{\alpha,z_{0}}}{\theta_{\alpha,z_{0}}}(\tfrac{1}{2}+\i\rho)d\rho.
\end{split}
\end{equation}
\end{prop}
\begin{proof}
We introduce the symmetric function
\begin{equation}
\Psi(s)=S_{\alpha,\,z_{0}}(s)S_{\alpha\,z_{0}}(1-s).
\end{equation}
and claim that
\begin{equation}
\begin{split}
&\frac{1}{\pi\i}\left[\int_{-\i\sigma-T}^{-\i\sigma+T}+\int_{-\i\sigma+T}^{\i\sigma+T}\right]h(\rho)\frac{d}{d\rho}\log\Psi(\tfrac{1}{2}+\i\rho)d\rho\\
=\,&-2\delta_{\Gamma}h(0)+4\sum_{\rho^{\alpha}_{j}\in (-T,T)}h(\rho^{\alpha}_{j})-4\sum_{\rho_{j}\in B(T)}h(\rho_{j})-2\sum_{r_{j}\in B(T)}h(r_{j})\\
&+2\sum_{\rho^{\alpha}_{j}\in (0,-\i\sigma)}h(\rho^{\alpha}_{j})+2\sum_{r^{\alpha}_{j}\in B(T)}h(r^{\alpha}_{j}).
\end{split}
\end{equation}
The above identity is easily proven by contour integration along the boundary of the box $[\i\sigma,-\i\sigma]\times[-T,T]$. We simply have to count all zeros and poles of $\Psi(s)$ that lie inside the box. From Lemma \ref{prop} we know that $S_{\alpha,\,z_{0}}(\tfrac{1}{2}+\i\rho)$ has simple poles at $\lbrace\rho_{j}\rbrace_{\rho_{j}\in B(T)}$ and $\lbrace-\rho_{j}\rbrace_{\rho_{j}\in B(T)}$ corresponding to eigenvalues $\lbrace\lambda_{j}\rbrace_{\rho_{j}\in B(T)}$, $\lambda_{j}=\tfrac{1}{4}+\rho_{j}^{2}$. They have the same residues, ie. $-h(\rho_{j})=-h(-\rho_{j})$ and $-h(\rho_{-j})=-h(-\rho_{-j})$, because of evenness of $h$. So we collect residues $-2\sum_{\rho_{j}\in B(T)}h(\rho_{j})$ altogether. We also collect residues $-\sum_{r_{j}\in B(T)}h(r_{j})$ at the unperturbed resonances $\lbrace r_{j}\rbrace_{r_{j}\in B(T)}$ that lie inside the box.

Also from Lemma \ref{prop} we know that $S_{\alpha,\,z_{0}}(s)$ has zeros of order one at $\lbrace\rho^{\alpha}_{j}\rbrace_{\rho^{\alpha}_{j}\in (-T,T)}$ and $\lbrace-\rho^{\alpha}_{j}\rbrace_{\rho^{\alpha}_{j}\in (-T,T)}$ corresponding to perturbed eigenvalues $\lbrace\lambda^{\alpha}_{j}\rbrace_{\rho^{\alpha}_{j}\in (-T,T)}$, $\lambda^{\alpha}_{j}=\tfrac{1}{4}+{\rho^{\alpha}_{j}}^{2}\geq\tfrac{1}{4}$, since the condition for the existence of a zero on the critical line implies $E(z_{0},\tfrac{1}{2}+\i\rho^{\alpha}_{j})=0$ and therefore by \eqref{FE}
\begin{equation}
0=S_{\alpha,\,z_{0}}(\tfrac{1}{2}+\i\rho^{\alpha}_{j})=S_{\alpha,\,z_{0}}(\tfrac{1}{2}-\i\rho^{\alpha}_{j}),
\end{equation}
so the residues are the same and altogether we collect $2\sum_{\rho^{\alpha}_{j}\in(-T,T)}h(\rho^{\alpha}_{j})$. 

It is crucial to note that in the case of the small perturbed eigenvalues it is not a requirement that $E(z_{0},s)$ vanishes at the corresponding point. Therefore we may have only one zero of order one corresponding to a small eigenvalue $\tfrac{1}{4}+{\rho^{\alpha}_{-j}}^{2}$ at $\rho^{\alpha}_{-j}\in(0,-\i\sigma)$. Altogether we collect residues $\sum_{\rho^{\alpha}_{j}\in(0,-\i\sigma)}h(\rho^{\alpha}_{j})$. For the perturbed resonances $\lbrace r^{\alpha}_{j}\rbrace_{r^{\alpha}_{j}\in B(T)}$, which correspond to zeros of $S_{\alpha,z_{0}}(\tfrac{1}{2}+\i\rho)$ in the halfplane $\Im\rho>0$, we collect residues $\sum_{r^{\alpha}_{j}\in B(T)}h(r^{\alpha}_{j})$. Finally, if $\lambda=\tfrac{1}{4}$ is an eigenvalue we collect a residue $-2h(0)$ at $s=\tfrac{1}{2}$ since the pole is of order 2. If $\lambda=\tfrac{1}{4}$ is not an eigenvalue, then we collect a residue $-h(0)$ and we shall write this separately in the final summation over residues, since it has no interpretation with respect to an eigenvalue. 

By evenness of $h$ and $\Psi$ we can write
\begin{equation}
\begin{split}
\frac{1}{2\pi\i}\int_{\partial B}h(\rho)\frac{d}{d\rho}\log\Psi(\tfrac{1}{2}+\i\rho)d\rho
=\,&\frac{1}{\pi\i}\left[\int_{-\i\sigma-T}^{-\i\sigma+T}+\int_{-\i\sigma+T}^{\i\sigma+T}\right]h(\rho)\frac{d}{d\rho}\log\Psi(\tfrac{1}{2}+\i\rho)d\rho
\end{split}
\end{equation}
By Cauchy's residue theorem the RHS equals the sum over the residues collected. By evenness of $h$ and $\Psi(\tfrac{1}{2}+\i\rho)=S_{\alpha,\,z_{0}}(\tfrac{1}{2}+\i\rho) S_{\alpha,\,z_{0}}(\tfrac{1}{2}-\i\rho)$ the sum over the residues equals exactly twice the sums over residues listed above. This proves the above claim. We now apply the previous Lemma in order to express the resonances which appear above in terms of an integral and a finite sum over the small perturbed eigenvalues. We have
\begin{equation}
\begin{split}
&-2\sum_{r_{j}\in B(T)}h(r_{j})+2\sum_{r^{\alpha}_{j}\in B(T)}h(r^{\alpha}_{j})\\
=\,&\frac{1}{\pi\i}\left[\int_{-\i\sigma-T}^{-\i\sigma+T}-\int_{-T}^{T}+\int_{-T}^{-\i\sigma-T}+\int_{-\i\sigma+T}^{T}\right]h(\rho)\frac{d}{d\rho}\log\theta_{\alpha,z_{0}}(\tfrac{1}{2}+\i\rho)d\rho
+2\sum_{\rho^{\alpha}_{j}\in (0,-\i\sigma)}h(\rho^{\alpha}_{j}).
\end{split}
\end{equation}
After substituting and dividing through by $4$ we obtain
\begin{equation}\label{expr01}
\begin{split}
&\frac{1}{4\pi\i}\left[\int_{-\i\sigma-T}^{-\i\sigma+T}+\int_{-\i\sigma+T}^{\i\sigma+T}\right]h(\rho)\frac{d}{d\rho}
\log\Psi(\tfrac{1}{2}+\i\rho)d\rho\\
=&-\tfrac{1}{2}\delta_{\Gamma}h(0)+\sum_{\rho^{\alpha}_{j}\in B(T)}h(\rho^{\alpha}_{j})
-\sum_{\rho_{j}\in B(T)}h(\rho_{j})\\
&+\frac{1}{4\pi\i}\left[\int_{-\i\sigma-T}^{-\i\sigma+T}-\int_{-T}^{T}+\int_{-T}^{-\i\sigma-T}+\int_{-\i\sigma+T}^{T}\right]h(\rho)
\frac{d}{d\rho}\log\theta_{\alpha,z_{0}}(\tfrac{1}{2}+\i\rho)d\rho
\end{split}
\end{equation}
We can rewrite
\begin{equation}
\int_{T}^{\i\sigma+T}h(\rho)\frac{d}{d\rho}\log\Psi(\tfrac{1}{2}+\i\rho)d\rho
=\int_{-T}^{-\i\sigma-T}h(\rho)\frac{d}{d\rho}\log\Psi(\tfrac{1}{2}+\i\rho)d\rho,
\end{equation}
since the integrand is odd. Therefore, observing the relation $\Psi(s)=S_{\alpha,z_{0}}(s)^{2}\theta_{\alpha,z_{0}}(s)$, we have
\begin{equation}
\begin{split}
&\frac{1}{4\pi\i}\int_{-\i\sigma+T}^{\i\sigma+T}h(\rho)\frac{d}{d\rho}
\log\Psi(\tfrac{1}{2}+\i\rho)d\rho
-\frac{1}{4\pi\i}\left[\int_{-T}^{-\i\sigma-T}+\int_{-\i\sigma+T}^{T}\right]h(\rho)
\frac{d}{d\rho}\log\theta_{\alpha,z_{0}}(\tfrac{1}{2}+\i\rho)d\rho\\
=&\;\frac{1}{2\pi\i}\left[\int_{-T}^{-T-\i\sigma}+\int_{T-\i\sigma}^{T}\right]h(\rho)\frac{d}{d\rho}\log S_{\alpha,z_{0}}(\tfrac{1}{2}+\i\rho)d\rho.
\end{split}
\end{equation}
Similarly we have
\begin{equation}
\begin{split}
&\frac{1}{4\pi\i}\int_{-\i\sigma-T}^{-\i\sigma+T}h(\rho)\frac{d}{d\rho}
\log\Psi(\tfrac{1}{2}+\i\rho)d\rho
-\frac{1}{4\pi\i}\int_{-\i\sigma-T}^{-\i\sigma+T}h(\rho)
\frac{d}{d\rho}\log\theta_{\alpha,z_{0}}(\tfrac{1}{2}+\i\rho)d\rho\\
=\;&\frac{1}{2\pi\i}\int_{-\i\sigma-T}^{-\i\sigma+T}h(\rho)\frac{d}{d\rho}\log S_{\alpha,z_{0}}(\tfrac{1}{2}+\i\rho)d\rho.
\end{split}
\end{equation}
So finally we can rewrite equation \eqref{expr01} as
\begin{equation}
\begin{split}
&\frac{1}{2\pi\i}\left[\int_{-\i\sigma-T}^{-\i\sigma+T}+\int_{-T}^{-T-\i\sigma}+\int_{T-\i\sigma}^{T}\right]h(\rho)\frac{d}{d\rho}\log S_{\alpha,z_{0}}(\tfrac{1}{2}+\i\rho)d\rho\\
=\;&-\tfrac{1}{2}\delta_{\Gamma}h(0)+\sum_{\rho^{\alpha}_{j}\in B(T)}h\left(\rho^{\alpha}_{j}\right)-\sum_{\rho_{j}\in B(T)}h\left(\rho_{j}\right)
-\frac{1}{4\pi\i}\int_{-T}^{T}h(\rho)\frac{d}{d\rho}\log\theta_{\alpha,z_{0}}(\tfrac{1}{2}+\i\rho)d\rho,
\end{split}
\end{equation}
which is the desired result.
\end{proof}

For the proof of the trace formula we need to be able to control the term containing the scattering coefficients. We state this result in the following proposition.
\begin{prop}\label{lim}
For any $h\in H_{\sigma,\,\delta}$
\begin{equation}
\frac{1}{4\pi}\int_{-\infty}^{+\infty}h(\rho)\left\{\frac{\varphi_{\alpha,z_{0}}'}{\varphi_{\alpha,z_{0}}}(\tfrac{1}{2}+\i\rho)-\frac{\varphi'}{\varphi}(\tfrac{1}{2}+\i\rho)\right\}d\rho
\end{equation}
converges absolutely.
\end{prop}
\begin{proof}
We recall that the perturbed scattering coefficient is given by
\begin{equation}
\varphi_{\alpha,z_{0}}(s)=\varphi(s)\frac{S_{\alpha,z_{0}}(1-s)}{S_{\alpha,z_{0}}(s)}.
\end{equation}
Consequently,
\begin{equation}
\frac{\varphi_{\alpha,z_{0}}'}{\varphi_{\alpha,z_{0}}}(s)-\frac{\varphi'}{\varphi}(s)=\frac{d}{ds}\log\frac{S_{\alpha,z_{0}}(1-s)}{S_{\alpha,z_{0}}(s)},
\end{equation}
and therefore
\begin{equation}
\begin{split}
&\frac{1}{4\pi}\int_{-T}^{T}h(\rho)\left\{\frac{\varphi_{\alpha,z_{0}}'}{\varphi_{\alpha,z_{0}}}(\tfrac{1}{2}+\i\rho)-\frac{\varphi'}{\varphi}(\tfrac{1}{2}+\i\rho)\right\}d\rho\\
=\,&\frac{1}{4\pi\i}\int_{-T}^{T}h(\rho)\frac{d}{d\rho}\log\frac{S_{\alpha,z_{0}}(\tfrac{1}{2}-\i\rho)}{S_{\alpha,z_{0}}(\tfrac{1}{2}+\i\rho)}d\rho\\
=\,&\frac{1}{2\pi\i}h(T)\log\frac{S_{\alpha,z_{0}}(\tfrac{1}{2}-\i T)}{S_{\alpha,z_{0}}(\tfrac{1}{2}+\i T)}
-\frac{1}{4\pi\i}\int_{-T}^{T}h'(\rho)\log\frac{S_{\alpha,z_{0}}(\tfrac{1}{2}-\i\rho)}{S_{\alpha,z_{0}}(\tfrac{1}{2}+\i\rho)}d\rho.
\end{split}
\end{equation}
So in order to prove the existence of the limit as $T\to\infty$ we have to estimate
\begin{equation}
\log\frac{S_{\alpha,z_{0}}(\tfrac{1}{2}-\i\rho)}{S_{\alpha,z_{0}}(\frac{1}{2}+\i\rho)}
\end{equation}
for $\rho\in\RR$. We know that
\begin{equation}
S_{\alpha,z_{0}}(\tfrac{1}{2}+\i\rho)=r(\rho)\e^{i\theta(\rho)},
\end{equation}
where $\theta(\rho)=\arg S_{\alpha,z_{0}}(\tfrac{1}{2}+\i\rho)$ and $r(\rho)=\left|S_{\alpha,z_{0}}(\tfrac{1}{2}+\i\rho)\right|$. We recall the functional equation
\begin{equation}
S_{\alpha,z_{0}}(s)=S_{\alpha,z_{0}}(1-s)-\frac{E(z_{0},s)E(z_{0},1-s)}{1-2s},
\end{equation}
and see as a consequence
\begin{equation}
r(\rho)\sin\theta(\rho)=\Im S_{\alpha,z_{0}}(\tfrac{1}{2}+\i\rho)=\frac{1}{4\rho}\left|E(z_{0},\tfrac{1}{2}+\i\rho)\right|^{2}.
\end{equation}
This together with meromorphicity of $S_{\alpha,z_{0}}(s)$ implies $\left|\theta(\rho)\right|\leq\pi$. In fact this bound holds in the strip $-\sigma\leq\Im\rho\leq0$. To see this observe that for $\Re\rho>0$
\begin{equation}\label{Im}
\begin{split}
\Im S_{\alpha,z_{0}}(\tfrac{1}{2}+\i\rho)
=\;&-\frac{\Re\rho\;\Im\rho}{2\pi}\int_{0}^{+\infty}\frac{|E(z_{0},\tfrac{1}{2}+\i r)|^{2}dr}{(\Re\rho)^{2}-r^{2}-(\Im\rho)^{2})^{2}+4(\Re\rho)^{2}(\Im\rho)^{2}}\\
&-\Re\rho\;\Im\rho\sum_{j=-M}^{\infty}\frac{|\varphi_{j}(z_{0})|^{2}} {((\Re\rho)^{2}-\rho_{j}^{2}-(\Im\rho)^{2})^{2}+4(\Re\rho)^{2}(\Im\rho)^{2}}\geq0
\end{split}
\end{equation}
and from the functional equation we see
\begin{equation}
\lim_{\Im\rho\to0}\Im S_{\alpha,z_{0}}(\tfrac{1}{2}+\i\rho)=\frac{|E(z_{0},\tfrac{1}{2}+\i\Re\rho)|^{2}}{2\Re\rho}\geq0.
\end{equation}
From this observation and since $S_{\alpha,z_{0}}(\tfrac{1}{2}+\i\rho)$ is meromorphic, it follows that $S_{\alpha,z_{0}}(\tfrac{1}{2}+\i\cdot)$ maps any smooth curve in the strip $-\sigma\leq\Im\rho\leq0$ which contains no poles of $S_{\alpha,z_{0}}(\tfrac{1}{2}+\i\rho)$ to a smooth curve in the upper halfplane. In particular this implies that the winding number of any such curve is 0. This implies $0\leq\arg S_{\alpha,z_{0}}(\tfrac{1}{2}+\i\rho)\leq\pi$ for $\Re\rho>0$ and $-\sigma\leq\Im\rho\leq0$ away from poles. In the same way one sees that $-\pi\leq\arg S_{\alpha,z_{0}}(\tfrac{1}{2}+\i\rho)\leq0$ for $\Re\rho<0$ and $-\sigma\leq\Im\rho\leq0$ away from poles.

So we have
\begin{equation}
\left|\log\frac{S_{\alpha,z_{0}}(\tfrac{1}{2}-\i\rho)}{S_{\alpha,z_{0}}(\tfrac{1}{2}+\i\rho)}\right|=\left|\log\frac{r(-\rho)\e^{-\i\theta(\rho)}}{r(\rho)\e^{\i\theta(\rho)}}\right|=2\left|\theta(\rho)\right|\leq 2\pi,
\end{equation}
where we have used
\begin{equation}
r(\rho)=|S_{\alpha,z_{0}}(\tfrac{1}{2}+\i\rho)|=|\overline{S_{\alpha,z_{0}}(\tfrac{1}{2}+\i\rho)}|
=|S_{\alpha,z_{0}}(\tfrac{1}{2}-\i\rho)|=r(-\rho).
\end{equation}
and the fact that $S_{\alpha,z_{0}}(\tfrac{1}{2}+\i\rho)S_{\alpha,z_{0}}(\tfrac{1}{2}-\i\rho)$ is analytic on the real line.

Now we can estimate the following tail
\begin{equation}
\begin{split}
&\frac{1}{4\pi}\left|\int_{T}^{\infty}h(\rho)\left\{\frac{\varphi_{\alpha,z_{0}}'}{\varphi_{\alpha,z_{0}}}(\tfrac{1}{2}+\i\rho)-\frac{\varphi'}{\varphi}(\tfrac{1}{2}+\i\rho)\right\}d\rho\right|\\
=\,&\frac{1}{4\pi}\left|\int_{T}^{\infty}h(\rho)\frac{d}{d\rho}\log\frac{S_{\alpha,z_{0}}(\tfrac{1}{2}-\i\rho)}{S_{\alpha,z_{0}}(\tfrac{1}{2}+\i\rho)}d\rho\right|\\
\leq\,&\frac{1}{4\pi}\left|h(T)\right|\left|\log\frac{S_{\alpha,z_{0}}(\tfrac{1}{2}-\i T)}{S_{\alpha,z_{0}}(\tfrac{1}{2}+\i T)}\right|
+\frac{1}{4\pi}\left|\int_{T}^{\infty}h'(\rho)\log\frac{S_{\alpha,z_{0}}(\tfrac{1}{2}-\i\rho)}{S_{\alpha,z_{0}}(\tfrac{1}{2}+\i\rho)}d\rho\right|\\
\leq\,&\frac{1}{4\pi}\left|h(T)\right|\left|\log\frac{S_{\alpha,z_{0}}(\tfrac{1}{2}-\i T)}{S_{\alpha,z_{0}}(\tfrac{1}{2}+\i T)}\right|
+\frac{1}{4\pi}\int_{T}^{\infty}|h'(\rho)|\left|\log\frac{S_{\alpha,z_{0}}(\tfrac{1}{2}-\i\rho)}{S_{\alpha,z_{0}}(\tfrac{1}{2}+\i\rho)}\right|d\rho.
\end{split}
\end{equation}
We only have to estimate one tail because the relation
\begin{equation}
\frac{\varphi'(s)}{\varphi(s)}=\frac{\varphi'(1-s)}{\varphi(1-s)}
\end{equation}
implies that the integrand is even on the critical line. As an immediate consequence of the decay of $h$ and the above bound on the logarithm we conclude
\begin{equation}
\lim_{T\to\infty}\frac{1}{4\pi}\left|h(T)\right|\left|\log\frac{S_{\alpha,z_{0}}(\tfrac{1}{2}-\i T)}{S_{\alpha,z_{0}}(\tfrac{1}{2}+\i T)}\right|=0.
\end{equation}
$h\in H_{\sigma,\delta}$ implies that we have the estimate
\begin{equation}
h'(\rho)<\!\!<(1+|\Re\rho|)^{-2-\delta}
\end{equation}
uniformly in the strip $|\Im\rho|\leq\sigma-\epsilon$ for any $0<\epsilon<\sigma$. So we obtain
\begin{equation}
\lim_{T\to\infty}\frac{1}{4\pi}\int_{T}^{\infty}|h'(\rho)|\left|\log\frac{S_{\alpha,z_{0}}(\tfrac{1}{2}-\i\rho)}{S_{\alpha,z_{0}}(\tfrac{1}{2}+\i\rho)}\right|d\rho=0.
\end{equation}
We conclude
\begin{equation}
\begin{split}
&\lim_{T\to\infty}\frac{1}{4\pi}\int_{-T}^{T}h(\rho)\left\{\frac{\varphi_{\alpha,z_{0}}'}{\varphi_{\alpha,z_{0}}}(\tfrac{1}{2}+\i\rho)-\frac{\varphi'}{\varphi}(\tfrac{1}{2}+\i\rho)\right\}d\rho\\
=&\frac{1}{4\pi}\int_{-\infty}^{\infty}h(\rho)\left\{\frac{\varphi_{\alpha,z_{0}}'}{\varphi_{\alpha,z_{0}}}(\tfrac{1}{2}+\i\rho)-\frac{\varphi'}{\varphi}(\tfrac{1}{2}+\i\rho)\right\}d\rho.
\end{split}
\end{equation}
\end{proof}

With the previous work we are able to justify the existence of the limit $T\to\infty$ of the truncated trace formula \eqref{prop16}.
\begin{cor}
Let $\lbrace T_{n}\rbrace$ be any sequence in $\RR_{+}$ in between eigenvalues $\lbrace\rho_{j}\rbrace_{j=0}^{+\infty}$, $\lbrace\rho_{j}^{\alpha}\rbrace_{j=0}^{+\infty}$ and real parts of resonances $\lbrace\Re r_{j}\rbrace_{j=0}^{+\infty}$ and accumulating at $+\infty$. Then the limit
\begin{equation}\label{bdterms}
\frac{1}{2\pi\i}\lim_{n\to\infty}\left\{\int_{-\i\sigma+T_{n}}^{T_{n}}+\int_{-T_{n}}^{-\i\sigma-T_{n}}\right\}h(\rho)\frac{d}{d\rho}\log S_{\alpha,\,z_{0}}(\tfrac{1}{2}+\i\rho)d\rho
\end{equation}
exists.
\end{cor}
\begin{proof}
We have the identity
\begin{equation}
\begin{split}
&\frac{1}{2\pi\i}\left\{\int_{-\i\sigma+T_{n}}^{T_{n}}+\int_{-T_{n}}^{-\i\sigma-T_{n}}\right\}h(\rho)\frac{d}{d\rho}\log S_{\alpha,\,z_{0}}(\tfrac{1}{2}+\i\rho)d\rho\\
=\;&\sum_{\rho^{\alpha}_{j}\in B(T_{n})}h(\rho^{\alpha}_{j})-\sum_{\rho_{j}\in B(T_{n})}h(\rho_{j})
-\frac{1}{2\pi\i}\int_{-\i\sigma-T_{n}}^{-\i\sigma+T_{n}}h(\rho)\frac{d}{d\rho}\log S_{\alpha,\,z_{0}}(\tfrac{1}{2}+\i\rho)d\rho\\
&-\tfrac{1}{2}\delta_{\Gamma}h(0)-\frac{1}{4\pi}\int_{-T_{n}}^{T_{n}}h(\rho)\frac{\theta_{\alpha,z_{0}}'}{{\theta_{\alpha,z_{0}}}}(\tfrac{1}{2}+\i\rho)d\rho.
\end{split}
\end{equation}
The standard upper bound on the number of eigenvalues and the decay of $h$ imply that the sums over the eigenvalues converge absolutely. By Proposition \ref{lim} the limit
\begin{equation}
\lim_{T\to\infty}\frac{1}{4\pi}\int_{-T}^{T}h(\rho)\frac{\theta_{\alpha,z_{0}}'}{{\theta_{\alpha,z_{0}}}}(\tfrac{1}{2}+\i\rho)d\rho
=\frac{1}{4\pi}\int_{-\infty}^{\infty}h(\rho)\frac{\theta_{\alpha,z_{0}}'}{{\theta_{\alpha,z_{0}}}}(\tfrac{1}{2}+\i\rho)d\rho
\end{equation}
exists. Finally, for $\Im\rho=-\sigma$, we have as a consequence of Lemma 9
\begin{equation}
S_{\alpha,z_{0}}(\tfrac{1}{2}+\i\rho)=m\psi(\tfrac{1}{2}+\i\rho)+O(1)
\end{equation}
which implies, for $\Im\rho=-\sigma$ and $|\rho|$ large, using boundedness of $\arg S_{\alpha,z_{0}}(\tfrac{1}{2}+\i\rho)$ away from poles,
\begin{equation}\label{Sbound}
\log S_{\alpha,\,z_{0}}(\tfrac{1}{2}+\i\rho)=O(\log\log|\rho|).
\end{equation}
From an integration by parts we have
\begin{equation}
\begin{split}
\int_{-\i\sigma-T}^{-\i\sigma+T}h(\rho)\frac{d}{d\rho}\log S_{\alpha,\,z_{0}}(\tfrac{1}{2}+\i\rho)d\rho
=\;&h(-\i\sigma+T)\log S_{\alpha,z_{0}}(\tfrac{1}{2}+\sigma+\i T)\\
&-h(-\i\sigma-T)\log S_{\alpha,z_{0}}(\tfrac{1}{2}+\sigma-\i T)\\
&-\int_{-\i\sigma-T}^{-\i\sigma+T}h'(\rho)\log S_{\alpha,\,z_{0}}(\tfrac{1}{2}+\i\rho)d\rho
\end{split}
\end{equation}
which implies
\begin{equation}\label{limbound}
\lim_{T\to\infty}\int_{-\i\sigma-T}^{-\i\sigma+T}h(\rho)\frac{d}{d\rho}\log S_{\alpha,\,z_{0}}(\tfrac{1}{2}+\i\rho)d\rho
=\lim_{T\to\infty}\int_{-\i\sigma-T}^{-\i\sigma+T}h'(\rho)\log S_{\alpha,\,z_{0}}(\tfrac{1}{2}+\i\rho)d\rho
\end{equation}
since
\begin{equation}
\lim_{T\to+\infty}h(-\i\sigma+T)\log S_{\alpha,z_{0}}(\tfrac{1}{2}+\sigma+\i T)
=\lim_{T\to+\infty}h(-\i\sigma-T)\log S_{\alpha,z_{0}}(\tfrac{1}{2}+\sigma-\i T)=0.
\end{equation}
Now the integral on the RHS of \eqref{limbound} clearly converges because of \eqref{Sbound} and $h\in H_{\sigma,\delta}$. Consequently the limit of the boundary term exists as we stretch the box to infinity.
\end{proof}

We can now apply Proposition \ref{seqbound} and Lemma \ref{testf} to derive a bound on a sequence of integrals which will eventually lead to the required bound on the corresponding sequence of boundary terms \eqref{bdterms} which we require in order to show that its limit vanishes.
\begin{prop}\label{prop24}
There exists $\lbrace T_{N(j)}\rbrace_{j}\subset\RR_{+}$, $\lim_{j\to\infty}T_{N(j)}=\infty$, such that for any $\epsilon>0$ we have
\begin{equation}
\int_{T_{N(j)}-\i\sigma}^{T_{N(j)}}\big|\log|S_{\alpha,z_{0}}(\tfrac{1}{2}+\i\rho)|\big||d\rho|<\!\!<_{\epsilon}T_{N(j)}^{2+\epsilon}.
\end{equation}
\end{prop}
\begin{proof}
Consider the sequence $\lbrace T_{N}\rbrace_{N}$ of Proposition \ref{seqbound}. We choose the subseqence $\lbrace T_{N(j)}\rbrace_{j}$ of Lemma \ref{testf}. Fix any $\epsilon>0$. We pick the test function $h_{\epsilon}\in H_{\sigma, \epsilon}$ given in Lemma \ref{testf}. We have by integration by parts
\begin{equation}\label{byparts}
\begin{split}
&\left\{\int_{T_{N(j)}-\i\sigma}^{T_{N(j)}}+\int_{-T_{N(j)}}^{-T_{N(j)}-\i\sigma}\right\}h_{\epsilon}(\rho)\frac{d}{d\rho}\log  S_{\alpha,z_{0}}(\tfrac{1}{2}+\i\rho)d\rho\\
=\;&h_{\epsilon}(T_{N(j)})\log S_{\alpha,z_{0}}(\tfrac{1}{2}+\i T_{N(j)})
-h_{\epsilon}(T_{N(j)}-\i\sigma)\log S_{\alpha,z_{0}}
(\tfrac{1}{2}+\sigma+\i T_{N(j)})\\
&+h_{\epsilon}(-T_{N(j)}-\i\sigma)\log S_{\alpha,z_{0}}
(\tfrac{1}{2}+\sigma-\i T_{N(j)})
-h_{\epsilon}(-T_{N(j)})\log S_{\alpha,z_{0}}(\tfrac{1}{2}-\i T_{N(j)})\\
&-\left\{\int_{T_{N(j)}-\i\sigma}^{T_{N(j)}}+\int_{-T_{N(j)}}^{-T_{N(j)}-\i\sigma}\right\}h_{\epsilon}'(\rho)\log S_{\alpha,z_{0}}(\tfrac{1}{2}+\i\rho)d\rho.
\end{split}
\end{equation}
We know that the LHS converges. Recall
\begin{equation}\label{log}
|S_{\alpha,z_{0}}(\tfrac{1}{2}+\sigma\pm\i T_{N(j)})|\sim\psi(\tfrac{1}{2}+\sigma\pm\i T_{N(j)})\sim\log T_{N(j)},\qquad j\to\infty.
\end{equation}
We have
\begin{equation}
h_{\epsilon}(\rho)<\!\!<(1+|\Re\rho|)^{-2-\epsilon}
\end{equation}
uniformly for $|\Im\rho|\leq\sigma$. It follows that the second and third term converge to zero as $j\to\infty$. For the first term and very similarly for the fourth term we have
\begin{equation}
\begin{split}
&h_{\epsilon}(T_{N(j)})\log S_{\alpha,z_{0}}(\tfrac{1}{2}+\i T_{N(j)})\\
=\;&h_{\epsilon}(T_{N(j)})\log|S_{\alpha,z_{0}}(\tfrac{1}{2}+\i T_{N(j)})|
+\i h_{\epsilon}(T_{N(j)})\arg S_{\alpha,z_{0}}(\tfrac{1}{2}+\i T_{N(j)}).
\end{split}
\end{equation}
We can rewrite the fifth term as
\begin{equation}
\begin{split}
&\left\{\int_{T_{N(j)}-\i\sigma}^{T_{N(j)}}+\int_{-T_{N(j)}}^{-T_{N(j)}-\i\sigma}\right\}h_{\epsilon}'(\rho)\log |S_{\alpha,z_{0}}(\tfrac{1}{2}+\i\rho)|d\rho\\
&+\i\left\{\int_{T_{N(j)}-\i\sigma}^{T_{N(j)}}+\int_{-T_{N(j)}}^{-T_{N(j)}-\i\sigma}\right\}h_{\epsilon}'(\rho)\arg S_{\alpha,z_{0}}(\tfrac{1}{2}+\i\rho)d\rho
\end{split}
\end{equation}

Recall that $\arg S_{\alpha,z_{0}}(\tfrac{1}{2}+\i\rho)$ is bounded away from poles in $-\sigma\leq\Im\rho\leq0$. This implies $\arg S_{\alpha,z_{0}}(\tfrac{1}{2}+\i\rho)<\!\!<1$ for $\rho\in[T_{N}(j),T_{N}(j)-\i\sigma]$, as we recall that the sequence $\lbrace T_{N(j)}\rbrace_{j}$ is chosen such that the intervals do not contain any poles of $S_{\alpha,z_{0}}(\tfrac{1}{2}+\i\rho)$. It hence follows
\begin{equation}
h_{\epsilon}(\pm T_{N(j)})\arg(S_{\alpha,z_{0}}(\tfrac{1}{2}\pm\i T_{N(j)}))<\!\!<T_{N(j)}^{-2-\epsilon}
\end{equation}
and
\begin{equation}
\left\{\int_{T_{N(j)}-\i\sigma}^{T_{N(j)}}+\int_{-T_{N(j)}}^{-T_{N(j)}-\i\sigma}\right\}h_{\epsilon}'(\rho)\arg S_{\alpha,z_{0}}(\tfrac{1}{2}+\i\rho)d\rho<\!\!<T_{N}^{-2-\epsilon}.
\end{equation}

Combining all this we conclude that the limit of
\begin{equation}
\begin{split}
&\lbrace h_{\epsilon}(T_{N(j)})-h_{\epsilon}(-T_{N(j)})\rbrace\log|S_{\alpha,z_{0}}(\tfrac{1}{2}+\i T_{N(j)})|\\
&-\left\{\int_{T_{N(j)}-\i\sigma}^{T_{N(j)}}+\int_{-T_{N(j)}}^{-T_{N(j)}-\i\sigma}\right\}h_{\epsilon}'(\rho)\log |S_{\alpha,z_{0}}(\tfrac{1}{2}+\i\rho)|d\rho
\end{split}
\end{equation}
which we can rewrite as
\begin{equation}\label{rearr}
\begin{split}
=\;&\int_{T_{N(j)}-\i\sigma}^{T_{N(j)}}\lbrace h'_{\epsilon}(\rho)-h'_{\epsilon}(-\bar{\rho})\rbrace\log |S_{\alpha,z_{0}}(\tfrac{1}{2}+\i\rho)|d\rho\\
=\;&-2\i\int_{0}^{\sigma}\Re h_{\epsilon}'(T_{N(j)}+\i(r-\sigma))\log
|S_{\alpha,z_{0}}(\tfrac{1}{2}+\sigma-r+\i T_{N(j)})|dr
\end{split}
\end{equation}
exists as $j\to\infty$. We have used evenness of $h_{\epsilon}$ and $h'_{\epsilon}(-\bar{\rho})=-h'_{\epsilon}(\bar{\rho}) =-\overline{h'_{\epsilon}(\rho)}$. 

From Proposition \ref{seqbound} we know that there exists a positive constant $c(\Gamma)>0$ such that 
\begin{equation}
\log \lbrace c(\Gamma)e^{-16C_{2}(\Gamma)T_{N(j)}^{2}\ln T_{N(j)}}|S_{\alpha,z_{0}}(\tfrac{1}{2}+\i\rho)|\rbrace\leq0
\end{equation} 
for all $j$ and $\rho\in[T_{N(j)},T_{N(j)}-\i\sigma]$. From Lemma 25 we have for $\rho\in[T_{N(j)}-\i\sigma,T_{N(j)}]$ and $T_{N(j)}$ large that $-\Re h_{\epsilon}'(\rho)=|\Re h_{\epsilon}'(\rho)|>\!\!>T_{N(j)}^{-2-\epsilon}$ uniformly in $j$ and $\rho$. Hence
\begin{equation}
\begin{split}
&T_{N(j)}^{-2-\epsilon}\int_{0}^{\sigma}\left|\log\left\{ c(\Gamma)e^{-16C_{2}(\Gamma)T_{N(j)}^{2}\ln T_{N}}|S_{\alpha,z_{0}}(\tfrac{1}{2}+\sigma-r+\i T_{N(j)})|\right\}\right|dr\\
<\!\!<\;&\int_{0}^{\sigma}|\Re h_{\epsilon}'(T_{N(j)}+\i(r-\sigma))|\;\left|\log\left\{ c(\Gamma)e^{-16C_{2}(\Gamma)T_{N(j)}^{2}\ln T_{N(j)}}|S_{\alpha,z_{0}}(\tfrac{1}{2}+\sigma-r+\i T_{N(j)})|\right\}\right|dr\\
=\;&\left|\int_{0}^{\sigma}\Re h_{\epsilon}'(T_{N(j)}+\i(r-\sigma))\;\log\left\{ c(\Gamma)e^{-16C_{2}(\Gamma)T_{N(j)}^{2}\ln T_{N(j)}}|S_{\alpha,z_{0}}(\tfrac{1}{2}+\sigma-r+\i T_{N(j)})|\right\} dr\right|
\end{split}
\end{equation}
which converges as $j\to\infty$. It follows for any $\delta>0$
\begin{equation}
\begin{split}
\int_{T_{N(j)}-\i\sigma}^{T_{N(j)}}\big|\log|S_{\alpha,z_{0}}(\tfrac{1}{2}+\i\rho)|\big||d\rho|
=&\int_{0}^{\sigma}\left|\log\left\{c(\Gamma)e^{-16C_{2}(\Gamma)T_{N(j)}^{2}\ln T_{N(j)}}|S_{\alpha,z_{0}}(\tfrac{1}{2}+\sigma-r+\i T_{N(j)})|\right\}\right|dr\\
&+O(T_{N(j)}^{2}\ln T_{N(j)})\\
<\!\!<\;&T_{N(j)}^{2+\delta}.
\end{split}
\end{equation}
\end{proof}

We can now apply Proposition \ref{prop24} to derive the vanishing of the sequence of boundary terms \eqref{bdterms} as in Theorem \ref{van}.
\begin{thm}\label{thm25}
Let $\delta,\epsilon>0$ and $\lbrace T_{N(j)}\rbrace_{N(j)}$ as above. Then for any $h\in H_{\sigma+\epsilon,\delta}$
\begin{equation}
\lim_{T_{N(j)}\to\infty}\int_{T_{N(j)}-\i\sigma}^{T_{N(j)}}h(\rho)\frac{S'_{\alpha,z_{0}}}{S_{\alpha,z_{0}}}{(\tfrac{1}{2}+\i\rho)}d\rho=0.
\end{equation}
\end{thm}
\begin{proof}
We follow the same lines as in Proposition \ref{prop24}. In exactly the same way as in the proof above we obtain the identity \eqref{rearr} for $h\in H_{\sigma+\epsilon,\delta}\subset H_{\sigma,\delta}$. So
\begin{equation}
\begin{split}
&\lim_{j\to\infty}\left\{\int_{T_{N(j)}-\i\sigma}^{T_{N(j)}}+\int_{-T_{N(j)}}^{-T_{N(j)}-\i\sigma}\right\}h(\rho)\frac{d}{d\rho}\log S_{\alpha,z_{0}}(\tfrac{1}{2}+\i\rho)d\rho\\
=&\lim_{j\to\infty}\int_{T_{N(j)}-\i\sigma}^{T_{N(j)}}\lbrace h'(\rho)-h'(-\bar{\rho})\rbrace\log|S_{\alpha,z_{0}}(\tfrac{1}{2}+\i\rho)|d\rho.
\end{split}
\end{equation}
The limit vanishes since
\begin{equation}
\begin{split}
&\int_{T_{N(j)}-\i\sigma}^{T_{N(j)}}|\lbrace h'(\rho)-h'(-\bar{\rho})\rbrace|\big|\log|S_{\alpha,z_{0}}(\tfrac{1}{2}+\i\rho)|\big||d\rho|\\
<\!\!<\;&T_{N(j)}^{-2-\delta}\int_{T_{N(j)}-\i\sigma}^{T_{N(j)}}\big|\log|S_{\alpha,z_{0}}(\tfrac{1}{2}+\i\rho)|\big|d\rho|.
\end{split}
\end{equation}
where we have used Proposition \ref{prop24} and observe that by Cauchy's theorem $h\in H_{\sigma+\epsilon,\delta}$ implies
\begin{equation}
h'(\rho)<\!\!<(1+|\Re\rho|)^{-2-\delta}
\end{equation}
uniformly in $|\Im\rho|\leq\sigma$.
\end{proof}

As an application of Theorem \ref{thm25} we can now prove the trace formula. Let $h\in H_{\sigma,\delta}$ for any $\sigma>\tfrac{1}{2}$ and $\delta>0$. We will make a specific choice of $\sigma$ below. We observe that
\begin{equation}
\begin{split}
\int_{-\i\sigma-T}^{-\i\sigma+T}h(\rho)\frac{S'_{\alpha,\,z_{0}}}{S_{\alpha,\,z_{0}}}(\tfrac{1}{2}+\i\rho)d\rho
=&h(-\i\sigma+T)\log S_{\alpha,z_{0}}(\tfrac{1}{2}+\sigma+\i T)\\
&-h(-\i\sigma-T)\log S_{\alpha,z_{0}}(\tfrac{1}{2}+\sigma-\i T)\\
&-\int_{-\i\sigma-T}^{-\i\sigma+T}h'(\rho)\log S_{\alpha,\,z_{0}}(\tfrac{1}{2}+\i\rho)d\rho
\end{split}
\end{equation}
and, because of \eqref{log}, \eqref{Im} and the decay of $h$,
\begin{equation}
\lim_{T\to\infty}h(-\i\sigma\pm T)\log S_{\alpha,z_{0}}(\tfrac{1}{2}+\sigma\pm\i T)=0.
\end{equation}
Thus we can take $T\to\infty$ in \eqref{prop16} to obtain
\begin{equation}\label{pretrace}
\begin{split}
\sum_{j\geq-M}h(\rho^{\alpha}_{j})-\sum_{j\geq-M}h(\rho_{j})
=\;&-\frac{1}{2\pi}\int_{-\i\sigma-\infty}^{-\i\sigma+\infty}h'(\rho)\log S_{\alpha,z_{0}}(\tfrac{1}{2}+\i\rho)d\rho\\
&+\tfrac{1}{2}\delta_{\Gamma}h(0)+\frac{1}{4\pi}\int_{-\infty}^{+\infty}h(\rho)\frac{\theta'_{\alpha,z_{0}}}{\theta_{\alpha,z_{0}}}(\tfrac{1}{2}+\i\rho)d\rho.
\end{split}
\end{equation}
where we have used Theorem \ref{thm25}, Proposition \ref{lim}, the existence of the limit \eqref{limbound}, the standard bound on the unperturbed cuspidal spectrum and our choice of generic $\alpha$ which ensures that the sum over perturbed eigenvalues is finite. Since $S_{\alpha,z_{0}}(s)$ depends continuously on $\alpha\neq0$ and Proposition \ref{lim} gives us control over the resonances we can extend \eqref{pretrace} to the countable set of non-generic values of $\alpha$. In particular Proposition \eqref{lim} and the fact that the zeros of of $S_{\alpha,z_{0}}(s)$ move continuously under variation of $\alpha\neq0$ imply that the trace over unperturbed eigenvalues will converge absolutely.

The first term on the RHS can now be expanded into an identity term and diffractive orbit terms in exactly the same way as in the compact case. Because of the identity
\begin{equation}
\varphi_{\alpha,\,z_{0}}(s)=\theta_{\alpha,z_{0}}(s)\varphi(s)
\end{equation}
we rewrite the second term as
\begin{equation}
\frac{1}{4\pi}\int_{-\infty}^{\infty}h(\rho)\left\{\frac{\varphi'_{\alpha,z_{0}}}{\varphi_{\alpha,z_{0}}}(\tfrac{1}{2}+\i\rho)-\frac{\varphi'}{\varphi}(\tfrac{1}{2}+\i\rho)\right\}d\rho.
\end{equation}

\section{The diffractive ghost of the sphere}

It is well known (cf. \cite{Mf}, section 6) that the identity term in Selberg's trace formula is closely related to the trace formula on the sphere, the only difference being that one sums over imaginary actions instead of real actions. In this section we want to give a similar interpretation of the leading term in the perturbed trace formula in the case that $z_{0}$ is not an elliptic fixed point.

Denote the 2-dimensional sphere by $\SS^{2}$. One can obtain a trace formula for the singular perturbations $\Delta_{\alpha,z_{0}}$ in a very similar way to the hyperbolic setting. The Green's function on $\SS^{2}$ is given by the same free Green's function as before, where we replace the Riemannian distance $d(\cdot,\cdot)$ on $\HH$ with the Riemannian distance $\tilde{d}(\cdot,\cdot)$ on $\SS^{2}$. The Green's function satisfies
\begin{equation}
(\Delta+w(w-1))G^{\SS^{2}}_{w}(\cdot,w)=\delta_{w}
\end{equation}
on $\SS^{2}$. In particular we note that $\SS^{2}$ is of finite volume as opposed to $\HH$, the Green's function on the surface is thus identical with the free Green's function and we don't need to employ the method of images as in the case of a hyperbolic finite-volume surface. Hence, analogously to the hyperbolic setting the meromorphic function
\begin{equation}
1+\alpha\lim_{z\to z_{0}}\left\{G^{\SS^{2}}_{w}(z,z_{0})-\Re G^{\SS^{2}}_{t}(z,z_{0})\right\}
=1+\alpha\lbrace\psi(w)-\Re\psi(t)\rbrace
\end{equation}
encodes all the information about the perturbed and unperturbed eigenvalues on $\SS^{2}$. Following an analogous reparametrisation as in the hyperbolic case, we obtain a relative zeta function given by
\begin{equation}\label{rz}
S^{\SS^{2}}_{\alpha,z_{0}}(w)=1+\beta(\alpha)\psi(w).
\end{equation}
$S^{\SS^{2}}_{\alpha,z_{0}}(w)$ has poles on the real line at $w=0,-1,-2,\cdots$ corresponding to the eigenvalues of the Laplacian on $\SS^{2}$ $w(w-1)=0,2,6,\cdots$ and zeros at $w=w^{\alpha}_{0},w^{\alpha}_{1},w^{\alpha}_{2}\cdots$ corresponding to the perturbed eigenvalues on $\SS^{2}$. We shall write the eigenvalues in the convenient form $\tfrac{1}{4}-\omega^{2}$ which corresponds to the choice $w=\tfrac{1}{2}+\omega$. These coordinates allow us to symmetrise the relative zeta function \eqref{rz} in a convenient way. For any $\delta>0$ we have after a contour integration and division by $2$
\begin{equation}\label{sphtrace}
\sum_{j=0}^{\infty}\lbrace h(\omega^{\alpha}_{j})-h(\omega_{j})\rbrace
=\frac{1}{4\pi\i}\int_{\Im\omega=-\delta}h(\omega)\frac{d}{d\omega}\log\lbrace(1+\beta\psi(\tfrac{1}{2}+\omega))(1+\beta\psi(\tfrac{1}{2}-\omega))\rbrace d\omega.
\end{equation}
Now, by contour integration, we can rewrite the identity term in the trace formula as
\begin{equation}
\begin{split}
&\frac{1}{2\pi\i}\int_{-\i\tilde{\sigma}-\infty}^{-\i\tilde{\sigma}+\infty}h(\rho)\frac{d}{d\rho}\log(1+\beta\psi(\tfrac{1}{2}+\i\rho))d\rho\\
=\;&-\tilde{\delta}_{\Gamma}h(w_{0}^{\alpha})+\frac{1}{2\pi\i}\int_{-\infty}^{\infty}h(\rho)\frac{d}{d\rho}\log(1+\beta\psi(\tfrac{1}{2}+\i\rho))d\rho
\end{split}
\end{equation}
where $\tilde{\delta}_{\Gamma}=1$, if $w_{0}^{\alpha}\in[\tfrac{1}{2},\tfrac{1}{2}+\tilde{\sigma}]$ and $\tilde{\delta}_{\Gamma}=0$ otherwise. Note that we have no zeros off the real line, since $\Im\psi(s)=0$ implies $\Im s=0$. Furthermore,
\begin{equation}
\begin{split}
&\frac{1}{2\pi\i}\int_{-\infty}^{\infty}h(\rho)\Re\frac{d}{d\rho}\log(1+\beta\psi(\tfrac{1}{2}+\i\rho))d\rho\\
=\;&\frac{1}{4\pi\i}\int_{-\infty}^{\infty}h(\rho)\frac{d}{d\rho}\log\lbrace(1+\beta\psi(\tfrac{1}{2}+\i\rho))(1+\beta\psi(\tfrac{1}{2}-\i\rho))\rbrace d\rho\\
=\;&\frac{1}{4\pi\i}\int_{-\i\delta-\infty}^{-\i\delta+\infty}h(\rho)\frac{d}{d\rho}\log\lbrace(1+\beta\psi(\tfrac{1}{2}+\i\rho))(1+\beta\psi(\tfrac{1}{2}-\i\rho))\rbrace d\rho.
\end{split}
\end{equation}
Hence the identity term in the perturbed hyperbolic trace formula is closely connected with the perturbed spherical trace via the identity \eqref{sphtrace}.

\section{The perturbed zeta function}

We define the perturbed zeta function on a quotient $\Gamma\backslash\HH$ to be
\begin{equation}\label{zeta}
Z_{\Gamma}^{\alpha,z_{0}}(s)=Z_{\Gamma}(s)S_{\alpha,z_{0}}(s)
\end{equation}
where
\begin{equation}
Z_{\Gamma}(s)=\prod_{\left\{p\right\}}\prod_{n=1}^{\infty}(1-\e^{-(n+s)l_{p}})
\end{equation}
is the Selberg zeta function associated with the group $\Gamma$ and the product representation is valid for $\Re s>1$. We have an analogous expression for $Z_{\Gamma}^{\alpha,z_{0}}(s)$.
\begin{cor}
The following idendity holds if $\Re s>\tfrac{1}{2}+\sigma(\beta)$
\begin{equation}\label{product}
Z_{\Gamma}^{\alpha,z_{0}}(s)=Z_{\Gamma}(s)(1+m\beta\psi(s))
\prod_{k=1}^{\infty}\prod_{\gamma_{1},\cdots,\gamma_{k}\in\Gamma\backslash\scrI}
\exp\left\{\frac{(-1)^{k}}{k}\left(\frac{\beta}{1+m\beta\psi(s)}\right)^{k}\prod_{j=1}^{k}G_{s}(l_{\gamma_{j},z_{0}})\right\}
\end{equation}
where the product on the RHS converges absolutely.
\end{cor}
\begin{proof}
The proof is an immediate consequence of the trace formula and follows simply by taking the logarithm of the product and using the same bound as in the proof of the trace formula to justify absolute convergence. We obtain
\begin{equation}
\begin{split}
\log{Z_{\Gamma}^{\alpha,z_{0}}(s)}-\log{Z_{\Gamma}(s)}-\log(1+m\beta\psi(s))
=\sum_{k=1}^{\infty}\frac{(-\beta)^{k}}{k}\sum_{\gamma_{1},...,\gamma_{k}\in\Gamma\backslash\scrI}\frac{\prod_{j=1}^{k}G_{s}(l_{\gamma_{j},z_{0}})}{(1+m\beta\psi(s))^{k}}
\end{split}
\end{equation}
where the sum over $k$ converges absolutely by \eqref{bigbound}.
\end{proof}
The identity \eqref{product} expresses the ratio $Z_{\Gamma}^{\alpha,z_{0}}(s)/Z_{\Gamma}(s)$ as a product of a local term $1+\beta\psi(s)$ and a product over the number of visits paid by a diffractive orbit to the scatterer multiplying orbit terms corresponding to all combinations of $k$ primitive diffractive orbits. 

In analogy with $Z_{\Gamma}(s)$ the perturbed zeta function satisfies a functional equation
\begin{equation}
Z_{\Gamma}^{\alpha,z_{0}}(s)=\xi(s)\theta_{\alpha,z_{0}}(1-s)Z_{\Gamma}^{\alpha,z_{0}}(1-s)
\end{equation}
where $Z_{\Gamma}(s)=\xi(s)Z_{\Gamma}(1-s)$.

If $\Gamma\backslash\HH$ has one cusp, Theorem \ref{prop} states that the relative zeta function $Z_{\Gamma}^{\alpha,z_{0}}(s)/Z_{\Gamma}(s)$ has simple poles corresponding to old eigenvalues $\lbrace s_{j}\rbrace_{j\geq-M}$ and simple zeros $\lbrace s^{\alpha}_{j}\rbrace_{j\geq-M}$ corresponding to new eigenvalues. These are situated on the critical line and a finite number on the real line. It also has poles and zeros corresponding to old and new resonances in the halfplane $\Re s<\tfrac{1}{2}$. It follows that $Z_{\Gamma}^{\alpha,z_{0}}(s)$ has zeros of order $m_{j}-1$ at $\lbrace s_{j}\rbrace_{j\geq-M}$, where $m_{j}$ denotes their old multiplicity with respect to the Laplacian. There is a zero or pole of order $2\tilde{m}-3$ at $s=\tfrac{1}{2}$ if $\tilde{\lambda}=\tfrac{1}{4}$ is an eigenvalue of the Laplacian and $\tilde{m}$ denotes the multiplicity of $\tilde{\lambda}$, and a pole of order $2$ otherwise. It also has simple zeros at $\lbrace s^{\alpha}_{j}\rbrace_{j\geq-M}$ corresponding to new eigenvalues as well as zeros corresponding to perturbed resonances and poles corresponding to unperturbed resonances in the halfplane $\Re s<\tfrac{1}{2}$. Note that an analogue of the Riemann hypothesis does not hold, in contrast to the case of the Laplacian, because under variation of $\alpha\in\RR\backslash\lbrace0\rbrace$ zeros of $Z_{\Gamma}^{\alpha,z_{0}}(s)$ on the critical line corresponding to new eigenvalues move continuously into the plane $\Re s<\tfrac{1}{2}$ to turn into resonances. We would like to remark that this is a characteristic of non-compactness. For compact surfaces such a phenomenon would not occur, the eigenvalues move continously along the critical line under variaton of $\alpha\in\RR\backslash\lbrace0\rbrace$, an analogue of the Riemann hypothesis still holds in this case. We also point out the analogy to the Phillips-Sarnak conjecture \cite{PhSa} on the behaviour of the eigenvalues of the Laplacian under perturbation of the metric.

\section*{Acknowledgments}
I would like to thank Jens Marklof for his guidance in the preparation of this work and many helpful and inspiring discussions. Furthermore, I would like to thank Andreas Str\"ombergsson for inviting me to Uppsala and both him and Professor Dennis Hejhal for very helpful discussions about this work during my visit. Furthermore, I would like to thank Professor Colin de Verdiere for his comments and suggestions.

\end{document}